\documentclass[11pt]{article}
\usepackage{amssymb}
\usepackage[top=1.20in,bottom=1.20in,left=1.00in,right=1.00in] {geometry}

\usepackage{fancyhdr}
\usepackage{amsmath}
\usepackage{bm}
\usepackage{esint}
\usepackage{amsfonts}
\usepackage{amsthm}
\usepackage{amssymb}
\usepackage{amsbsy}
\usepackage{verbatim}
\usepackage{graphicx}
\usepackage[caption=false,font=footnotesize]{subfig}
\usepackage{url}
\usepackage{mathrsfs}
\usepackage{color}
\usepackage{amstext}
\usepackage{stmaryrd}
\usepackage{enumerate}
\usepackage{algpseudocode}
\usepackage{algorithm}
\usepackage{pifont}
\usepackage{prettyref}

\font\msbm=msbm10

\numberwithin{equation}{section}

\theoremstyle{plain}
\newtheorem{theorem}{Theorem}[section]
\newtheorem{lemma}[theorem]{Lemma}

\newtheorem{proposition}[theorem]{Proposition}

\newtheorem{remark}[theorem]{Remark}
\def\mathbb#1{\hbox{\msbm{#1}}}

\renewcommand{\epsilon}{\varepsilon}
\newcommand{\ba}{\boldsymbol{a}}
\newcommand{\bb}{\boldsymbol{b}}

\newcommand{\bd}{\boldsymbol{d}}
\newcommand{\be}{\boldsymbol{e}}
\newcommand{\bff}{\boldsymbol{f}}
\newcommand{\bg}{\boldsymbol{g}}
\newcommand{\bh}{\boldsymbol{h}}
\newcommand{\bsm}{\boldsymbol{m}}

\newcommand{\bs}{\boldsymbol{s}}

\newcommand{\bu}{\boldsymbol{u}}
\newcommand{\bv}{\boldsymbol{v}}
\newcommand{\bw}{\boldsymbol{w}}
\newcommand{\bx}{\boldsymbol{x}}
\newcommand{\by}{\boldsymbol{y}}
\newcommand{\bz}{\boldsymbol{z}}
\newcommand{\bone}{\boldsymbol{1}}

\newcommand{\BA}{\boldsymbol{A}}
\newcommand{\BB}{\boldsymbol{B}}
\newcommand{\BC}{\boldsymbol{C}}
\newcommand{\BD}{\boldsymbol{D}}
\newcommand{\BE}{\boldsymbol{E}}
\newcommand{\BF}{\boldsymbol{F}}
\newcommand{\BG}{\boldsymbol{G}}

\newcommand{\BH}{\boldsymbol{H}}
\newcommand{\BM}{\boldsymbol{M}}
\newcommand{\BP}{\boldsymbol{P}}
\newcommand{\BQ}{\boldsymbol{Q}}
\newcommand{\BS}{\boldsymbol{S}}

\newcommand{\BU}{\boldsymbol{U}}
\newcommand{\BV}{\boldsymbol{V}}

\newcommand{\BX}{\boldsymbol{X}}
\newcommand{\BY}{\boldsymbol{Y}}
\newcommand{\BZ}{\boldsymbol{Z}}

\newcommand{\BLam}{\boldsymbol{\Lambda}}

\newcommand{\bphi}{\boldsymbol{\phi}}
\newcommand{\beps}{\boldsymbol{\eps}}

\newcommand{\tbe}{\widetilde{\boldsymbol{e}}}

\newcommand{\tbw}{\widetilde{\boldsymbol{w}}}

\newcommand{\bzero}{\boldsymbol{0}}

\newcommand{\A}{\mathcal{A}}

\newcommand{\CC}{\mathbb{C}}

\newcommand{\CZ}{\mathcal{Z}}

\newcommand{\MS}{\mathcal{S}}

\newcommand{\I}{\boldsymbol{I}}
\newcommand{\RR}{\mathbb{R}}
\newcommand{\lag}{\left\langle}
\newcommand{\rag}{\right\rangle}

\newcommand{\Tr}{\text{Tr}}
\newcommand{\eps}{\epsilon}

\newcommand{\mi}{\mathrm{i}}
\renewcommand{\Pr}{\mathbb{P} }

\DeclareMathOperator{\Var}{Var}
\DeclareMathOperator{\E}{\mathbb{E}}

\DeclareMathOperator{\diag}{diag}

\DeclareMathOperator{\SNR}{SNR}

\DeclareMathOperator{\minimize}{\text{minimize}}
\DeclareMathOperator{\rank}{\text{rank}}
\DeclareMathOperator{\Corr}{Corr}

\newcommand{\vct}[1]{\bm{#1}}

\definecolor{xl}{RGB}{200,50,50}

\begin{document}
\title{\bf Self-Calibration and Bilinear Inverse Problems \\ via Linear Least Squares\thanks{This research was supported by the NSF via Award Nr.~DTRA-DMS 1322393 and Award Nr.~DMS 1620455.}}

\author{Shuyang Ling\thanks{Courant Institute of Mathematical Sciences and the Center for Data Science, New York University, NY 10003 (Email: sling@cims.nyu.edu)}, 
Thomas Strohmer\thanks{Department of Mathematics, University of California Davis, CA 95616 (Email: strohmer@math.ucdavis.edu).}}

\maketitle

\begin{abstract}
Whenever we use devices to take measurements, calibration is indispensable. While the purpose of calibration is to reduce bias and uncertainty in the measurements, it can be quite difficult, expensive, and sometimes even impossible to implement.  We study a challenging problem called \emph{self-calibration}, i.e., the task of designing an algorithm for devices so that the algorithm is able to perform calibration automatically. More precisely, we consider the setup $\boldsymbol{y} = \mathcal{A}(\boldsymbol{d}) \boldsymbol{x} + \boldsymbol{\epsilon}$ where only partial information about the sensing matrix $\mathcal{A}(\boldsymbol{d})$ is known and  where $\mathcal{A}(\boldsymbol{d})$ linearly depends on $\boldsymbol{d}$. The goal is to estimate the calibration parameter $\boldsymbol{d}$ (resolve the uncertainty in the sensing process) and the signal/object of interests $\boldsymbol{x}$ simultaneously.  For three different models of practical relevance, we show how such a \emph{bilinear} inverse problem, including blind deconvolution as an important example, can be solved via a simple \emph{linear least squares} approach. As a consequence, the proposed algorithms are numerically extremely efficient, thus potentially allowing for real-time deployment. We also present a variation of the least squares approach, which leads to a~\emph{spectral method}, where the solution to the bilinear inverse problem can be found by computing the singular vector associated with the smallest singular value of a certain matrix derived from the bilinear system.
Explicit theoretical guarantees and stability theory are derived for both techniques; and the number of sampling complexity is nearly optimal (up to a poly-log factor). Applications in imaging sciences and signal processing are discussed and numerical simulations are presented to demonstrate the effectiveness and efficiency of our approach. 

\end{abstract}




\section{Introduction}\label{s:intro}
Calibration is ubiquitous in all fields of science and engineering. It is an essential step to guarantee that the devices measure accurately what scientists and engineers want. 
If sensor devices are not properly calibrated, their measurements are likely of little use to the application. While calibration is mostly done by specialists, it often can be expensive, time-consuming and sometimes even impossible to do in practice. Hence, one may wonder whether it is possible to enable machines to calibrate themselves automatically with a smart algorithm and give the desired measurements. This leads to the challenging field of \emph{self-calibration} (or \emph{blind calibration}). It has a long history in imaging sciences, such as camera self-calibration~\cite{PKV99,ISV14}, 
blind image deconvolution~\cite{campisi2007blind}, self-calibration in medical imaging~\cite{shin2014calibrationless}, and the well-known phase retrieval problem (phase calibration)~\cite{Fienup78}. It also plays an important role in signal processing~\cite{FS14} and wireless communications~\cite{Tong95,mrydlsb12}.

Self-calibration is not only a challenging problem for engineers, but also for mathematicians. It means that one needs to estimate the calibration parameter of the devices to adjust the measurements as well as recover the signal of interests. 
More precisely, many self-calibration problems are expressed in the following mathematical form,
\begin{equation}\label{eq:model}
\by = \A(\bd)\bx + \beps,
\end{equation}
where $\by$ is the observation, $\A(\bd)$ is a partially unknown sensing matrix, which depends on an unknown parameter $\bd$ and $\bx$ is the desired signal.
An uncalibrated sensor/device directly corresponds to ``\textit{imperfect sensing}", i.e., uncertainty exists within the sensing procedure and we do not know everything about $\A(\bd)$ due to the lack of calibration. The purpose of self-calibration is to resolve the uncertainty i.e., to estimate $\bd$ in $\A(\bd)$ and to recover the signal $\bx$ at the same time. 

The general model~\eqref{eq:model} is too hard to get meaningful solutions without any further assumption since there are many variants of the general model under different settings. 
In~\eqref{eq:model}, $\A(\bd)$ may depend on $\bd$ in a nonlinear way, e.g., $\bd$ can be the unknown orientation of a protein molecule and $\bx$ is the desired object~\cite{wang2013orientation};  in  phase retrieval, $\bd$ is the unknown phase information of the Fourier transform of the object~\cite{Fienup78}; in direction-of-arrival estimation $\bd$ represents unknown offset, gain, and phase of the sensors~\cite{WF90}.
Hence, it is impossible to resolve every issue in this field, but we want to understand several scenarios of self-calibration which have great potential in real world applications. Among all the cases of interest, we assume that $\A(\bd)$ \emph{linearly} depends on the unknown $\bd$  and will explore three different types of self-calibration models that are of considerable practical relevance.
However, even for  linear dependence, the problem is already quite challenging, since in fact we are dealing with \emph{bilinear (nonlinear) inverse problems}. All those three models have wide applications in imaging sciences, signal processing, wireless communications, etc., which will be addressed later. Common to these applications is the desire or need for {\em fast} algorithms, which ideally should be accompanied by theoretical performance guarantees. 
We will show under certain cases,  these bilinear  problems can be solved by \emph{linear least squares} exactly and efficiently if no noise exists, which is guaranteed by rigorous mathematical proofs.  Moreover, we prove that the solution is also robust to noise with tools from random matrix theory. Furthermore, we show that a variation of our approach leads to a spectral method, where the solution to the bilinear problem can be found by computing the singular vector associated with the smallest singular matrix of a certain matrix derived from the bilinear system.

\subsection{State of the art}

By assuming that $\A(\bd)$ linearly depends on $\bd$,~\eqref{eq:model} becomes a bilinear inverse problem, i.e., we want to estimate $\bd$ and $\bx$ from $\by$, where $\by$ is the output of a bilinear map from $(\bd, \bx).$ Bilinear inverse problems, due to its importance, are getting more and more attentions over the last few years. On the other hand, they are also notoriously difficult  to solve in general. 
Bilinear inverse problems are closely related to low-rank matrix recovery, see~\cite{DR16} for a comprehensive review. There exists extensive literature on this topic and it is not possible do justice to all these contributions. Instead we will only highlight some of the works which have inspired us.

Blind deconvolution might be one of the most important examples of bilinear inverse problems~\cite{campisi2007blind}, i.e., recovering $\bff$ and $\bg$ from $\by = \bff\ast \bg$, where ``$\ast$" stands for convolution. 
If both $\bff$ and $\bg$ are inside known low-dimensional subspaces, the blind deconvolution can be rewritten as $\mathcal{F}(\by) = \diag(\BB\bd)\BA\bx$, where $\mathcal{F}(\bff) = \BB\bd$, $\mathcal{F}(\bg) = \BA\bx$ and ``$\mathcal{F}$" denotes the Fourier transform. 
In the inspiring work~\cite{ARR12}, Ahmed, Romberg and Recht  apply the ``lifting" techniques~\cite{CESV15} and convert the problem into estimation of the rank-one matrix $\bd\bx^*$. 
It is shown that solving a convex relaxation  enables recovery of $\bd\bx^*$ under certain choices of $\BB$ and $\BA$. Following a similar spirit,~\cite{LS15} uses ``lifting" combined with a convex approach to solve the scenarios with sparse $\bx$ and~\cite{LS15Blind} studies the so called ``blind deconvolution and blind demixing" problem. 
The other line of blind deconvolution follows a nonconvex optimization approach~\cite{AKR16,LLSW16,LeeLJB17}. In~\cite{AKR16}, Ahmed, Romberg and Krahmer,  using tools from generic chaining, obtain local convergence of a sparse power factorization algorithm to solve this blind deconvolution problem when $\bh$ and $\bx$ are sparse and $\BB$ and $\BA$ are Gaussian random matrices. Under the same setting as~\cite{AKR16}, Lee et al.~\cite{LeeLJB17} propose a projected gradient descent algorithm based on matrix factorizations and provide a convergence analysis to recover sparse signals from subsampled convolution. However, this projection step can be hard to implement. As an alternative, the expensive projection step is replaced by a heuristic approximate projection, but then the global convergence is not fully guaranteed. Both~\cite{AKR16,LeeLJB17} achieve  nearly optimal sampling complexity.~\cite{LLSW16} proves global convergence of a gradient descent type algorithm when $\BB$ is a deterministic Fourier type matrix and $\BA$ is Gaussian. Results about identifiability issue of bilinear inverse problems can be found in~\cite{LiLB16b,KecK17,LiLB16}.

Another example of self-calibration focuses on the setup $\by_l = \BD\BA_l\bx$, where $\BD = \diag(\bd)$. The difference from the previous model consists in replacing the subspace assumption by multiple measurements. There are two main applications of this model. One application deals with blind deconvolution in an imaging system which uses randomly coded masks~\cite{BR15,TB14}. The measurements are obtained by (subsampled) convolution of an unknown blurring function $\BD$ with several random binary modulations of one image. Both~\cite{BR15} and~\cite{AD16} developed convex relaxing approaches (nuclear norm minimization) to achieve exact recovery of the signals and the blurring function. 
The other application is concerned with calibration of the unknown gains and phases $\BD$ and recovery of the signal $\bx$, see e.g.~\cite{CamJ16,CJ16sparse}. Cambareri and Jacques propose a gradient descent type algorithm in~\cite{CamJ16,CJ16b} and show convergence of the iterates by first constructing a proper initial guess. An empirical study is given in~\cite{CJ16sparse} when $\bx$ is sparse  by applying an alternating hard thresholding algorithm. 
Recently,~\cite{ACD15,AD16} study the blind deconvolution when inputs are changing. More precisely, the authors consider $\by_l = \bff \ast \bg_l$ where each $\bg_l$ belongs to a different known subspace, i.e., $\bg_l = \BC_l\bx_l$. They employ a similar convex approach as in~\cite{ARR12} to achieve exact recovery with number of measurements close to the information theoretic limit. 

An even more difficult, and from a practical viewpoint highly relevant, scenario focuses on \emph{self-calibration from multiple snapshots}~\cite{WF90}. Here, one wishes to recover the unknown gains/phases $\BD = \diag(\bd)$ and a signal matrix $\BX = [\bx_1,\cdots,\bx_p]$ from $\BY = \BD\BA\BX$. For this model, the sensing matrix $\BA$ is fixed throughout the sensing process and one measures output under different snapshots $\{\bx_l\}_{l=1}^p$. One wants to understand under what conditions we can identify  $\BD$ and $\{\bx_l\}_{l=1}^p$ jointly. If $\BA$ is a Fourier type matrix, this model has applications in both image restoration from multiple filters~\cite{GY99} and also network calibration~\cite{BN07,LB14}. 
We especially benefitted from work by Gribonval and coauthors~\cite{GCD12,BPGD13}, as well as by Balzano and Nowak~\cite{BN07,balzano2008blind}.
The papers~\cite{BN07,balzano2008blind} study the noiseless version of the problem by solving a linear system and~\cite{LB14} takes a total least squares approach in order to obtain empirically  robust recovery in the presence of noise. If each $\bx_l$ is sparse, this model becomes more difficult and Gribonval et al.~\cite{BPGD13,GCD12} give a thorough numerical study. Very recently,~\cite{WC16} gave a theoretic result under certain conditions. This calibration problem is viewed as a special case of the dictionary learning problem where the underlying dictionary $\BD\BA$ possesses some additional structure.  The idea of transforming a blind deconvolution problem into a linear problem can also be found in~\cite{mrydlsb12}, where the authors analyze a certain 
non-coherent wireless communication scenario.
 
\subsection{Our contributions}

In our work, we consider three different models of self-calibration, namely, $\by_l = \BD\BA_l\bx$, $\by_l = \BD\BA_l\bx_l$ and $\by_l = \BD\BA\bx_l$. Detailed descriptions of these models are given in the next Section. We do not impose any sparsity constraints on $\bx$ or $\bx_l$. We want to find out $\bx_l$ (or $\bx$) and $\BD$ when $\by_l$ (or $\by$) and $\BA_l$ (or $\BA$) are given. Roughly, they correspond to the models in~\cite{BR15,ACD15,BPGD13} respectively. Though all of the three models belong to the class of bilinear inverse problems, we will prove that simply solving \emph{linear least squares} will give solutions to all those models exactly and robustly for invertible $\BD$ and for several useful choices of $\BA$ and $\BA_l.$ Moreover, the sampling complexity is nearly optimal (up to poly-log factors) with respect to the information theoretic limit (degree of freedom of unknowns). 

As mentioned before, our approach is largely inspired by~\cite{BPGD13} and~\cite{BN07,balzano2008blind}; there the authors convert a bilinear inverse problem into a linear problem via a proper transformation. 
We follow a similar approach in our paper. The paper~\cite{BPGD13}
provides an extensive empirical study, but no theoretical analysis. Nowak and Balzano, in~\cite{BN07,balzano2008blind} provide numerical simulations as well
as theoretical conditions on the number of measurements required to solve the noiseless case. 

Our paper goes an important step further: On the one hand we consider
more general self-calibration settings. And on the other hand we provide a rigorous theoretical analysis for recoverability and, perhaps most importantly, stability theory 
in the presence of measurement errors. Owing to the simplicity of our approach and the structural properties of the underlying matrices, our framework yields  self-calibration algorithms that are numerically extremely efficient, thus potentially allowing for deployment in applications where real-time self-calibration is needed.

\subsection{Notation and Outline}
We introduce notation which will be used throughout the paper. Matrices are denoted in boldface or a calligraphic font such as
$\BZ$ and $\mathcal{Z}$; vectors are denoted by boldface lower case letters, e.g.~$\bz.$
The individual entries of a matrix or a vector are denoted in normal font such as $Z_{ij}$ or
$z_i.$
For any matrix $\BZ$, $\|\BZ\|$
denotes its operator norm, i.e., the largest singular value, and $\|\BZ\|_F$ denotes its the Frobenius norm, i.e.,
$\|\BZ\|_F =\sqrt{\sum_{ij} |Z_{ij}|^2 }$. For any vector $\bz$, $\|\bz\|$ denotes its Euclidean norm. For both
matrices and vectors, $\BZ^T$ and $\bz^T$ stand for the transpose of $\BZ$ and $\bz$ respectively while $\BZ^*$ and
$\bz^*$ denote their complex conjugate transpose.  For any real number $z$, we let $z_+ = \frac{1}{2}(z + |z|).$ We equip the matrix space $\CC^{K\times N}$ with the inner
product defined by $\lag \BU, \BV\rag : =\Tr(\BU^*\BV).$ A special case is the inner product of two vectors, i.e.,
$\lag \bu, \bv\rag = \Tr(\bu^*\bv) = \bu^*\bv.$ We define the correlation between two vectors $\bu$ and $\bv$ as $\Corr(\bu,\bv) = \frac{\bu^*\bv}{\|\bu\|\|\bv\|}$. For a given vector $\bv$, $\diag(\bv)$ represents the diagonal
matrix whose diagonal entries are given by the vector $\bv$.

$C$ is an absolute constant and $C_{\gamma}$ is a constant which depends linearly on $\gamma$, but on no other parameters. $\I_n$ and $\bone_n$ always denote the $n\times n$ identity matrix and a column vector of ``$1$" in $\RR^n$ respectively. And $\{\be_i\}_{i=1}^m$ and $\{\tbe_l\}_{l=1}^p$ stand for the standard orthonormal basis in $\RR^m$ and $\RR^p$ respectively. ``$\ast$" is the circular convolution and ``$\otimes$" is the Kronecker product.

\medskip
The paper is organized as follows. The more detailed discussion of  the models under consideration and the proposed method will be given in 
Section~\ref{s:model}. Section~\ref{s:main} presents the main results of our paper and we will give numerical simulations in Section~\ref{s:numerics}. Section~\ref{s:proof} contains the proof  for each scenario. We collect some useful auxiliary results in the Appendix.

\section{Problem setup: Three self-calibration models}\label{s:model}

This section is devoted to describing three different models for self-calibration in detail. We will also explain how those bilinear inverse problems are reformulated and solved via linear least squares.

\subsection{Three special models of self-calibration}
\paragraph{Self-calibration via repeated measurements}
Suppose we are seeking for information with respect to an unknown signal $\bx_0$ with several randomized linear sensing designs. Throughout this procedure, the calibration parameter $\BD$ remains the same for each sensing procedure. How can we recover the signal $\bx_0$ and $\BD$ simultaneously? Let us make it more concrete by introducing the following model,
\begin{equation}\label{eq:model1}
\vct{y}_l = \vct{D}\vct{A}_l \bx_0 + \beps_l, \quad 1\leq l\leq p
\end{equation}
where $\BD = \diag(\bd)\in\CC^{m\times m}$ is a diagonal matrix and each $\BA_l\in\CC^{m\times n}$ is a measurement matrix. Here $\by_l$ and $\BA_l$ are given while $\BD$ and $\bx_0$ are unknown. For simplicity, we refer to the setup~\eqref{eq:model1} as ``\emph{self-calibration from repeated measurements}". This model has various applications in self-calibration for imaging systems~\cite{CamJ16,CJ16b}, networks~\cite{BN07}, as well as in blind deconvolution from 
random masks~\cite{BR15,TB14}.

\paragraph{Blind deconvolution via diverse inputs}
Suppose that one sends several different signals through the same unknown channel, and each signal is encoded differently. Namely, we are considering
\begin{equation*}
\by_l = \bff \ast \BC_l \bx_l + \beps_l, \quad 1\leq l\leq p.
\end{equation*}
How can one estimate the channel and each signal jointly? 
In the frequency domain, this ``\emph{blind deconvolution via diverse inputs}"~\cite{ACD15,AD16} problem can be written as (with a bit abuse of notation),
\begin{equation}\label{eq:model2}
\by_l = \BD\BA_l \bx_l + \beps_l, \quad 1\leq l\leq p
\end{equation} 
where $\BD = \diag(\bd)\in\CC^{m\times m}$ and $\BA_l\in\CC^{m\times n}$ are the Fourier transform of $\bff$ and $\BC_l$ respectively.  
We aim to recover $\{\bx_l\}_{l=1}^p$ and $\BD$ from $\{\by_l,\BA_l\}_{l=1}^p.$

\paragraph{Self-calibration from multiple snapshots} Suppose we take measurements of several signals $\{\bx_l\}_{l=1}^p$ with the same set of design matrix $\BD\BA$ (i.e., each sensor corresponds one row of $\BA$ and has an unknown complex-valued calibration term $d_i$). When and how can we recover $\BD$ and $\{\bx_l\}_{l=1}^p$ simultaneously?
More precisely, we consider the following model of \emph{self-calibration model from multiple snapshots}:
\begin{equation}\label{eq:model3}
\by_l = \BD\BA \bx_l + \beps_l, \quad 1\leq l\leq p.
\end{equation} 
Here $\BD = \diag(\bd)$ is an unknown diagonal matrix, $\BA\in \CC^{m\times n}$ is a sensing matrix, $\{\bx_l\}_{l=1}^p$ are $n\times 1$ unknown signals and $\{\by_l\}_{l=1}^p$ are their corresponding observations. This multiple snapshots model has been used in image restoration from multiple filters~\cite{GY99} and self-calibration model for sensors~\cite{WF90,BN07,LB14,GCD12,BPGD13}.

\subsection{Linear least squares approach}

Throughout our discussion, we assume that $\BD$ is {\em invertible},
and we let $\BS := \diag(\bs) = \BD^{-1}$. Here, $\BD = \diag(\bd)$ stands for the  calibration factors of the sensor(s)~\cite{BN07,BPGD13} and hence it is reasonable to assume invertibility of $\BD$. For, if a sensor's gain were equal to zero, then it would not contribute any measurements to the observable
$\by$, in which case the associated entry of $\by$ would be zero. But then we could simply discard that entry and consider the  correspondingly reduced
system of equations, for which the associated $\BD$ is now invertible.

One simple solution is to minimize a \emph{nonlinear} least squares objective function. Let us take~\eqref{eq:model1} as an example (the others~\eqref{eq:model2} and~\eqref{eq:model3} have quite similar formulations), 
\begin{equation}\label{eq:nlls}
\min_{\BD, \bx} \sum_{l=1}^p\| \BD\BA_l \bx - \by_l \|^2.
\end{equation}
The obvious difficulty lies in the \emph{biconvexity} of~\eqref{eq:nlls}, i.e., if either $\BD$ or $\bx$ is fixed, minimizing over the other variable is a convex program. In general, there is no way to guarantee that any gradient descent algorithm/alternating minimization will give the global minimum. However, for the three models described above, there is one shortcut towards the exact and robust recovery of the solution via \emph{linear} least squares if $\BD$ is invertible.

We continue with~\eqref{eq:model1} when $\beps_l = \bzero$, i.e.,
\begin{equation}\label{eq:model1-linear}
\diag(\by_l)\bs = \vct{A}_l \bx, \quad 1\leq l\leq p
\end{equation}
where $\BS\by_l = \diag(\by_l)\bs$ with $s_i = d_i^{-1}$ and $\BS$ is defined as $\diag(\bs).$
The original measurement equation turns out to be a \emph{linear} system with unknown $\bs$ and $\bx$. The same idea of \emph{linearization} can be also found in~\cite{GCD12,BPGD13,WC16,BN07}. In this way, the ground truth $\bz_0 : = (\bs_0, \bx_0)$ lies actually inside the \emph{null space} of this linear system.

Two issues arise immediately: One the one hand, we need to make sure that $(\bs_0, \bx_0)$ spans the whole null space of this linear system. This is equivalent to the identifiability issue of bilinear problems of the form~\eqref{eq:nlls}, because if the pair $(\alpha^{-1}\bd_0, \alpha\bx_0)$ for some $\alpha \neq 0$ is (up to the scalar $\alpha$) unique solution to~\eqref{eq:model1}, then $(\alpha\bs_0,  \alpha\bx_0)$ spans the null space of~\eqref{eq:model1-linear}, see also~\cite{LiLB16b,KecK17,LiLB16}. On the other hand, we also need to avoid the trivial scenario 
$(\bs_0, \bx_0) = (\bzero, \bzero)$, since it has no physical meaning. To resolve the latter issue, we add the extra linear constraint (see also~\cite{BPGD13,BN07})
\begin{equation}
\label{wconstraint}
\left\langle \bw , \begin{bmatrix} \bs \\ \bx \end{bmatrix} \right\rangle = c,
\end{equation}
where the scalar $c$ can be any nonzero number (we note that $\bw$ should of course not be orthogonal to the solution).
Therefore, we hope  that in the noiseless case it suffices to solve the following linear system to recover $(\bd, \bx_0)$ up to a scalar, i.e., 
\begin{equation}
\label{eq:linear}
\underbrace{
\begin{bmatrix}
\diag(\by_1) & -\BA_1 \\
\vdots & \vdots \\
\diag(\by_p) & -\BA_p \\
 \multicolumn{2}{c}{\bw^*}
\end{bmatrix}
_{(mp+1) \times (m + n)}
}_{\A_{\bw}}
\underbrace{
\begin{bmatrix}
\bs \\
\bx \\
\end{bmatrix}_{(m+n)\times 1}}_{\bz}
= \underbrace{
\begin{bmatrix}
\bzero \\
c \\
\end{bmatrix}_{(mp+1)\times 1}}_{\bb}
\end{equation}
In the presence of additive noise, 
we replace the linear system above by a linear least squares problem
\begin{equation*}
\minimize \, \sum_{l=1}^p\|\diag(\by_l)\bs - \BA_l \bx\|^2 + | \bw^*\bz - c|^2
\end{equation*}
with respect to $\bs$ and $\bx$, or equivalently, 
\begin{equation}\label{eq:obj-ls}
\underset{\bz}{\minimize} \,\, \|\A_{\bw} \bz - \bb \|^2
\end{equation}
where $\bz = \begin{bmatrix}
\bs \\
\bx \\
\end{bmatrix}$
, $\bb = \begin{bmatrix}
\bzero \\
c \\
\end{bmatrix}$, and $\A_{\bw}$ is the matrix on the left hand side of~\eqref{eq:linear}. Following  the same idea,~\eqref{eq:model2} and~\eqref{eq:model3} can also be reformulated into linear systems and be solved via linear least squares. The matrix $\A_{\bw}$ and the vector $\bz$ take a slightly different form for those cases, see~\eqref{eq:A0-2-ms} 
and~\eqref{eq:Amodel3}, respectively.

\medskip
\noindent

\begin{remark} 
Note that solving \eqref{eq:obj-ls} may not be the optimal choice to recover the unknowns from the perspective of statistics since the noisy perturbation actually enters  into $\A_{\bw}$ instead of $\bb$. More precisely,  the noisy perturbation $\delta \A$ to the left hand side of the corresponding linear system for~\eqref{eq:model1},~\eqref{eq:model2} and~\eqref{eq:model3}, is always in the form of 
\begin{equation}\label{def:deltaA}
\delta \A := \begin{bmatrix}
\diag(\beps_1) & \bzero\\
\vdots & \vdots \\
\diag(\beps_p) & \bzero\\
\bzero & \bzero
\end{bmatrix}.
\end{equation}
The size of $\delta \A$ depends on the models.
Hence total least squares~\cite{LB14} could be a better alternative while it is more difficult to analyze and significantly more costly to compute. Since computational efficiency 
is essential for many practical applications, a straightforward  implementation of total least squares is of limited use. Instead one should keep in mind that the actual perturbation enters only into $\diag(\by_l)$, while the other matrix blocks remain unperturbed. Constructing a total least squares solution that obeys these constraints, doing so in a numerically efficient manner and providing theoretical error bounds for it, is a  rather challenging task, which we plan to address in our future work.
\end{remark}

\begin{remark} 
Numerical simulations imply that the performance under noisy measurements depends on the choice of $\bw$, especially how much $\bw$ and $\bz_0$ are correlated. One extreme case is that $\lag \bw, \bz_0\rag = 0$, in which case we cannot avoid the solution $\bz = \bzero$. It might be better to add a constraint like $\|\bw\| = 1$. However, this will lead to a nonlinear problem which may not be solved efficiently and not come with rigorous recovery guarantees. Therefore, we present an alternative approach in the next subsection.
\end{remark}

\subsection{Spectral method}\label{ss:svd}
In this subsection, we discuss a method for solving the self-calibration problem, whose performance does not depend on the choice of $\bw$ as it avoids the need of $\bw$ in the first place. Let $\MS$ be the matrix $\A_{\bw}$ excluding the last row (the one which contains $\bw$). We decompose $\MS$  into $\MS = \MS_0 + \delta\MS$ where $\MS_0$ is the noiseless part of $\MS$ and $\delta\MS$ is the noisy part\footnote{This is a slight abuse of notation, since the $\delta \A$ defined in~\eqref{def:deltaA} has an additional row with zeros. However, it will be clear from the context which $\delta \A$ we refer to, and more importantly, in our estimates we mainly care about $\|\delta \A\|$
which coincides for both choices of $\delta \A$.}.

We start with the noise-free scenario: if $\delta \MS = 0$, there holds $\MS = \MS_0$ and the right singular vector of $\MS$ corresponding to the smallest singular value is actually $\bz_0.$ Therefore we can recover $\bz_0$ by solving the following optimization problem:
\begin{equation}\label{prog:svdmin}
\min_{ \|\bz\| = 1} \| \MS\bz \|.
\end{equation}
Obviously, its solution is equivalent to the smallest singular value of $\MS$ and its corresponding singular vector. 

If noise exists, the performance will depend on how large the second smallest singular value of $\MS_0$ is and on the amount of noise, given by $\|\delta\MS\|$. 
We will discuss the corresponding theory and algorithms in Section~\ref{ss:svd-thm}, and the proof in Section~\ref{ss:svd-bilinear}.

\section{Theoretical results}\label{s:main}
We present our theoretical findings for the three models~\eqref{eq:model1},~\eqref{eq:model2} and~\eqref{eq:model3} respectively for different choices of $\BA_l$ or $\BA.$  In one of our choices the $\BA_l$ are Gaussian random matrices. The rationale for this choice is that, while a Gaussian random matrix is
 not useful or feasible in most applications, it often provides a benchmark for theoretical guarantees and numerical performance. Our other choices for the sensing matrices are structured random matrices, such as e.g.~the product of a deterministic partial (or a randomly subsampled) Fourier matrix or a Hadamard matrix\footnote{A randomly subsampled Fourier matrix is one, where we randomly choose a certain number of rows or columns of the Discrete Fourier Transform matrix, and analogously for a randomly subsampled Hadamard matrix} with a diagonal binary random matrix\footnote{At this point, we are not able to prove competitive results for fully deterministic sensing matrices.}. These matrices bring us closer to what
we encounter in real world applications. Indeed, structured random matrices of this type have been deployed for instance in imaging and 
wireless communications, see e.g.~\cite{gan2008fast,verdu1999spectral}. 

By solving simple variations (for different models) of~\eqref{eq:obj-ls}, we can guarantee that the ground truth is recovered exactly up to a scalar if no noise exists and robustly if noise is present. 
The number of measurements required for exact and robust recovery is nearly optimal, i.e., close to the information-theoretic limit up to a poly-log factor. However, the error bound for robust recovery is not optimal. 
It is worth mentioning that once the signals and calibration parameter $\BD$ are identifiable, we are able to recover both of them exactly in absence of noise by simply solving a linear system. However, identifiability alone cannot guarantee robustness. 

Throughout this section, we let $d_{\max} : = \max_{1\leq i\leq m}|d_{i,0}|$ and $d_{\min} : = \min_{1\leq i\leq m}|d_{i,0}|$ where $\{d_{i,0}\}_{i=1}^m$ are the entries of the ground truth $\bd_0.$ We also define $\A_{\bw,0}$ as the noiseless part of $\A_{\bw}$ for each individual model. 
\subsection{Self-calibration via repeated measurements}
\label{s:model1-setup}
For model~\eqref{eq:model1} we will focus on three cases:
\begin{enumerate}[(a)]
\item  $\BA_l$ is an $m\times n$ complex Gaussian random matrix, i.e., each entry in $\BA_l$ is given by $\frac{1}{\sqrt{2}}\mathcal{N}(0, 1) + \frac{\mi}{\sqrt{2}}\mathcal{N}(0, 1)$.
\item  $\BA_l$ is an $m\times n$ ``tall" random DFT/Hadamard matrix with $m\geq n$, i.e., $\BA_l : = \BH\BM_l$ where  $\BH$ consists of the first $n$ columns of an $m\times m$  DFT/Hadamard matrix and each $\BM_l : = \diag(\bsm_l)$ is a diagonal matrix with entries taking on the value  $\pm 1$ with equal probability. In particular, there holds,
\begin{equation*}
\BA_l^* \BA_l = \BM_l^*\BH^*\BH\BM_l = m\I_n.
\end{equation*}
\item $\BA_l$ is an $m\times n$ ``fat" random partial DFT matrix with $m < n$, i.e., $\BA_l : = \BH\BM_l$ where  $\BH$ consists of $m$ columns of an $n\times n$  DFT/Hadamard matrix and each $\BM_l : = \diag(\bsm_l)$ is a diagonal matrix, which is defined the same as case (b), 
\begin{equation*}
\BA_l\BA_l^* = \BH\BH^* = n\I_m.
\end{equation*}
\end{enumerate}
Our main findings are summarized as follows:
\begin{theorem}\label{thm:main1}
Consider the self-calibration model given in~\eqref{eq:model1}, where $\A_{\bw}$ is as in~\eqref{eq:linear} and $\A_{\bw,0}$ is the noiseless part of $\A_{\bw}$. Then, for the solution  $\hat{\bz}$ of~\eqref{eq:obj-ls} and $\alpha = \frac{c}{\bw^*\bz_0}$, there holds
\begin{equation*}
\frac{ \| \hat{\bz} - \alpha\bz_0 \| }{\|\alpha\bz_0\|} \leq \kappa(\A_{\bw, 0})\eta \left( 1 + \frac{2}{1 - \kappa(\A_{\bw, 0})\eta }\right)
\end{equation*}
if $\kappa(\A_{\bw,0})\eta < 1$ where $\eta = \frac{2\|\delta \A\|}{\sqrt{mp}}$. The condition number of $\A_{\bw}$ satisfies
\begin{equation*}
\kappa(\A_{\bw,0}) \leq  \sqrt{\frac{6 (mp + \|\bw\|^2)}{\min\{mp,\|\bw\|^2 |\Corr(\bw,\bz_0)|^2\}} \frac{ \max\{d_{\max}^2\|\bx\|^2, m\}}{\min\{d_{\min}^2\|\bx\|^2, m\}}}, 
\end{equation*}
where $\Corr(\bw,\bz_0) = \frac{\bw^*\bz_0}{\|\bw\|\|\bz_0\|}$, and for $\|\bw\| = \sqrt{mp}$ 
\begin{equation*}
\kappa(\A_{\bw,0}) \leq \frac{2\sqrt{3}}{ |\Corr(\bw,\bz_0)|}  \sqrt{\frac{ \max\{d_{\max}^2\|\bx\|^2, m\}}{\min\{d_{\min}^2\|\bx\|^2, m\}}}, 
\end{equation*}

\begin{enumerate}[(a)]
\item with probability $1 - (m + n)^{-\gamma}$ if $\BA_l$ is Gaussian and $p \geq c_0 \gamma \max\left\{ 1, \frac{n}{m} \right\}  \log^2(m + n)$;
\item with probability $1 - (m + n)^{-\gamma} - 2(mp)^{-\gamma + 1}$ if each $\BA_l$ is a ``tall" $(m\times n, m \geq n)$ random Hadamard/DFT matrix and $p \geq c_0 \gamma^2 \log(m + n)\log(mp)$;
\item with probability $1 - (m + n)^{-\gamma} - 2(mp)^{-\gamma + 1}$ if each $\BA_l$ is a ``fat" $(m\times n, m\leq n)$ random Hadamard/DFT matrix and $mp \geq c_0 \gamma^2n \log(m + n)\log(mp)$.
\end{enumerate}

\end{theorem}

\begin{remark}
Our result is nearly optimal in terms of required number of measurements, because the number of constraints $mp$ is required to be slightly greater than $n +m$, the number of unknowns. 
Theorem~\ref{thm:main1} can be regarded as a generalized result of~\cite{CamJ16,CJ16b}, in which $\BD$ is assumed to be positive and $\BA_l$ is Gaussian. In our result, we only need $\BD$ to be an invertible complex diagonal matrix and $\BA_l$ can be a Gaussian or random Fourier type matrix. The approaches are quite different, i.e.,~\cite{CamJ16} essentially uses nonconvex optimization by first constructing a good initial guess and then applying gradient descent to recover $\BD$ and $\bx$. Our result also provides a provable fast alternative algorithm to ``the blind deconvolution via random masks" in~\cite{BR15,TB14} where a SDP-based approach is proposed.
\end{remark}


\subsection{Blind deconvolution via diverse inputs}

We now analyze model~\eqref{eq:model2}. Following similar steps that led us from~\eqref{eq:model1} to~\eqref{eq:linear}, it is easy to see that the linear system associated 
with~\eqref{eq:model2} is given by
\begin{equation}\label{eq:A0-2-ms}
\underbrace{
\begin{bmatrix}
\diag(\by_1) & -\BA_1 & \bzero & \cdots & \bzero \\
\diag(\by_2) & \bzero & -\BA_2 & \cdots & \bzero \\
\vdots & \vdots & \vdots & \ddots &  \vdots \\
\diag(\by_p) & \bzero & \bzero & \cdots &  -\BA_p  \\
 \multicolumn{5}{c}{\bw^*}
\end{bmatrix}
_{(mp+1) \times (np + m)}}_{\A_{\bw}}
\underbrace{
\begin{bmatrix}
\bs \\
\bx_1 \\
\bx_2 \\
\vdots \\
\bx_p
\end{bmatrix}_{(m+n)\times 1}
}_{\bz}
= 
\begin{bmatrix}
\bzero \\
c \\
\end{bmatrix}_{(mp+1)\times 1}.
\end{equation}
\label{s:model2-setup}
We consider two scenarios:
\begin{enumerate}[(a)]
\item  $\BA_l$ is an $m\times n$ complex Gaussian random matrix, i.e., each entry in $\BA_l$ yields $\frac{1}{\sqrt{2}}\mathcal{N}(0, 1) + \frac{\mi}{\sqrt{2}}\mathcal{N}(0, 1)$.
\item $\BA_l$ is of the form
\begin{equation}\label{eq:model2-Al}
\BA_l = \BH_l\BM_l,
\end{equation}
where $\BH_l\in \CC^{m\times n}$ is a random partial Hadamard/Fourier matrix, i.e., the columns of $\BH_l$ are uniformly sampled without replacement from an $m\times m$ DFT/Hadamard matrix; $\BM_l := \diag(\bsm_l) =  \diag(m_{l,1}, \cdots, m_{l,n})$ is a diagonal matrix with $\{m_{l,i}\}_{i=1}^n$ being i.i.d. Bernoulli random variables.
\end{enumerate}

\begin{theorem}\label{thm:main2}
Consider the self-calibration model given in~\eqref{eq:model2}, where $\A_{\bw}$ is as in~\eqref{eq:A0-2-ms}. Let $x_{\min} : = \min_{1\leq l\leq p} \|\bx_l\|$ and $x_{\max} : = \max_{1\leq l\leq p} \|\bx_l\|$. Then, for the solution  $\hat{\bz}$ of~\eqref{eq:obj-ls} and $\alpha = \frac{c}{\bw^*\bz_0}$, there holds
\begin{equation*}
\frac{ \| \hat{\bz} - \alpha\bz_0 \| }{\|\alpha\bz_0\|} \leq \kappa(\A_{\bw, 0})\eta \left( 1 + \frac{2}{1 - \kappa(\A_{\bw, 0})\eta }\right)
\end{equation*}
if $\kappa(\A_{\bw,0})\eta < 1$ where $\eta = \frac{2\|\delta \A\|}{\sqrt{m}}$. The condition number of $\A_{\bw,0}$ obeys
\begin{equation*}
\kappa(\A_{\bw,0}) \leq  \sqrt{\frac{6 x_{\max}^2(m + \|\bw\|^2)}{x^2_{\min}\min\{ m, \|\bw\|^2|\Corr(\bw,\bz_0)|^2 \}}  \frac{ \max\{ pd_{\max}^2, \frac{m}{x^2_{\min}} \} }{ \min\{ pd_{\min}^2, \frac{m}{x^2_{\max}} \}}},
\end{equation*}
and for $\|\bw\| = \sqrt{m}$, there holds
\begin{equation*}
\kappa(\A_{\bw,0}) \leq  \frac{2\sqrt{3} x_{\max}}{x_{\min}|\Corr(\bw,\bz_0)|} \sqrt{  \frac{ \max\{ pd_{\max}^2, \frac{m}{x^2_{\min}} \} }{ \min\{ pd_{\min}^2, \frac{m}{x^2_{\max}} \}}},
\end{equation*}

\begin{enumerate}[(a)]
\item with probability at least $1 - (np +m)^{-\gamma}$ if $\BA_l$ is an $m\times n$ $(m > n)$ complex Gaussian random matrix and 
\begin{equation*}
C_0\left( \frac{1}{p} + \frac{n}{m}\right)(\gamma + 1)\log^2(np + m)  \leq \frac{1}{4}.
\end{equation*}
\item with probability at least $1 - (np + m)^{-\gamma} $ if $\BA_l$ yields~\eqref{eq:model2-Al} and
\begin{equation*}
C_0\left(\frac{1}{p} + \frac{n-1}{m-1}\right)\gamma^3 \log^4(np + m) \leq \frac{1}{4}, \quad m \geq 2.
\end{equation*}

\end{enumerate}

\end{theorem}
\begin{remark}
Note that if $\delta \A =\bzero,$ i.e., in the noiseless case, we have $\bz = \alpha\bz_0$ if $mp \geq (np + m)\text{\em poly}(\log(np+m)).$ Here $mp$ is the number of constraints and $np+m$ is the degree of freedom. Therefore, our result is nearly optimal in terms of information theoretic limit. Compared with a similar setup in~\cite{ACD15}, we have a more efficient algorithm since~\cite{ACD15} uses nuclear norm minimization to achieve exact recovery. However, the assumptions are slightly different, i.e., we assume that $\BD$ is invertible and hence the result depends on $\BD$ while~\cite{ACD15} imposes ``incoherence" on $\bd$ by requiring $\frac{\|\mathcal{F}(\bd)\|_{\infty}}{\|\bd\|}$ relatively small, where $\mathcal{F}$ denotes Fourier transform. 

\end{remark}

\subsection{Self-calibration from multiple snapshots}

We again follow a by now familiar procedure to derive the linear system associated with~\eqref{eq:model3}, which turns out to be
\begin{equation}\label{eq:Amodel3}
\underbrace{
\begin{bmatrix}
\diag(\by_1) & -\BA & \bzero & \cdots & \bzero \\
\diag(\by_2) & \bzero & -\BA & \cdots & \bzero \\
\vdots & \vdots &\vdots & \ddots &  \vdots \\
\diag(\by_p) & \bzero & \bzero & \cdots &  -\BA  \\
\multicolumn{5}{c}{\bw^*}
\end{bmatrix}
_{(mp+1) \times (m + np)}}_{\A_{\bw}}
\underbrace{
\begin{bmatrix}
\bs \\
\bx_1 \\
\bx_2 \\
\vdots \\
\bx_p
\end{bmatrix}_{(m+np)\times 1}}_{\bz}
= 
\begin{bmatrix}
\bzero \\
c \\
\end{bmatrix}_{(mp+1)\times 1}.
\end{equation}
For this scenario we only consider the case when $\BA$ is a complex Gaussian random matrix. 
\begin{theorem}\label{thm:main3}
Consider the self-calibration model given in~\eqref{eq:model3}, where $\A_{\bw}$ is as in~\eqref{eq:Amodel3} and $\A_{\bw,0}$ corresponds to the noiseless part of $\A_{\bw}.$
Let $x_{\min} : = \min_{1\leq l\leq p} \|\bx_l\|$ and $x_{\max} : = \max_{1\leq l\leq p} \|\bx_l\|$.
Then, for the solution  $\hat{\bz}$ of~\eqref{eq:obj-ls} and $\alpha = \frac{c}{\bw^*\bz_0},$ there holds
\begin{equation*}
\frac{ \| \hat{\bz} - \alpha\bz_0 \| }{\|\alpha\bz_0\|} \leq \kappa(\A_{\bw, 0})\eta \left( 1 + \frac{2}{1 - \kappa(\A_{\bw, 0})\eta }\right)
\end{equation*}
if $\kappa(\A_{\bw,0})\eta < 1$ where $\eta = \frac{2\|\delta\A\|}{\sqrt{m}}$. Here the upper bound of $\kappa(\A_{\bw,0})$ obeys
\begin{equation*}
\kappa(\A_{\bw,0}) \leq  \sqrt{\frac{6 x_{\max}^2(m + \|\bw\|^2)}{x^2_{\min}\min\{ m, \|\bw\|^2|\Corr(\bw,\bz_0)|^2 \}}  \frac{ \max\{ pd_{\max}^2, \frac{m}{x^2_{\min}} \} }{ \min\{ pd_{\min}^2, \frac{m}{x^2_{\max}} \}}}., 
\end{equation*}
and for $\|\bw\| = \sqrt{m}$, there holds
\begin{equation*}
\kappa(\A_{\bw,0}) \leq  \frac{2\sqrt{3} x_{\max}}{x_{\min}|\Corr(\bw,\bz_0)| } \sqrt{  \frac{ \max\{ pd_{\max}^2, \frac{m}{x^2_{\min}} \} }{ \min\{ pd_{\min}^2, \frac{m}{x^2_{\max}} \}}}, 
\end{equation*}

with probability at least $1 - 2m(np + m)^{-\gamma}$ if $\BA$ is a complex Gaussian random matrix and
\begin{equation}\label{eq:cond-mult-2}
C_0\left(\max\left\{\frac{\|\BG\|}{p}, \frac{\|\BG\|_F^2}{p^2} \right\} + \frac{n}{m} \right)\log^2(np + m) \leq \frac{1}{16(\gamma + 1)}
\end{equation}
where $\BG$ is the Gram matrix of $\left\{ \frac{\bx_l}{\|\bx_l\|} \right\}_{l=1}^p.$
In particular, if $\BG = \I_p$ and $\|\BG\|_F = \sqrt{p}$,~\eqref{eq:cond-mult-2} becomes
\begin{equation*}
C_0\left( \frac{1}{p} + \frac{n}{m} \right)\log^2(np + m) \leq \frac{1}{16(\gamma + 1)}.
\end{equation*}

\end{theorem}
\begin{remark}
When $\|\delta \A\| = 0$, Theorem~\ref{thm:main3} says that the solution to~\eqref{eq:model3} is uniquely determined up to a scalar if $mp = \mathcal{O}( \max\{np +m\|\BG\|, \sqrt{mnp^2+m^2\|\BG\|_F^2}\})$, which involves the norm of the Gram matrix $\BG$. This makes sense if we consider two extreme cases: if $\{\bv_l\}_{l=1}^p$ are all identical, then we have $\|\BG\| = p$ and $\|\BG\|_F = p$ and if the $\{\bv_l\}_{l=1}^p$ are mutually orthogonal, then $\BG  =\I_p$. 
\end{remark}
\begin{remark}
Balzano and Nowak~\cite{BN07} show exact recovery of this model when $\BA$ is a deterministic DFT matrix (discrete Fourier matrix) and $\{\bx_l\}_{l=1}^p$ are generic signals drawn from a probability distribution, but their results  do not include stability theory in the presence of noise. 
\end{remark}
\begin{remark}
For Theorem~\ref{thm:main2} and Theorem~\ref{thm:main3} it does not come as a surprise that the error bound depends on the norm of $\{\bx_l\}_{l=1}^p$ and $\BD$ as well as on how much $\bz_0$ and $\bw$ are correlated. We cannot expect a relatively good condition number for $\A_{\bw}$ if $\|\bx_l\|$ varies greatly over $1\leq l\leq p$. Concerning the correlation between $\bz_0$ and $\bw$, one extreme case is $\lag \bz_0, \bw\rag = 0$, which does not rule out the possibility of $\bz= \bzero$. Hence, the quantity $\lag \bz_0, \bw\rag$ affects the condition number. 
\end{remark}

\subsection{Theoretical results for the spectral method}
\label{ss:svd-thm}
Let $\MS$ be the matrix $\A_{\bw}$ excluding the last row and consider $\MS =\MS_0 + \delta\MS$ where $\MS_0$ is the noiseless part of $\MS$ and $\delta\MS$ is the noise term. The performance under noise depends on the second smallest singular value of $\MS_0$ and the noise strength $\|\delta \MS\|$.

\begin{theorem}\label{thm:svd}
Denote $\hat{\bz}$ as the solution to~\eqref{prog:svdmin}, i.e., the right singular vector of $\MS$ w.r.t. the smallest singular value and $\|\hat{\bz}\| = 1.$ 
Then there holds,
\begin{equation*}
\min_{\alpha_0\in\CC} \frac{\|\alpha_0 \hat{\bz} - \bz_0\|}{\|\bz_0\|} = \left\|\frac{(\I - \hat{\bz}\hat{\bz}^*)\bz_0}{\|\bz_0\|}\right\| \leq \frac{\|\delta \MS\|}{[ \sigma_2(\MS_0) - \|\delta\MS\| ]_+}
\end{equation*}
where $\sigma_2(\MS_0)$ is the second smallest singular value of $\MS_0$, $\bz_0$ satisfies $\MS_0\bz_0 = 0$, and $\hat{\bz}$ is the right singular vector with respect to the smallest singular value of $\MS$, i.e., the solution to~\eqref{prog:svdmin}.

Moreover,  the lower bound of $\sigma_2(\MS_0)$ satisfies
\begin{enumerate}[(a)]
\item $\sigma_2(\MS_0) \geq \sqrt{\frac{p}{2}} \min\{ \sqrt{m}, d_{\min}\|\bx\| \}$ for model~\eqref{eq:model1} under the assumption of Theorem~\ref{thm:main1};
\item $\sigma_2(\MS_0) \geq \frac{1}{\sqrt{2}}x_{\min} \min\left\{\sqrt{p}d_{\min}, \frac{\sqrt{m}}{x_{\max}}\right\}$ for model~\eqref{eq:model2} under the assumption of Theorem~\ref{thm:main2};
\item $\sigma_2(\MS_0) \geq \frac{1}{\sqrt{2}}x_{\min} \min\left\{\sqrt{p}d_{\min}, \frac{\sqrt{m}}{x_{\max}}\right\}$ for model~\eqref{eq:model3} under the assumption of Theorem~\ref{thm:main3}.
\end{enumerate}
\end{theorem}
\begin{remark}
Note that finding $z$ which minimizes~\eqref{prog:svdmin} is equivalent to finding the eigenvector with respect to the smallest eigenvalue of $\MS^*\MS$. If we  have a good approximate upper bound $\lambda$ of $\|\MS^*\MS\|$, then it suffices to find the leading eigenvector of $\lambda\I - \MS^*\MS$, which can be done efficiently by using power iteration.


How to choose $\lambda$ properly in practice? We do not want to choose $\lambda$ too large since this would imply slow convergence of the power iteration.
For each case, it is easy to get a good upper bound of $\|\MS\|^2$ based on the measurements and sensing matrices, 
\begin{enumerate}[(a)]
\item for~\eqref{eq:model1}, $\lambda = \|\sum_{l=1}^p\diag(\by_l)\diag(\by_l^*)\| + \|\sum_{l=1}^p\BA_l^* \BA_l\|$; 
\item for~\eqref{eq:model2}, $\lambda = \|\diag(\by_l)\diag(\by_l^*)\| + \max_{1\leq l\leq p}\|\BA_l\|^2$; 
\item for~\eqref{eq:model3}, $\lambda = \|\diag(\by_l)\diag(\by_l^*)\| + \|\BA\|^2.$
\end{enumerate}
Those choices of $\lambda$ are used in our numerical simulations. 
\end{remark}

\section{Numerical simulations}\label{s:numerics}
This section is devoted to numerical simulations. Four experiments for both synthetic and real data will be presented to address the effectiveness, efficiency and robustness of the proposed approach. 

For all three models presented~\eqref{eq:model1},~\eqref{eq:model2} and~\eqref{eq:model3}, the corresponding linear systems have simple block structures which allow for fast implementation via the conjugate gradient method for non-hermitian matrices~\cite{saad2003iterative}. In our simulations, we do not need to set up $\A_{\bw}$ explicitly to carry out the matrix-vector multiplications arising in the conjugate gradient method. Moreover, applying preconditioning via rescaling all the columns to be of similar norms can give rise to an even faster convergence rate.  Therefore, we are able to deal with medium- or large-scale problems from image processing in a computationally efficient manner. 

In our simulations
the iteration stops if either the number of iterations reaches at most $2000$ or the residual of the corresponding normal equation is smaller than $10^{-8}$. 
Throughout our discussion, the $\SNR$ (signal-to-noise ratio) in the scale of dB is defined as
\begin{equation*}
\SNR : = 10\log_{10}\left( \frac{\sum_{l=1}^p \|\by_l\|^2 }{ \sum_{l=1}^p \|\beps_l\|^2 }\right).
\end{equation*}
We measure the performance with $\text{RelError (in dB)} := 20\log_{10} \text{RelError}$ where 
\begin{equation*}
\text{RelError} = \max\left\{ \min_{\alpha_1\in\CC}\frac{\|\alpha_1 \hat{\bx} - \bx_0\|}{\|\bx_0\|},\min_{\alpha_2\in\CC}\frac{\|\alpha_2 \hat{\bd} - \bd_0\|}{\|\bd_0\|}\right\}.
\end{equation*}
Here $\bd_0$ and $\bx_0$ are the ground truth. Although RelError does not match the error bound in our theoretic analysis, certain equivalent relations hold if one assumes all $|d_{i,0}|$ are bounded away from 0 because there holds $\hat{\bs}\approx \bs_0$ if and only if $\hat{\bd}\approx \bd_0.$

For the  imaging examples, we only measure the relative error with respect to the recovered image $\hat{\bx}$, i.e., $\min_{\alpha_1\in\CC}\frac{\|\alpha_1 \hat{\bx} - \bx_0\|}{\|\bx_0\|}.$

\subsection{Self-calibration from repeated measurements}
Suppose we have a target image $\bx$ and try to estimate $\bx$ through multiple measurements. However, the sensing process is not perfect because of the missing calibration of the sensors. In order to estimate both the unknown gains and phases as well as the target signal, a randomized sensing procedure is used by employing several random binary masks.

We assume that $\by_l = \BD\BH\BM_l\bx_0 + \beps_l$ where $\BH$ is a ``tall" low-frequency DFT matrix, $\BM_l$ is a diagonal  $\pm 1$-random matrix and $\bx_0$ is an image of size $512\times 512.$ We set $m = 1024^2,n = 512^2, p = 8$ and  $\BD = \diag(\bd)\in 1024^2 \times 1024^2$ with $\bd\in\CC^{1024^2}$; the oversampling ratio is $\frac{pm}{m + n} = 6.4$.  We compare two cases: (i) $\bd$ is a sequence distributed uniformly over $[0.5,1.5]$ with $\bw = \bone_{m + n}$,  and (ii) $\bd$ is a Steinhaus sequence (uniformly distributed over the complex unit circle) with $\bw =\begin{bmatrix} \bzero_m \\ \bone_n\end{bmatrix}$. We pick those choices of $\bw$ because we know that the  image we try to reconstruct has only non-negative values. Thus, by choosing $\bw$ to be non-negative, there are fewer cancellation in the expression $\bw^*\bz_0$, which in turn leads to a smaller condition number and better robustness. 
The corresponding results of our simulations are shown in Figure~\ref{fig:pos-repeat-1} and Figure~\ref{fig:pos-repeat-2}, respectively. In both cases, we only measure the relative error of the recovered image. 

\begin{figure}[h!]
\begin{minipage}{1\textwidth}
\centering
\includegraphics[width = 120mm]{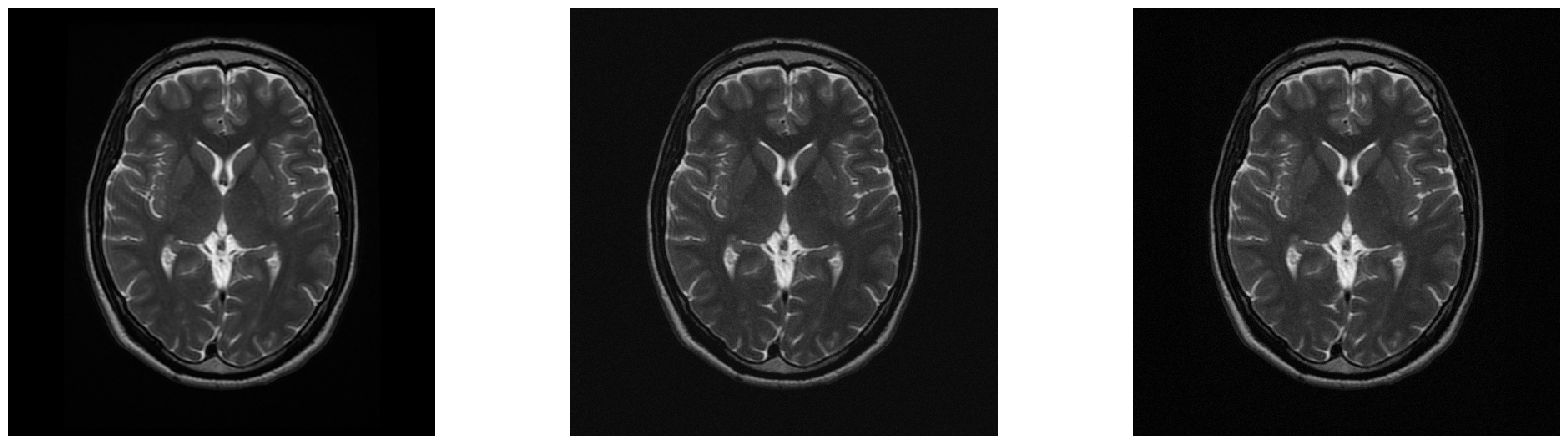}
\caption{Here $m = 1024^2$, $n = 512^2$, $p = 8$, SNR=$5$dB, $\BD = \diag(\bd)$ where $d_i$ is uniformly distributed over $[0.5,1.5]$. Left: original image; Middle: uncalibrated image, RelError $= -13.85$dB; Right: calibrated image, RelError = $-20.23$dB}
\label{fig:pos-repeat-1}
\end{minipage}
\vfill
\vskip0.5cm
\begin{minipage}{1\textwidth}
\centering
\includegraphics[width = 120mm]{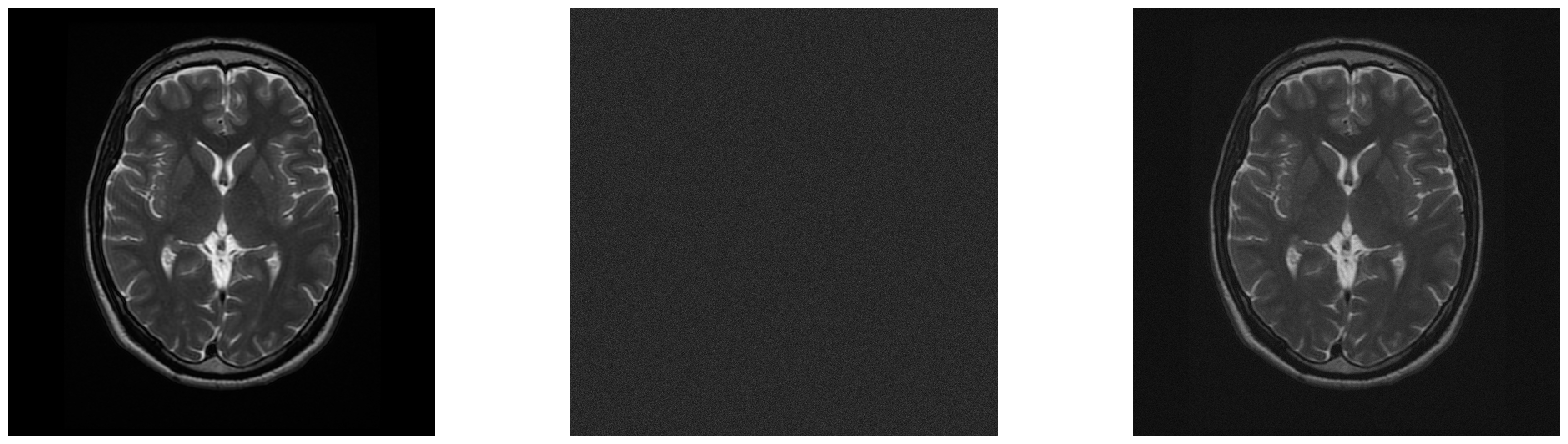}
\caption{Here $m = 1024^2$, $n = 512^2$, $p = 8$, SNR=$5$dB, $\BD = \diag(\bd)$ where $\bd$ is a Steinhauss sequence. Left: original image; Middle: uncalibrated image; Right: calibrated image, RelError = $-10.02$dB}
\label{fig:pos-repeat-2}
\end{minipage}
\end{figure}

In Figure~\ref{fig:pos-repeat-1}, we can see that both, the uncalibrated and the calibrated image are quite good. Here the uncalibrated image is obtained by first applying the inverse Fourier transform and the inverse of the mask to each $\by_i$ and then taking the average of $p$ samples. We explain briefly why the uncalibrated image still looks good. Note that
\begin{equation*}
\hat{x}_{uncali} = \frac{1}{p}\sum_{l=1}^p\BM_l^{-1}\BH^{\dagger }( \BD\BH\BM_l\bx_0 ) = \bx_0 + \frac{1}{p}\sum_{l=1}^p\BM_l^{-1}\BH^{\dagger  } (\BD -\I) \BH\BM_l\bx_0
\end{equation*}
where $\BH^{\dagger} =\frac{1}{m}\BH^*$ is the pseudo inverse of $\BH.$
Here $\BD - \I$ is actually a diagonal matrix with random  entries $\pm \frac{1}{2}$. As a result, each
$\BM_l^{-1}\BH^{\dagger  } (\BD -\I) \BH\BM_l\bx_0 $ is the sum of $m$ rank-1 matrices with random $\pm\frac{1}{2}$ coefficients and is relatively small due to many cancellations. Moreover,~\cite{CLO84} showed that most 2-D signals can be reconstructed within a scale factor from only knowing the phase of its Fourier transform, which applies to the case when $\bd$ is positive. 

However, when the unknown calibration parameters are complex variables (i.e., we do not know much about the phase information), Figure~\ref{fig:pos-repeat-2} shows that the uncalibrated recovered image is totally meaningless. Our approach still gives a quite satisfactory result even at a relatively low SNR of  5dB. 

\subsection{Blind deconvolution in random mask imaging}
The second experiment is about blind deconvolution in random mask imaging~\cite{BR15,TB14}. Suppose we observe the convolution of two components,
\begin{equation*}
\by_l = \bh \ast \BM_l \bx_0 + \beps_l, \quad 1\leq l\leq p
\end{equation*}
where both, the filter $\bh$ and the signal of interests $\bx_0$ are unknown. Each $\BM_l$ is a random $\pm 1$-mask. The blind deconvolution problem is to recover $(\bh, \bx_0)$. Moreover, here we assume that the filter is actually a low-pass filter, i.e., $\mathcal{F}(\bh)$ is compactly supported in an interval around the origin, where $\mathcal{F}$ is the Fourier transform. After taking the Fourier transform on both sides, the model actually ends up being of the form~\eqref{eq:model1} with $\BA_l = \BH\BM_l$ where $\BH$ is a ``fat" partial DFT matrix and $\bd$ is the nonzero part of $\mathcal{F}(\bh)$. In our experiment, we let $\bx_0$ be a $128\times 128$ image and $\bd = \mathcal{F}(\bh)$ be a 2-D Gaussian filter of size $45\times 45$ as shown in Figure~\ref{fig:random-mask-1}. 

\begin{figure}[h!]
\begin{minipage}{1\textwidth}
\centering
\includegraphics[width = 80mm]{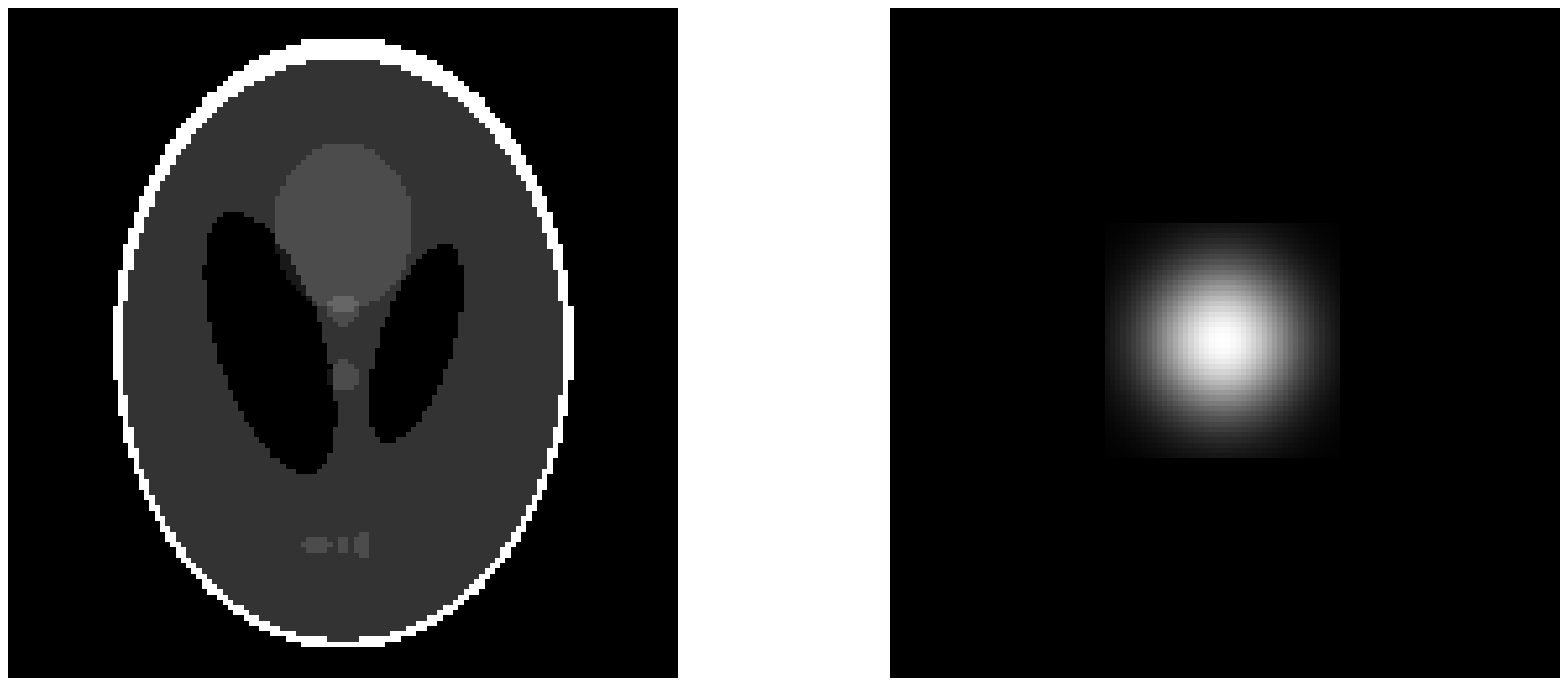}
\caption{Left: Original image; Right: Gaussian filter in Fourier domain. The support of the filter is $45\times 45$ and hence $m = 45^2 = 2025, n = 128^2, p =32.$}
\label{fig:random-mask-1}
\end{minipage}
\vfill
\vskip0.5cm
\begin{minipage}{1\textwidth}
\centering
\includegraphics[width = 80mm]{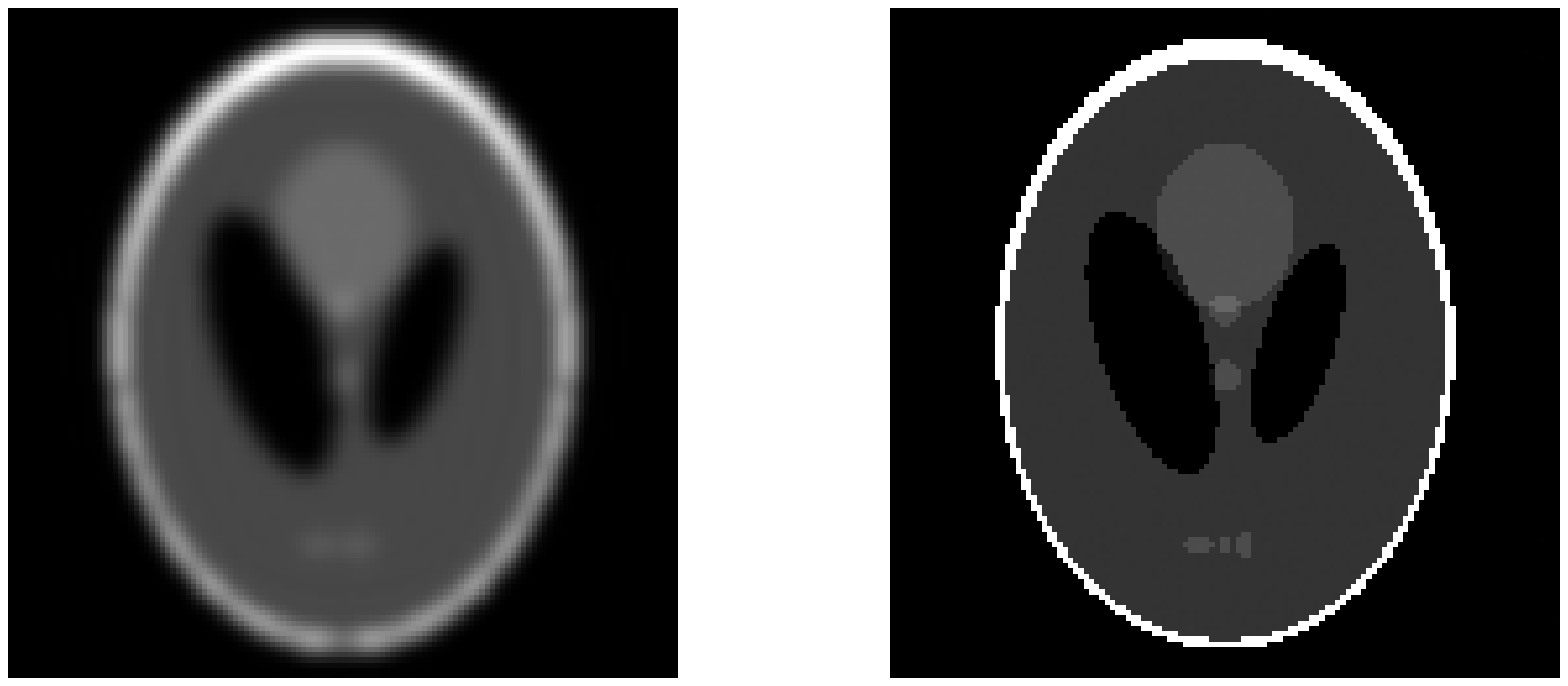}
\caption{Left: Blurred image without noise, Right: Recovered image, RelError = -45.47dB}
\label{fig:random-mask-2}
\end{minipage}
\vfill
\vskip0.5cm
\begin{minipage}{1\textwidth}
\centering
\includegraphics[width = 80mm]{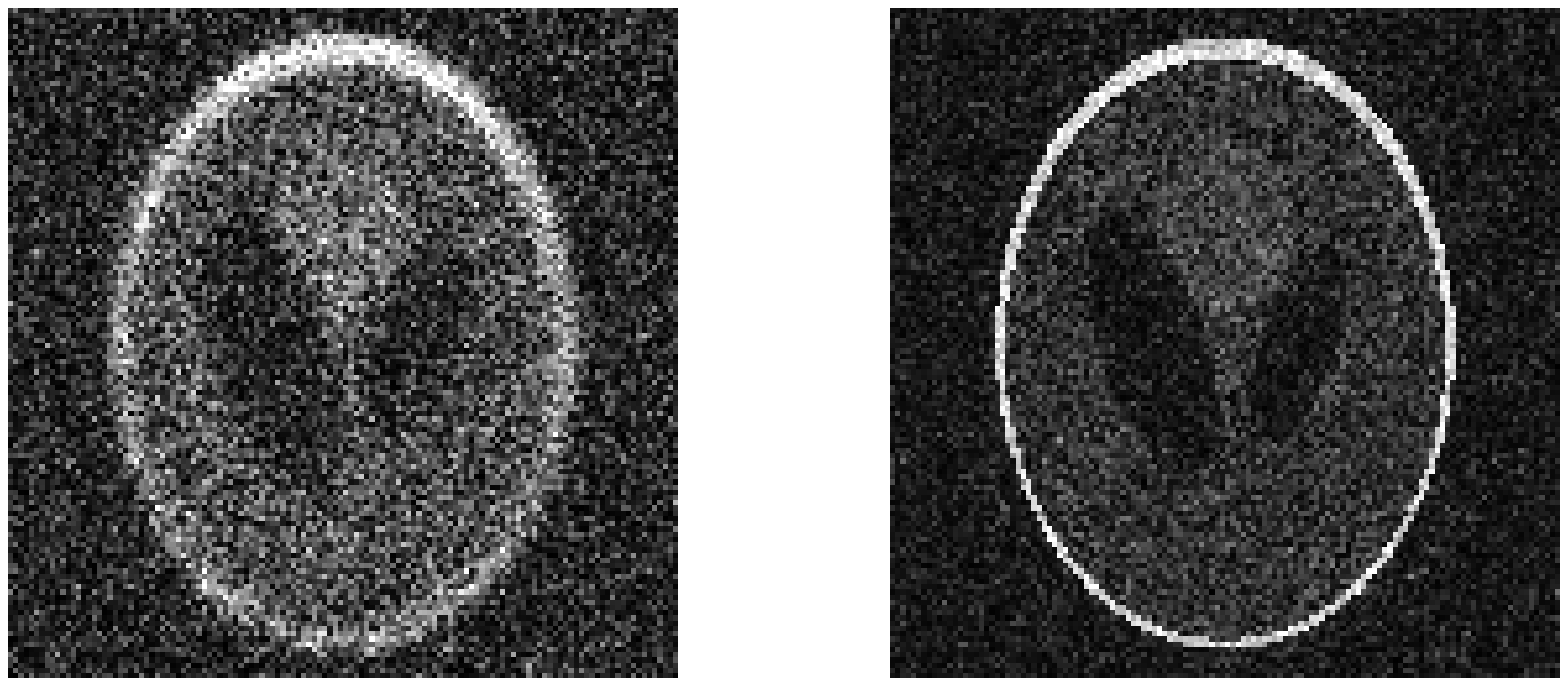}
\caption{Left: Blurred image with SNR = $5$dB,  Right: Recovered image, RelError = -5.84dB}
\label{fig:random-mask-3}
\end{minipage}
\end{figure}

In those experiments, we choose $\bw = \bone_{m+n}$ since both $\bd$ and $\bx_0$ are nonnegative. Figure~\ref{fig:random-mask-2} shows the recovered image from $p=32$ sets of noiseless measurements and the performance is quite satisfactory. Here the oversampling ratio is $\frac{pm }{m + n } \approx 3.52.$ We can see from Figure~\ref{fig:random-mask-3} that the blurring effect has been removed while the noise still exists. That is partially because we did not impose any denoising procedure after the deconvolution. A natural way to improve this reconstruction further would be to combine the blind deconvolution method with a total variation minimization denoising step.

\subsection{Blind deconvolution via diverse inputs}
\label{s:numerics-diverse}
We choose $\BA_l$ to be random Hadamard matrices with $m = 256$ and $n = 64$ and $\BD = \diag(\bd_0)$ with $\bd_0$ being a positive/Steinhaus sequence, as we did previously. Each $\bx_l$ is sampled from standard Gaussian distribution.  
We choose $\bw = \begin{bmatrix}\bone_m \\ \bzero_{np\times 1}\end{bmatrix}$ if $\bd_0$ is uniformly distributed over $[0.5, 1.5]$ and $\bw = \begin{bmatrix}\sqrt{m}\be_1 \\ \bzero_{np\times 1}\end{bmatrix}$ for Steinhaus $\bd_0.$ 10 simulations are performed for each level of $\SNR$. The test is also given under different choices of $p$. The oversampling ratio $\frac{pm}{pn + m}$ is $2$, $2.67$  and  $3$ for $p = 4,8,12$ respectively. From Figure~\ref{fig:diverse}, we can see that the error scales linearly with SNR in dB. The performance of Steinhaus $\bd_0$ is not as good as that of positive $\bd_0$ for SNR $\leq 10$ when we use the least squares method. That is because $\bw^*\bz_0$ is quite small when $\bd_0$ is complex and $\bw = \begin{bmatrix}\sqrt{m}\be_1 \\ \bzero_{np\times 1}\end{bmatrix}$. Note that the error between $\hat{\bz}$ and $\bz_0$ is bounded by $\kappa(\A_{\bw,0})\eta \left( 1 + \frac{2}{1 - \kappa(\A_{\bw,0})\eta }\right)$.
Therefore, the error bound does not depend on $\|\delta \A\|$ linearly if $\kappa(\A_{\bw,0})\eta$ is close to 1. This may explain the nonlinear behavior of the relative error in the low SNR regime.

We also apply spectral method to this model with complex gains $\bd_0$ and $\BA_l$ chosen as either a random Hadamard matrix or a Gaussian matrix. Compared to the linear least squares approach, the spectral method is much more robust to noise, as suggested in Figure~\ref{fig:diverse2}, especially in the low SNR regime.
\begin{figure}[h!]
\begin{minipage}{0.48\textwidth}
\includegraphics[width=66mm]{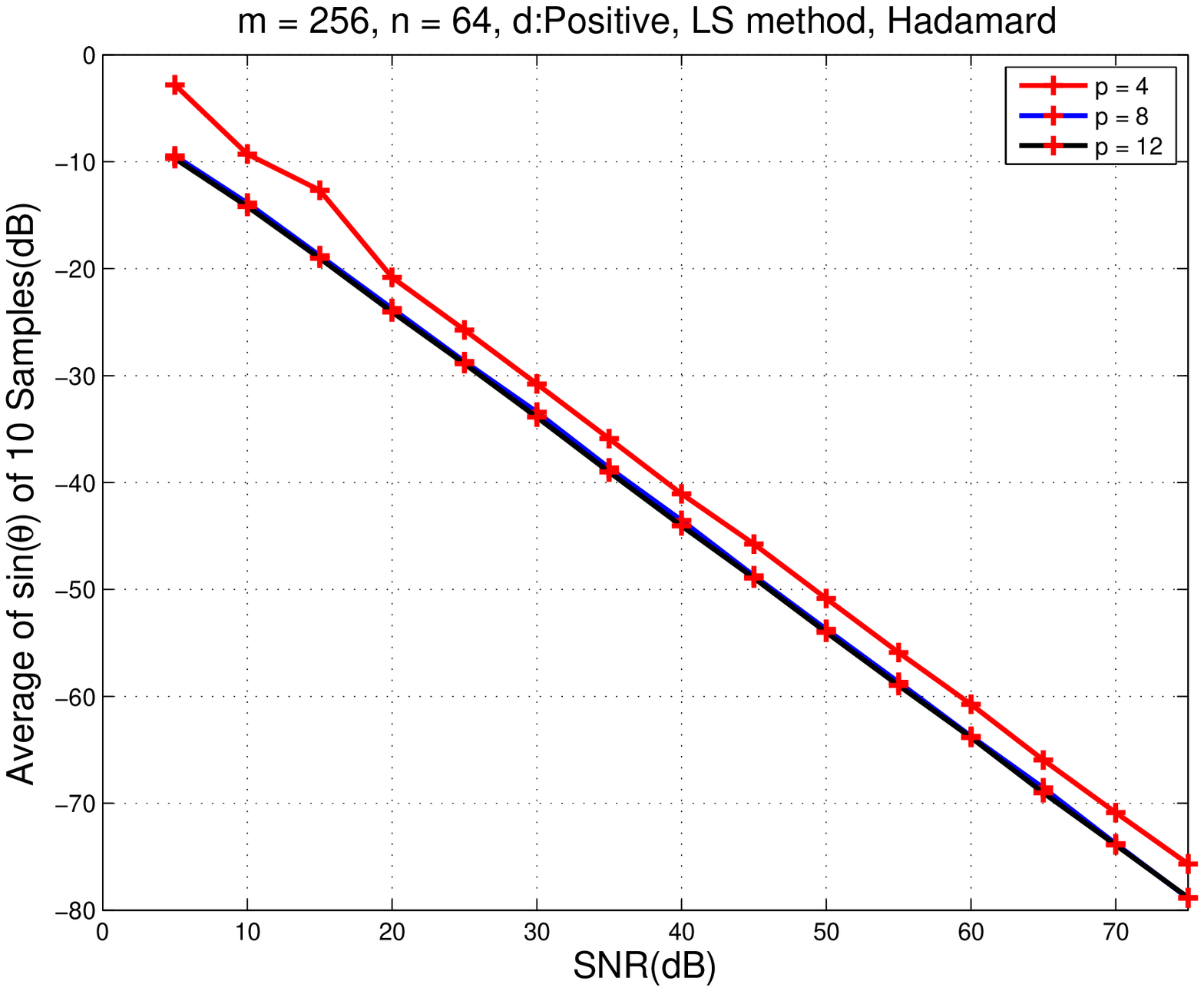}
\end{minipage}
\hfill
\begin{minipage}{0.48\textwidth}
\includegraphics[width=66mm]{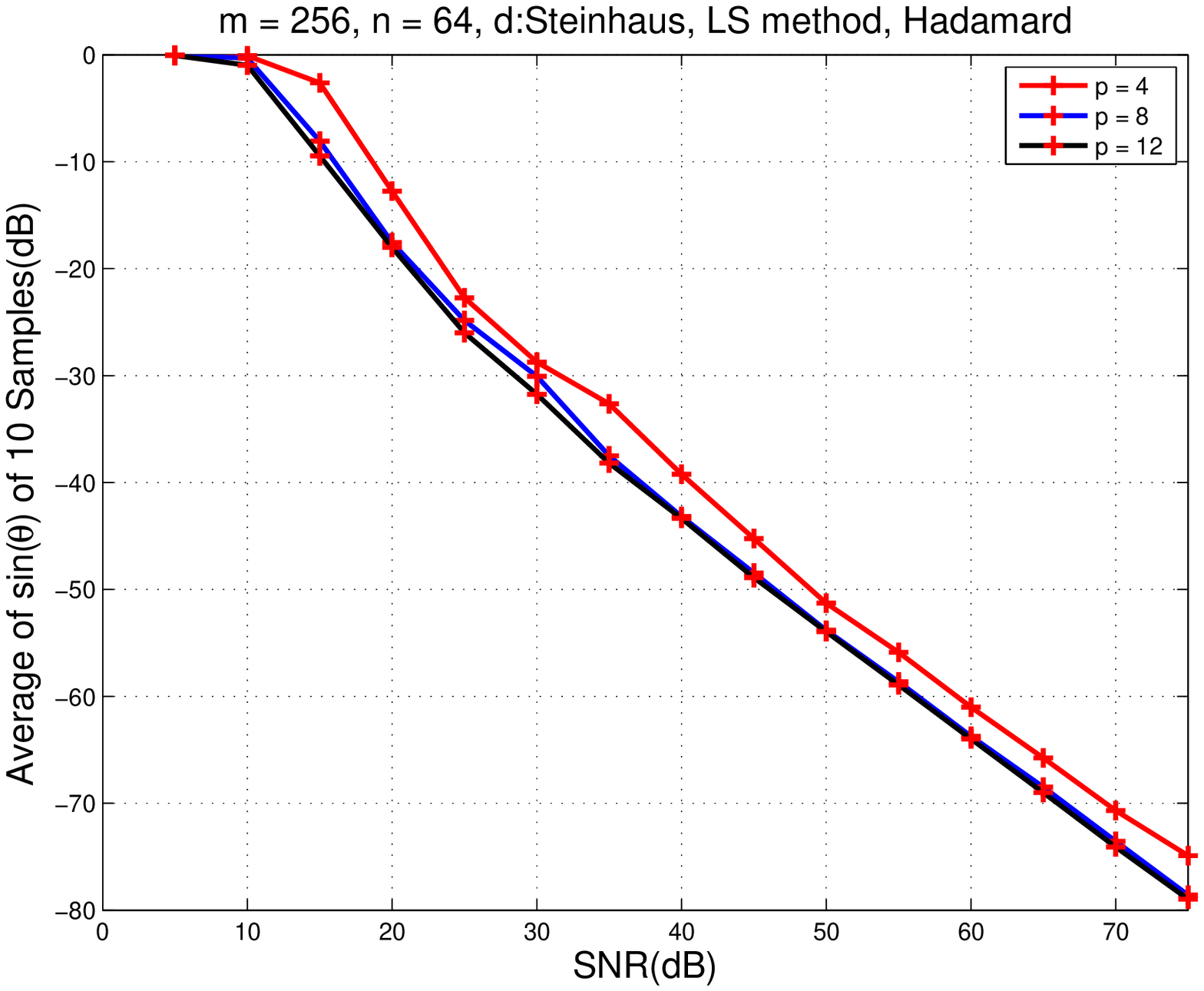}
\end{minipage}
\caption{ Performance of linear least squares approach: RelError (in dB) vs. SNR for $\by_l = \BD\BA_l\bx_l + \beps_l,1\leq l\leq p$ where $m = 256,n=64$, $\BD = \diag(\bd_0)$ and each $\BA_l$ is a random Hadamard matrix. $\bd_0$ is a Steinhaus sequence.}
\label{fig:diverse}
\end{figure}

\begin{figure}[h!]
\begin{minipage}{0.48\textwidth}
\includegraphics[width=66mm]{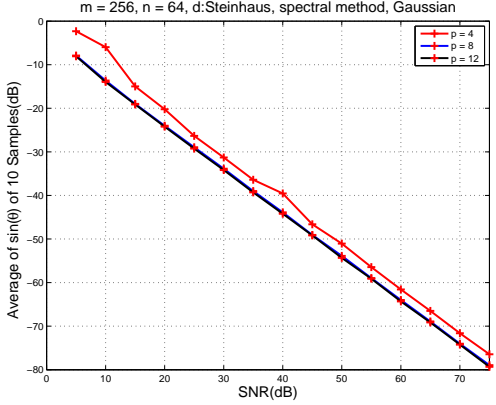}
\end{minipage}
\hfill
\begin{minipage}{0.48\textwidth}
\includegraphics[width=66mm]{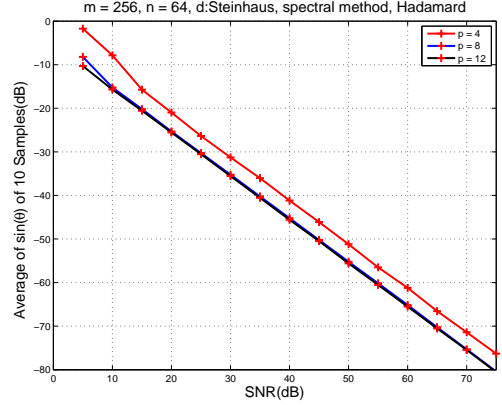}
\end{minipage}
\caption{ Performance of spectral method: RelError (in dB) vs. SNR for $\by_l = \BD\BA_l\bx_l + \beps_l,1\leq l\leq p$ where $m = 256,n=64$, $\BD = \diag(\bd_0)$. Here $\BA_l$ is either a random Hadamard matrix or a Gaussian matrix. $\bd_0$ is a Steinhaus sequence.}
\label{fig:diverse2}
\end{figure}

\subsection{Multiple snapshots}
We make a comparison of performances between the linear least squares approach and the spectral method when $\bd_0$ is a Steinhaus sequence and $\BA$ is a Gaussian random matrix $\BA$. 
Each $\bx_l$ is sampled from the standard Gaussian distribution and hence the underlying Gram matrix $\BG$ is quite close to $\I_p$ (this closeness could be easily  made more precise, but we refrain doing so here).
The choice of $\bw$ and oversampling ratio are the same as those in Section~\ref{s:numerics-diverse}. From Figure~\ref{fig:multiple}, we see that the performance is not satisfactory for the Steinhaus case using the linear least squares approach, especially in the lower SNR regime (SNR $\leq 10$). The reason is the low correlation between $\bw$ and $\bz_0$ if $\bd_0$ is Steinhaus and $\bw = \begin{bmatrix}\sqrt{m}\be_1 \\ \bzero_{np\times 1}\end{bmatrix}.$ The difficulty of choosing $\bw$ is avoided by the spectral method. As we can see in Figure~\ref{fig:multiple}, the relative error given by spectral method is approximately 7dB smaller than that given by linear least squares approach when $\bd_0$ is a complex vector and the SNR is smaller than $20$dB.

\begin{figure}[h!]
\begin{minipage}{0.48\textwidth}
\includegraphics[width=66mm]{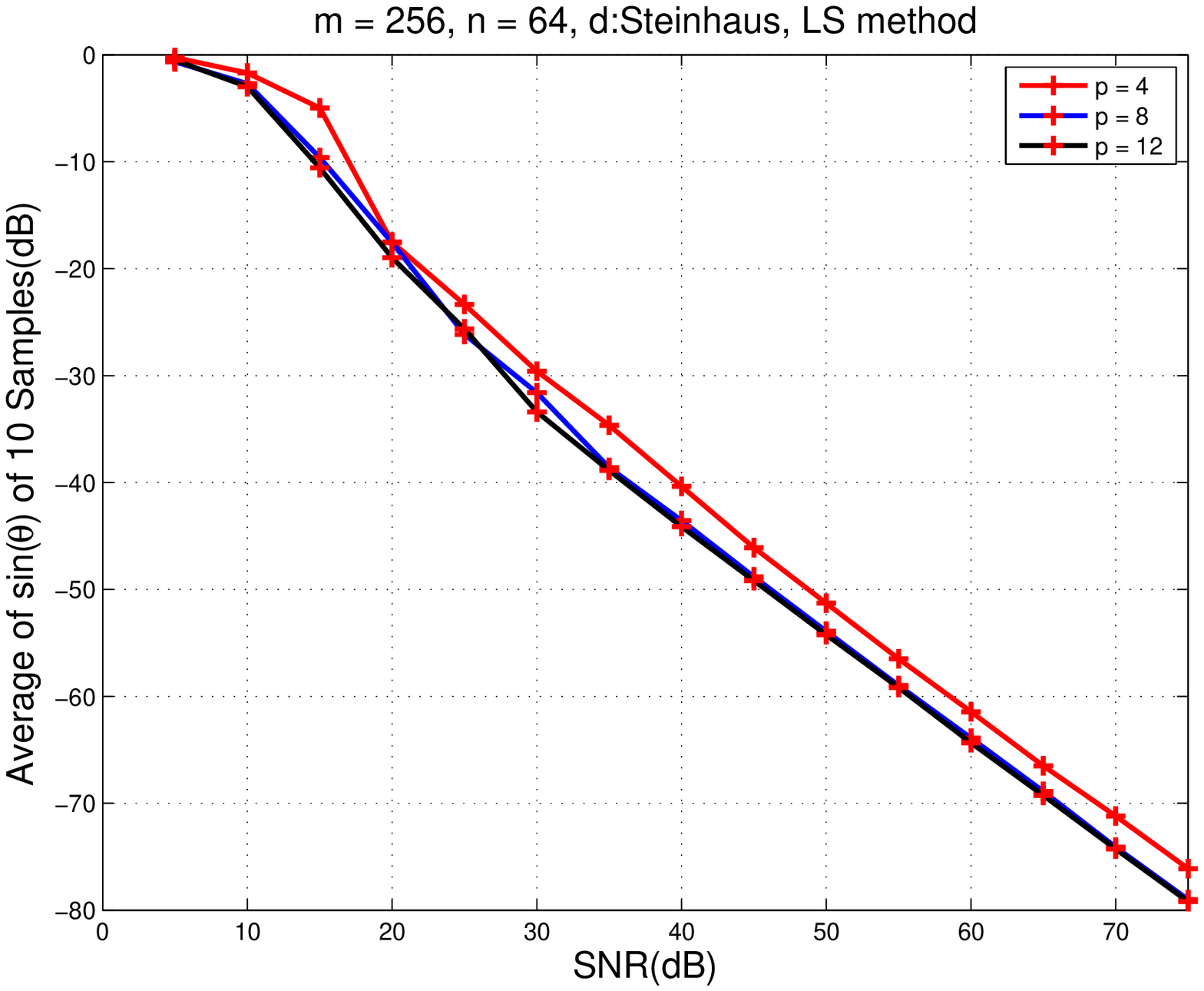}
\end{minipage}
\hfill
\begin{minipage}{0.48\textwidth}
\includegraphics[width=66mm]{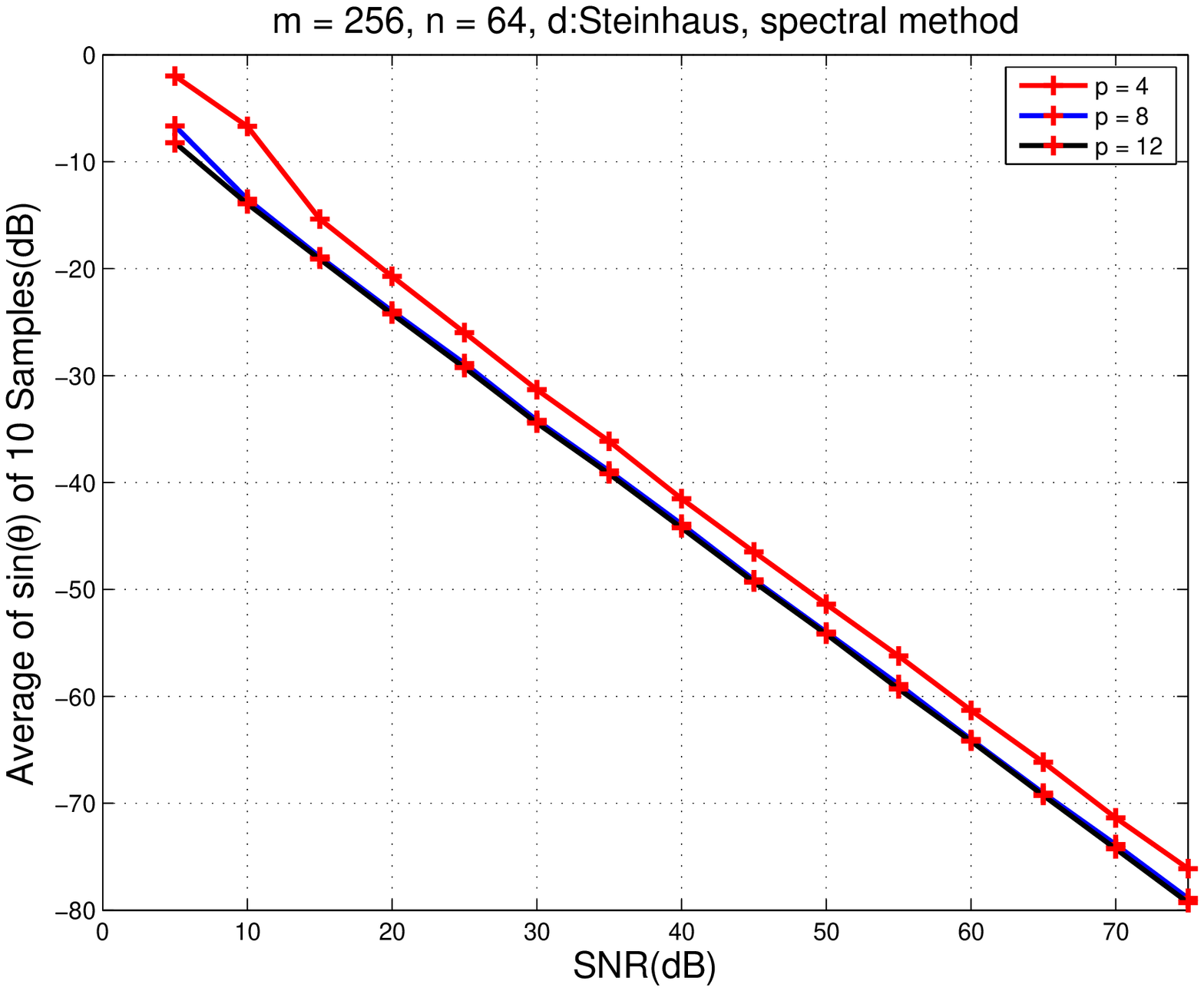}
\end{minipage}
\caption{ Comparison between linear least squares approach and spectral method:  RelError (in dB) vs. SNR for $\by_l = \BD\BA\bx_l + \beps_l,1\leq l\leq p$ where $m = 256,n=64$, $\BD = \diag(\bd_0)$ and $\BA$ is a Gaussian random matrix. The gain $\bd_0$ is a random vector with each entry uniformly distributed over unit circle.}
\label{fig:multiple}
\end{figure}

\section{Proofs}\label{s:proof}
For each subsection, we will first give the result of noiseless measurements. We then prove the stability theory by using the result below. The proof of spectral method can be found in Section~\ref{ss:svd-bilinear}.
\begin{proposition}{~\cite{W89}}\label{prop:perturb}
Suppose that $\BA \bu_0 = \bb$ is a consistent and overdetermined system. Denote $\hat{\bu}$ as the least squares solution to $\|(\BA + \delta\BA)\bu - \bb\|^2$ with $\|\delta \BA\| \leq\eta \|\BA\|$. If $\kappa(\BA)\eta < 1$, there holds,
\begin{equation*}
\frac{\|\hat{\bu} - \bu_0\|}{\|\bu_0\|} \leq \kappa(\BA)\eta \left( 1 + \frac{2}{1 - \kappa(\BA)\eta }\right).
\end{equation*}
\end{proposition}
To apply the proposition above, it suffices to bound $\kappa(\BA)$ and $\eta$.

\subsection{Self-calibration from repeated measurements}
Let us start with~\eqref{eq:linear} when $\beps_l = \bzero$ and denote $\BLam_l := \diag(\overline{\BA_l \bv})$ and $\bv := \frac{\bx}{\|\bx\|}\in\CC^n$, 
\begin{equation*}
\A_0: = \begin{bmatrix}
\diag(\by_1) & -\BA_1 \\
\vdots & \vdots\\
\diag(\by_p) &  - \BA_p
\end{bmatrix}
= 
\begin{bmatrix}
\BLam_1^* & -\frac{1}{\sqrt{m}}\BA_1 \\
\vdots & \vdots \\
\BLam_p^* & -\frac{1}{\sqrt{m}}\BA_p
\end{bmatrix}
\begin{bmatrix}
\BD\|\bx\| & \bzero \\
\bzero & \sqrt{m}\I_n
\end{bmatrix}.
\end{equation*}
Then we rewrite $\A_{\bw}^*\A_{\bw}$ as 
\begin{equation*}
\A_{\bw}^*\A_{\bw} = \A_0^*\A_0+ \bw\bw^* 
= \sum_{l=1}^p \BP
\BZ_l\BZ_l^*
\BP^*
+ 
\bw\bw^* 
\end{equation*}
where 
\begin{equation}\label{def:model1-Zl}
\BZ_l := 
\begin{bmatrix}
\BLam_l \\
 -\frac{1}{\sqrt{m}}\BA_l^* \\
\end{bmatrix} \in \CC^{ (m + n)\times m},
\quad 
\BP := 
\begin{bmatrix}
\BD^* \|\bx\|& \bzero \\
\bzero & \sqrt{m}\I_n
\end{bmatrix}
\in \CC^{(m + n)\times (m + n)}.
\end{equation}
By definition,
\begin{equation*}
\BZ_l\BZ_l^* = 
\begin{bmatrix}
\BLam_l\BLam_l^* & -\frac{1}{\sqrt{m}}\BLam_l \BA_l \\
-\frac{1}{\sqrt{m}} \BA_l^*\BLam_l^* &  \frac{1}{m}\BA_l^* \BA_l
\end{bmatrix} \in \CC^{(m + n)\times (m + n)}.
\end{equation*}
Our goal is to find out the smallest and the largest eigenvalue of $\A_{\bw}$. Actually it suffices to understand the spectrum of $\sum_{l=1}^p \BZ_l\BZ_l^*$. Obviously, its smallest eigenvalue is zero and the corresponding eigenvector is $\bu_1 := \frac{1}{\sqrt{2}}\begin{bmatrix}
\frac{\bone_m}{\sqrt{m}} \\
\bv
\end{bmatrix}.
$
Let $\ba_{l,i}$ be the $i$-th column of $\BA_l^*$ and we have $\E(\ba_{l,i}\ba_{l,i}^*) = \I_n$
under all the three settings in Section~\ref{s:model1-setup}.
Hence,
$\BC : = \E(\BZ_l\BZ_l^*) = \begin{bmatrix}
\I_m & -\frac{1}{\sqrt{m}} \bone_m\bv^* \\
-\frac{1}{\sqrt{m}}\bv\bone_m^* & \I_n 
\end{bmatrix}.$

It is easy to see that $\rank(\BC) = m + n - 1$ and the null space of $\BC$ is spanned by $\bu_1$. $\BC$ has an eigenvalue with value $1$ of multiplicity $m + n-2$ and an eigenvalue with value $2$ of multiplicity $1$.
More importantly, the following proposition holds and combined with Proposition~\ref{prop:perturb}, we are able to prove Theorem~\ref{thm:main1}. 
\begin{proposition}\label{prop:main1} 
There holds
\begin{equation*}
\left\| \sum_{l=1}^p \BZ_l\BZ_l^* - p\BC\right\| \leq \frac{p}{2} 
\end{equation*}
\begin{enumerate}[(a)]
\item with probability $1 - (m + n)^{-\gamma}$ if $\BA_l$ is Gaussian and $p \geq c_0 \gamma \max\left\{ 1, \frac{n}{m} \right\}  \log^2(m + n)$;
\item with probability $1 - (m + n)^{-\gamma} - 2(mp)^{-\gamma + 1}$ if each $\BA_l$ is a ``tall" $(m\times n, m \geq n)$ random Hadamard/DFT matrix and $p \geq c_0 \gamma^2 \log(m + n)\log(mp)$;
\item with probability $1 - (m + n)^{-\gamma} - 2(mp)^{-\gamma + 1}$ if each $\BA_l$ is a ``fat" $(m\times n, m\leq n)$ random Hadamard/DFT matrix and $mp \geq c_0 \gamma^2n \log(m + n)\log(mp)$.
\end{enumerate}
\end{proposition}

\begin{remark}\label{rmk:prop1}
Proposition~\ref{prop:main1} actually addresses the identifiability issue of the model~\eqref{eq:model1} in absence of noise. More precisely, the invertibility of $\BP$ is guaranteed by that of $\BD$. By Weyl's theorem for singular value perturbation in~\cite{Stewart90}, $m + n - 1$ eigenvalues of $\sum_{l=1}^p \BZ_l\BZ_l^*$ are greater than $\frac{p}{2}.$ Hence, the rank of $\A_0$ is equal to $\rank(\sum_{l=1}^p \BZ_l\BZ_l^*) = m+n-1$ if $p$ is close to the information theoretic limit under the conditions given above, i.e., $p \geq \mathcal{O}(\frac{n}{m})$. In other words, the null space of $\A_0$ is completely spanned by $\bz_0 := (\bs_0, \bx_0)$.
\end{remark}


\subsubsection{Proof of Theorem~\ref{thm:main1} }
Note that Proposition~\ref{prop:main1} gives the result if $\beps_l = \bzero$. The noisy counterpart is obtained by applying perturbation theory for linear least squares.
\begin{proof}
Let $\A_{\bw} : = \A_{\bw,0} + \delta\A$ where $\A_{\bw,0}$ is the noiseless part and $\delta \A$ is defined in~\eqref{def:deltaA}. Note that $\alpha\bz_0$ with $\alpha = \frac{c}{\bw^*\bz_0}$ is actually a solution to the overdetermined system $\A_{\bw, 0} \bz = \begin{bmatrix} \bzero \\ c\end{bmatrix}$ by the definition of $\A_{\bw, 0}$.
Proposition~\ref{prop:perturb} implies that it suffices to estimate the condition number $\kappa(\A_{\bw,0})$ of $\A_{\bw,0}$ and $\eta$ such that $\|\delta\A\| \leq \eta \|\A_{\bw,0}\|$ holds. 
Note that
\begin{eqnarray}
\A^*_{\bw,0}\A_{\bw,0}
& = & \BP\left(\sum_{l=1}^p \BZ_l\BZ_l^*\right) \BP^* + \bw\bw^*  \label{eq:Aw-1} \\
& = & \BP \left(  \sum_{l=1}^p
\BZ_l\BZ_l^*
+ 
\widetilde{\bw}\widetilde{\bw}^* 
\right)
\BP^* =:\BP\widetilde{\BC}\BP^* \label{eq:Aw-2}
\end{eqnarray}
where $\tbw : = \BP^{-1}\bw.$
From Proposition~\ref{prop:main1} and Theorem 1 in~\cite{Stewart90}, we know that
\begin{equation*}
\lambda_2\left(\sum_{l=1}^p \BZ_l\BZ_l^*\right) \geq \frac{p}{2}, \quad \lambda_{n + m}\left(\sum_{l=1}^p \BZ_l\BZ_l^*\right) \leq \frac{5p}{2}
\end{equation*}
where $\lambda_1 \leq  \cdots \leq \lambda_{n + m}$ and $\lambda_1(\sum_{l=1}^p \BZ_l\BZ_l^*) = 0.$
Following from~\eqref{eq:Aw-1}, we have
\begin{eqnarray}
\|\A_{\bw,0}^*\A_{\bw,0}\| & \leq & \|\BP\|^2 \left\| \sum_{l=1}^p \BZ_l\BZ_l^*\right\| + \|\bw\|^2 \leq 3p\max\{ d_{\max}^2\|\bx\|^2, m \} + \|\bw\|^2 \nonumber \\
& \leq & 3\left(p + \frac{\|\bw\|^2}{m}\right) \lambda_{\max}^2(\BP). \label{eq:Aw-upbound}
\end{eqnarray}
On the other hand,
\begin{equation}\label{eq:Aw-op-low}
\|\A_{\bw,0}^*\A_{\bw,0}\| \geq \left\| \sum_{l=1}^p \BA_l^*\BA_l \right\|  \geq \frac{mp}{2}
\end{equation}
follows from Proposition~\ref{prop:main1}.
In other words, we have found the lower and upper bounds for $\|\A_{\bw,0}\|$ or equivalently, $\lambda_{\max}(\A^*_{\bw,0}\A_{\bw,0}).$ Now we proceed to the estimation of $\lambda_{\min}(\A^*_{\bw,0}\A_{\bw,0}).$
Let $\bu := \sum_{j=1}^{m+n} \alpha_j \bu_j$ be a unit vector, where $\bu_1 = \frac{1}{\sqrt{2}}\begin{bmatrix}
\frac{\bone_m}{\sqrt{m}} \\
\bv
\end{bmatrix}
$
with $\sum_{j=1}^{m+n} |\alpha_j|^2 = 1$
and $\{\bu_j\}_{j=1}^{m+n}$ are the eigenvectors of $\sum_{l=1}^p \BZ_l\BZ_l^*$. Then the smallest eigenvalue of $\widetilde{\BC}$ defined in~\eqref{eq:Aw-2} has a lower bound as follows:
\begin{eqnarray}
\bu^*\widetilde{\BC}\bu
& = & 
\bu^* \left(\sum_{l=1}^p \BZ_l\BZ_l^*\right) \bu + 
\bu^* 
\tbw\tbw^*
\bu \nonumber \\
& \geq & \sum_{j=2}^{m+n} \lambda_j |\alpha_j|^2 + 
\bu^* \bu_1\bu_1^*\tbw\tbw^* \bu_1\bu_1^*
\bu \nonumber \\
& \geq &\sum_{j=2}^{m+n} \lambda_j |\alpha_j|^2 +  \frac{|\bw^*\bz_0|^2  }{2m\|\bx\|^2} |\alpha_1|^2 \label{eq:Aw1-low}
\end{eqnarray}
which implies $\lambda_{\min}(\widetilde{\BC}) \geq \frac{1}{2}\min\left\{p, \frac{|\bw^*\bz_0|^2}{m\|\bx\|^2}\right\} \geq \frac{1}{2}\min\left\{p, \frac{ \|\bw\|^2 |\Corr(\bw,\bz_0)|^2}{m}\right\} $.
Combined with $\A_{\bw,0}^*\A_{\bw,0} = \BP\widetilde{\BC}\BP^*$, 
\begin{equation*}
\lambda_{\min}(\A_{\bw,0}^*\A_{\bw,0}) \geq \lambda_{\min}(\widetilde{\BC})\lambda_{\min}^2(\BP) \geq \frac{1}{2m}\min\left\{mp, \|\bw\|^2|\Corr(\bw,\bz_0)|^2\right\} \lambda^2_{\min}(\BP).
\end{equation*}
Therefore, with~\eqref{eq:Aw-upbound}, the condition number of $\A_{\bw,0}^*\A_{\bw,0}$ is bounded by
\begin{equation*}
\kappa(\A_{\bw,0}^*\A_{\bw,0}) \leq  \frac{6(mp + \|\bw\|^2)}{\min\{mp, \|\bw\|^2|\Corr(\bw,\bz_0)|^2\}} \kappa^2(\BP).
\end{equation*}
From~\eqref{eq:Aw-op-low}, we set $\eta =  \frac{2\|\delta\A\|}{\sqrt{mp}} \geq \frac{\|\delta \A\|}{\|\A_{\bw,0}\|}.$

Applying Proposition~\ref{prop:perturb} gives the following upper bound of the estimation error $
\frac{ \| \hat{\bz} - \alpha\bz_0 \| }{\|\alpha\bz_0\|} \leq \kappa(\A_{\bw, 0})\eta \left( 1 + \frac{2}{1 - \kappa(\A_{\bw, 0})\eta }\right)$
where $\alpha = \frac{c}{\bw^*\bz_0}$ and 
\begin{equation*}
\kappa(\A_{\bw,0})  \leq  \sqrt{ \frac{6 (mp + \|\bw\|^2)}{\min\{mp,\|\bw\|^2 |\Corr(\bw,\bz_0)|^2\}} }\kappa(\BP).
\end{equation*}
If $\|\bw\| = \sqrt{mp}$, the upper bound of $\kappa(\A_{\bw,0})$ satisfies
\begin{equation*}
\kappa(\A_{\bw,0})  \leq  \sqrt{ \frac{12}{ |\Corr(\bw,\bz_0)|^2 } }\kappa(\BP) \leq \frac{2\sqrt{3}}{|\text{Corr}(\bw,\bz_0)|}  \sqrt{  \frac{\max\{d^2_{\max}\|\bx\|^2, m\}}{\min\{d^2_{\min}\|\bx\|^2,m\}}}.
\end{equation*}
\end{proof}

\subsubsection{Proof of Proposition~\ref{prop:main1}(a)}
\begin{proof}{\textbf{[Proof of Proposition~\ref{prop:main1}(a)]}}
From now on, we assume $\ba_{l,i}\in\CC^n$, i.e., the $i$-th column of $\BA_l^*$, obeys a complex Gaussian distribution, $\frac{1}{\sqrt{2}}\mathcal{N}(\bzero, \I_n) + \frac{\mi}{\sqrt{2}}\mathcal{N}(\bzero, \I_n)$. 
Let $\bz_{l,i}$ be the $i$-th column of $\BZ_l$; it can be written in explicit form as
\begin{equation*}
\bz_{l,i} = 
\begin{bmatrix}
\overline{\lag \ba_{l,i}, \bv\rag} \be_i \\
-\frac{1}{\sqrt{m}} \ba_{l,i}
\end{bmatrix},
\quad 
\bz_{l,i}\bz_{l,i}^* = 
\begin{bmatrix}
|\lag \ba_{l,i}, \bv\rag|^2 \be_i\be_i^* & -\frac{1}{\sqrt{m}} \be_i\bv^*\ba_{l,i}\ba_{l,i}^* \\ 
-\frac{1}{\sqrt{m}}\ba_{l,i}\ba_{l,i}^*\bv\be_i^* & \frac{1}{m} \ba_{l,i}\ba_{l,i}^*.
\end{bmatrix}
\end{equation*}
Denoting $\CZ_{l,i} : = \bz_{l,i}\bz_{l,i}^* - \E(\bz_{l,i}\bz_{l,i}^*)$, we obtain 
\begin{equation*}
\CZ_{l,i} := 
\begin{bmatrix}
(|\lag \ba_{l,i}, \bv\rag|^2 - 1) \be_i\be_i^* & -\frac{1}{\sqrt{m}} \be_i\bv^*(\ba_{l,i}\ba_{l,i}^* - \I_n)\\ 
-\frac{1}{\sqrt{m}}(\ba_{l,i}\ba_{l,i}^*- \I_n)\bv\be_i^* & \frac{1}{m} (\ba_{l,i}\ba_{l,i}^* - \I_n)
\end{bmatrix}.
\end{equation*}
Obviously each $\CZ_{l,i}$ is independent. In order to apply Theorem~\ref{thm:bern1} to estimate $\|\sum_{l,i} \CZ_{l,i}\|$, we need to bound $\max_{l,i} \|\CZ_{l,i}\|_{\psi_1}$ and $\left\| \sum_{l=1}^p\sum_{i=1}^m\E(\CZ_{l,i}\CZ_{l,i}^*) \right\|$. Due to the semi-definite positivity of $\bz_{l,i}\bz_{l,i}^*$, we have $\|\CZ_{l,i}\| \leq \max\{ \|\bz_{l,i}\|^2, \|\E(\bz_{l,i}\bz_{l,i}^*)\|\} \leq \max\left\{|\lag\ba_{l,i}, \bv \rag|^2 + \frac{1}{m} \|\ba_{l,i}\|^2, 2\right\}$ and hence
\begin{equation*}
\|\CZ_{l,i}\|_{\psi_1} \leq  (|\lag\ba_{l,i}, \bv \rag|^2)_{\psi_1} + \frac{1}{m} (\|\ba_{l,i}\|^2)_{\psi_1} \leq C\left(1 + \frac{n}{m}\right)
\end{equation*}
which follows from Lemm~\ref{lem:psi1} and $\|\cdot\|_{\psi_1}$ is a norm. 
This implies $R : = \max_{l,i}\|\CZ_{l,i}\|_{\psi_1} \leq C\left(1 + \frac{n}{m}\right).$

Now we consider $\sigma^2_0 = \left\|\sum_{l=1}^p\sum_{i=1}^m\E(\CZ_{l,i}\CZ_{l,i}^*)\right\|$ by computing $(\CZ_{l,i}\CZ_{l,i}^*)_{1,1}$ and $(\CZ_{l,i}\CZ_{l,i}^*)_{2,2}$, i.e., the $(1,1)$-th and $(2,2)$-th block of $\CZ_{l,i}\CZ_{l,i}^*$,
\begin{eqnarray*}
(\CZ_{l,i}\CZ_{l,i}^*)_{1,1} & = & \left[(|\lag \ba_{l,i}, \bv\rag|^2 - 1)^2 + \frac{1}{m} \bv^*(\ba_{l,i}\ba_{l,i}^* - \I_n)^2 \bv\right] \be_i\be_i^*, \\
(\CZ_{l,i}\CZ_{l,i}^*)_{2,2} & = & \frac{1}{m} (\ba_{l,i}\ba_{l,i}^* - \I_n)\bv\bv^*(\ba_{l,i}\ba_{l,i}^* - \I_n) + \frac{1}{m^2}(\ba_{l,i}\ba_{l,i}^* - \I_n)^2.
\end{eqnarray*}
Following from~\eqref{eq:gauss-3}, ~\eqref{eq:gauss-4}, ~\eqref{eq:gauss-5} and Lemma~\ref{lem:pos}, there holds
\begin{eqnarray*}
\sigma^2_0 & = & \left\|\sum_{l=1}^p\sum_{i=1}^m\E(\CZ_{l,i}\CZ_{l,i}^*)\right\|  \leq 2\left\| \sum_{l=1}^p \sum_{i=1}^m 
\begin{bmatrix}
\E(\CZ_{l,i}\CZ_{l,i}^*)_{1,1} & \bzero \\
\bzero & \E(\CZ_{l,i}\CZ_{l,i}^*)_{2,2}
\end{bmatrix}
\right\| \\
& = & 
2\left\| \sum_{l=1}^p \sum_{i=1}^m 
\begin{bmatrix}
\left(1 + \frac{n}{m}\right)\be_i\be_i^* & \bzero \\
\bzero & \left(\frac{1}{m} + \frac{n}{m^2}\right)\I_n
\end{bmatrix}
\right\| \\
& = & 
2p\left\| 
\begin{bmatrix}
\left(1 + \frac{n}{m}\right)\I_m & \bzero \\
\bzero & \left(1 + \frac{n}{m}\right)\I_n
\end{bmatrix}
\right\| = 2p\left(1 + \frac{n}{m}\right).
\end{eqnarray*}
By applying the matrix Bernstein inequality (see Theorem~\ref{BernGaussian}) we obtain
\begin{eqnarray*}
 \left\| \sum_{l=1}^p\sum_{i=1}^m \CZ_{l,i} \right\| 
 & \leq & C_0\max\Big\{ \sqrt{p\left(1 + \frac{n}{m}\right)} \sqrt{t + \log(m + n)}, \\ 
&& \left(1 + \frac{n}{m}\right)(t + \log(m + n))\log(m + n) \Big\}  \leq \frac{p}{2}
\end{eqnarray*}
with probability $1 - e^{-t}$. In particular, by choosing $t = \gamma \log(m + n)$, i.e, 
$p \geq c_0 \gamma \max\left\{ 1, \frac{n}{m} \right\}  \log^2(m + n),$
the inequality above holds with probability $1 -  (m + n)^{-\gamma}.$
\end{proof}

\subsubsection{Proof of Proposition~\ref{prop:main1}(b)}
\begin{proof}{\textbf{[Proof of Proposition~\ref{prop:main1}(b)]}}
Each $\BZ_l$ is independent by its definition in~\eqref{def:model1-Zl} if $\BA_l : = \BH\BM_l$ where $\BH$ is an $m\times n$ partial DFT/Hadamard matrix with $m\geq n$ and $\BH^*\BH = m\I_n$ and $\BM_l = \diag(\bsm_l)$ is a  diagonal random binary $\pm 1$ matrix.
Let $\CZ_{l} := \BZ_l\BZ_l^* - \BC \in \CC^{(m + n)\times (m +n )} $; in explicit form
\begin{equation*}
\CZ_l = 
\begin{bmatrix}
\BLam_l\BLam_l^* - \I_m & - \frac{1}{\sqrt{m}}(\BLam_l \BA_l  - \bone_m\bv^*)\\
- \frac{1}{\sqrt{m}}(\BA_l^*\BLam_l^*  - \bv\bone_m^*)  & \bzero
\end{bmatrix}.
\end{equation*}
where $\BA_l^*\BA_l = m\I_n$ follows from the assumption.
First we take a look at the upper bound of $\|\CZ_l\|.$ It suffices to bound $\|\BZ_l\BZ_l^*\|$ since $\|\CZ_l\| = \|\BZ_l\BZ_l^*- \BC\|  \leq \max\left\{\|\BZ_l \BZ_l^*\|, \|\BC\|\right\}$ and $\BC$ is positive semi-definite. On the other hand, due to Lemma~\ref{lem:pos}, we have $\|\BZ_l\BZ_l^*\| \leq 2\max\{\|\BLam_l\BLam_l^*\|, 1\}$ and hence we only need to bound $\|\BLam_l\|$.
For $\|\BLam_l\|$, there holds
\begin{equation*}
\max_{1\leq l\leq p}\|\BLam_{l}\| =\max_{1\leq l\leq p} \|\BA_l \bv\|_{\infty}  = \max_{1\leq l\leq p, 1\leq i\leq m} |\lag \ba_{l,i}, \bv\rag|.
\end{equation*}
Also for any pair of $(l,i)$, $\lag \ba_{l,i}, \bv \rag$ can be rewritten as
\begin{equation*}
|\lag \ba_{l,i}, \bv \rag| = |\lag \diag(\bsm_l)\bh_{i}, \bv \rag| = |\lag \bsm_l, \diag(\bar{\bh}_{i})\bv\rag|
\end{equation*}
where $\bh_{i}$ is the $i$-th column of $\BH^*$ and $\|\diag(\bar{\bh}_{i})\bv\| = \|\bv\| = 1$. Then there holds
\begin{eqnarray}
\Pr\left(\max_{1\leq l\leq p}\|\BLam_{l}\| \geq \sqrt{2\gamma\log (mp)} \right)
& \leq & \sum_{l=1}^p\sum_{i=1}^m \Pr\left( |\lag \ba_{l,i}, \bv \rag| \geq \sqrt{2\gamma\log (mp)}\right) \nonumber \\
& \leq & mp \Pr\left( |\lag \ba_{l,i}, \bv \rag| \geq \sqrt{2\gamma\log (mp)}\right) \nonumber \\
& \leq & 2mp \cdot e^{- \gamma\log(mp)} \leq 2 (mp)^{-\gamma +1}, \label{eq:BLam}
\end{eqnarray}
where the third inequality follows from Lemma~\ref{lem:rade}.
Applying Lemma~\ref{lem:pos} to $\CZ_l$, 
\begin{equation*}
R : = \max_{1\leq l\leq p}\|\CZ_l\| \leq 2 \max_{1\leq l\leq p}\{ \|\BLam_l\BLam_l^*\|, 1\} \leq 4\gamma\log (mp), 
\end{equation*}
with probability at least $1 - 2(mp)^{-\gamma+ 1}$. Denote the event $\{\max_{1\leq l\leq p}\|\CZ_l\| \leq 4\gamma\log (mp)\}$ by $E_1$.

Now we try to understand $\sigma_0^2 = \| \sum_{l=1}^p\E(\CZ_l\CZ_l^*)\|.$
The $(1,1)$-th and $(2,2)$-th block of $\CZ_l\CZ_l^*$ are given by
\begin{eqnarray*}
(\CZ_l\CZ_l^*)_{1,1} & = & (\BLam_l\BLam_l^* - \I_m)^2 + \frac{1}{m}(\BLam_l\BA_l - \bone_m\bv^*)(\BLam_l\BA_l - \bone_m\bv^*)^*, \\ 
(\CZ_l\CZ_l^*)_{2,2} & = & \frac{1}{m} (\BA_l^*\BLam_l^* - \bv\bone_m^*) (\BLam_l\BA_l - \bone_m\bv^*).
\end{eqnarray*}
By using~\eqref{eq:binary-3},~\eqref{eq:binary-4} and $\BA_l\BA_l^* = \BH\BH^* \preceq m\I_m$, we have
\begin{eqnarray}
\E((\BLam_l\BLam_l^* - \I_m)^2) & = & \E(\BLam_l\BLam_l^*)^2 - \I_m \preceq 2\I_m, \label{eq:self-1} \\
\qquad \E(\BLam_l\BA_l - \bone_m\bv^*)(\BLam_l\BA_l - \bone_m\bv^*)^* & = & \E(\BLam_l\BA_l\BA_l^*\BLam_l^*) - \bone_m\bone_m^* \preceq m\I_m,  \label{eq:self-2} \\
\E (\BA_l^*\BLam_l^* - \bv\bone_m^*) (\BLam\BA_l - \bone_m\bv^*) & = & \E(\BA_l^*\BLam_l^*\BLam_l\BA_l) - m\bv\bv^* \nonumber \\
& = & \sum_{i=1}^m |\lag \ba_{l,i}, \bv\rag|^2 \ba_{l,i}\ba_{l,i}^* - m\bv\bv^*
\preceq  3m\I_n \label{eq:self-3}.
\end{eqnarray}
Combining~\eqref{eq:self-1},~\eqref{eq:self-2},~\eqref{eq:self-3} and Lemma~\ref{lem:pos}, 
\begin{equation*}
\sigma_0^2 \leq 2\left\|\sum_{l=1}^p
\begin{bmatrix}
\E(\CZ_l\CZ_l^*)_{1,1} & \bzero \\
\bzero & \E(\CZ_l\CZ_l^*)_{2,2} 
\end{bmatrix}
\right\| \leq 2p \left\|
\begin{bmatrix}
2\I_m + \I_m & \bzero \\
\bzero & 3\I_n
\end{bmatrix}
\right\|
\leq 6p.
\end{equation*}
By applying~\eqref{thm:bern} with $t = \gamma \log(m + n)$ and $R \leq 4\gamma \log(mp)$ over event $E_1$, we have
\begin{equation*}
\left\| \sum_{l=1}^p (\CZ_l\CZ_l^*- \BC)\right\| \leq C_0 \max\{ \sqrt{p} \sqrt{(\gamma + 1) \log(m + n)}, \gamma (\gamma + 1) \log (mp)\log(m + n) \} \leq \frac{p}{2}
\end{equation*}
with probability $1 - (m + n)^{-\gamma} - 2(mp)^{-\gamma + 1}$ if $p \geq c_0 \gamma^2 \log(m + n)\log(mp)$.
\end{proof}

\subsubsection{Proof of Proposition~\ref{prop:main1}(c)}

\begin{proof}{\textbf{[Proof of Proposition~\ref{prop:main1}(c)]}}
Each $\BZ_l$ is independent due to~\eqref{def:model1-Zl}. 
Let $\CZ_{l} := \BZ_l\BZ_l^* - \BC \in \CC^{(m + n)\times (m +n )} $; 
in explicit form
\begin{equation*}
\CZ_l = 
\begin{bmatrix}
\BLam_l\BLam_l^* - \I_m & - \frac{1}{\sqrt{m}}(\BLam_l \BA_l  - \bone_m\bv^*)\\
- \frac{1}{\sqrt{m}}(\BA_l^*\BLam_l  - \bv\bone_m^*)  & \frac{1}{m}\BA_l^*\BA_l - \I_n
\end{bmatrix} .
\end{equation*}
Here $\BA_l = \BH\BM_l$  where $\BH$ is a ``fat" $m\times n, (m \leq n)$ partial DFT/Hadamard matrix satisfying $\BH\BH^* = n\I_m$ and $\BM_l$ is a diagonal  $\pm1$-random matrix. There holds $\BA_l^*\BA_l = \BM_l^*\BH^*\BH\BM_l \in \CC^{n\times n}$ where $\BH\in\CC^{m\times n}$ and $\E(\BA_l^*\BA_l) = m\I_n.$ For each $l$, 
$\left\| \BA_l^*\BA_l \right\| = \left\| \BH^*\BH \right\| = \|\BH\BH^*\| = n.$
Hence, there holds,
\begin{equation*}
\|\CZ_l\| \leq \max \left\{\| \BZ_l\BZ_l^* \|, \|\BC\|\right\| \leq 2\max\left\{\frac{1}{m}\|\BA_l^*\BA_l\|, \|\BLam_l\|^2, 1\right\} \leq 2\max\left\{ \frac{n}{m},\gamma\log(mp) \right\}
\end{equation*}
with probability at least $1 - 2(mp)^{-\gamma + 1}$, which follows exactly from~\eqref{eq:BLam} and Lemma~\ref{lem:pos}.

Now we give an upper bound for $\sigma_0^2 := \|\sum_{l=1}^p \E(\CZ_l\CZ_l^*)\|$. 
The $(1,1)$-th and $(2,2)$-th block of $\CZ_l\CZ_l^*$ are given by
\begin{eqnarray*}
(\CZ_l\CZ_l^*)_{1,1} & = & (\BLam_l\BLam_l^* - \I_m)^2 + \frac{1}{m}(\BLam_l\BA_l - \bone_m\bv^*)(\BLam\BA_l - \bone_m\bv^*)^*, \\ 
(\CZ_l\CZ_l^*)_{2,2} & = & \frac{1}{m} (\BA_l^*\BLam_l^* - \bv\bone_m^*) (\BLam_l\BA_l - \bone_m\bv^*) + \left(\frac{1}{m}\BA_l^*\BA_l - \I_n\right)^2.
\end{eqnarray*}
By using~\eqref{eq:binary-3},~\eqref{eq:binary-4} and $\BA_l\BA_l^* = \BH\BH^* = n\I_m$, we have
\begin{eqnarray}
\E((\BLam_l\BLam_l^* - \I_m)^2) & = & \E(\BLam_l\BLam_l^*)^2 - \I_m \preceq 2\I_m, \label{eq:self-1b} \\
\qquad\quad \E(\BLam_l\BA_l - \bone_m\bv^*)(\BLam_l\BA_l - \bone_m\bv^*)^* & = & \E(\BLam_l\BA_l\BA_l^*\BLam_l^*) - \bone_m\bone_m^* \preceq n\I_m,  \label{eq:self-2b} \\
\E (\BA_l^*\BLam_l^* - \bv\bone_m^*) (\BLam_l\BA_l - \bone_m\bv^*) & = & \E(\BA_l^*\BLam_l^*\BLam_l\BA_l) - m\bv\bv^* \nonumber \\
& = & \sum_{i=1}^m |\lag \ba_{l,i}, \bv\rag|^2 \ba_{l,i}\ba_{l,i}^* - m\bv\bv^*
\preceq 3m\I_n. \label{eq:self-3b} 
\end{eqnarray}
For $\E\left(\frac{1}{m}\BA_l^*\BA_l - \I_n\right)^2$, we have
\begin{equation*}
\left(\frac{1}{m}\BA_l^*\BA_l - \I_n\right)^2 = \frac{1}{m^2} \BA_l^*\BA_l\BA_l^*\BA_l - \frac{2}{m}\BA_l^*\BA_l + \I_n = \frac{n - 2m}{m^2}\BA_l^*\BA_l + \I_n
\end{equation*}
where $\BA_l\BA_l^* = n\I_m.$
Note that $\E(\BA_l^*\BA_l) = m\I_n$, and there holds, 
\begin{equation}\label{eq:self-4b}
\E\left(\frac{1}{m}\BA_l^*\BA_l - \I_n\right)^2 = \frac{n - m}{m} \I_n.
\end{equation}
Combining~\eqref{eq:self-1b},~\eqref{eq:self-2b},~\eqref{eq:self-3b},~\eqref{eq:self-4b} and Lemma~\ref{lem:pos}, 
\begin{equation*}
\sigma_0^2 \leq 2p \left\|
\begin{bmatrix}
\left(2+ \frac{n}{m}\right)\I_m & \bzero \\
\bzero & \left(2+ \frac{n}{m}\right)\I_n
\end{bmatrix}
\right\|
\leq \frac{6np}{m}.
\end{equation*}
By applying~\eqref{thm:bern}  with $t = \gamma \log(m + n)$, we have
\begin{equation*}
\left\| \sum_{l=1}^p (\CZ_l\CZ_l^*- \BC)\right\| \leq C_0 \max\{ \sqrt{\frac{np}{m}} \sqrt{(\gamma + 1) \log(m + n)}, (\gamma + 1) \left(\gamma \log (mp) + \frac{n}{m}\right)\log(m + n) \} \leq \frac{p}{2}
\end{equation*}
with probability $1 - (m + n)^{-\gamma} - 2(mp)^{-\gamma + 1}$ if $mp \geq c_0 \gamma^2n \log(m + n)\log(mp)$.
\end{proof}

\subsection{Blind deconvolution via diverse inputs }
We start with~\eqref{eq:A0-2-ms} by setting $\beps_l = \bzero$. In this way, we can factorize the matrix $\A_{\bw}$ (excluding the last row) into 
\begin{equation}\label{eq:A0-2}
\A_0 :=
\underbrace{
\begin{bmatrix}
\|\bx_1\|\I_m & \cdots & \bzero \\
\vdots & \ddots & \vdots \\
\bzero &  \cdots &\|\bx_p\|\I_m
\end{bmatrix}}_{\BQ}
\underbrace{
\begin{bmatrix}
\frac{\diag(\BA_1\bv_1)}{\sqrt{p}} & -\frac{\BA_1}{\sqrt{m}} & \cdots & \bzero \\
\vdots  & \vdots &  \ddots & \vdots \\
\frac{\diag(\BA_p\bv_p)}{\sqrt{p}}  & \bzero & \cdots &  -\frac{\BA_p}{\sqrt{m}}
\end{bmatrix}
}_{\BZ\in\CC^{mp\times (np+m)}}
\underbrace{
\begin{bmatrix}
\sqrt{p}\BD & \bzero & \cdots & \bzero \\
\bzero & \frac{\sqrt{m}\I_n}{\|\bx_1\|} & \cdots & \bzero \\
\vdots & \vdots & \ddots & \vdots \\
\bzero & \bzero &  \cdots & \frac{\sqrt{m}\I_n}{\|\bx_p\|}
\end{bmatrix}}_{\BP}
\end{equation}
where $\bv_l = \frac{\bx_l}{\|\bx_l\|}$ is the normalized $\bx_l$, $1\leq l\leq p$. We will show that the matrix $\BZ$ is of rank $np + m - 1$ to guarantee that the solution is unique (up to a scalar).
Denote $\bv := \begin{bmatrix} \bv_1 \\ \vdots \\ \bv_p\end{bmatrix}\in\CC^{np\times 1}$ and $\sum_{l=1}^p \tbe_l\otimes\bv_l = \bv$ with $\|\bv\| = \sqrt{p}$ where $\{\tbe_l\}_{l=1}^p$ is a standard orthonormal basis in $\RR^p$.

\subsubsection{Proof of Theorem~\ref{thm:main2}}

The proof of Theorem~\ref{thm:main2} relies on the following proposition. 
We defer the proof of Proposition~\ref{prop:main2} to the Sections~\ref{s:model2-A} and~\ref{s:model2-B}.
\begin{proposition}\label{prop:main2}
There holds,
\begin{equation}\label{eq:ZZ-C}
\left\| \BZ^*\BZ - \BC\right\|   \leq \frac{1}{2}, \quad \BC: = \E(\BZ^*\BZ)= \begin{bmatrix}
\I_m & -\frac{1}{\sqrt{mp}}\bone_m\bv^* \\
-\frac{1}{\sqrt{mp}}\bv\bone_m^* & \I_{np}
\end{bmatrix}
\end{equation}
\begin{enumerate}[(a)]
\item with probability at least $1 - (np +m)^{-\gamma}$ with $\gamma \geq 1$ if $\BA_l$ is an $m\times n$ $(m > n)$ complex Gaussian random matrix and 
\begin{equation*}
C_0\left( \frac{1}{p} + \frac{n}{m}\right)(\gamma + 1)\log^2(np + m)  \leq \frac{1}{4};
\end{equation*}
\item with probability at least $1 - (np + m)^{-\gamma} - 2 (mp)^{-\gamma +1}$ with $\gamma \geq 1$ if $\BA_l$ yields~\eqref{eq:model2-Al} and
\begin{equation*}
C_0\left(\frac{1}{p} + \frac{n-1}{m-1}\right)\gamma^3 \log^4(np + m) \leq \frac{1}{4}.
\end{equation*}

\end{enumerate}

\end{proposition}
\begin{remark}\label{rmk:prop2}
Note that $\BC$ has one eigenvalue equal to 0 and all the other eigenvalues are at least 1. Hence the rank of $\BC$ is $np+m-1$. Similar to Remark~\ref{rmk:prop1}, Proposition~\ref{prop:main2} shows that the solution $(\bs_0, \{\bx_l\}_{l=1}^p)$ to~\eqref{eq:A0-2-ms} is uniquely identifiable with high probability when $mp \geq (np +m)\cdot \text{\em poly}(\log(np+m))$ and $\|\BZ^*\BZ - \BC\| \leq \frac{1}{2}.$ 
\end{remark}

\begin{proof}{[\bf Proof of Theorem~\ref{thm:main2}]}
From~\eqref{eq:A0-2-ms}, we let $\A_{\bw} = \A_{\bw,0} + \delta \A$ where $\A_{\bw,0}$ is the noiseless part of $\A_{\bw}$. By definition of $\A_{\bw, 0}$, we know that $\alpha\bz_0$ is the solution to the overdetermined system $\A_{\bw,0}\bz = \begin{bmatrix} \bzero \\c \end{bmatrix}$ where $\alpha= \frac{c}{\bw^*\bz_0}$. Now,~\eqref{eq:A0-2} gives
\begin{equation*}
\A_{\bw,0}^*\A_{\bw,0} = \A_0^*\A_0 + \bw\bw^* = \BP^*\BZ^*\BQ^*\BQ\BZ\BP + \bw\bw^*.
\end{equation*}
Define $x_{\max} := \max_{1\leq l\leq p}\|\bx_l\|$ and $x_{\min} := \min_{1\leq l\leq p}\|\bx_l\|$.
From Proposition~\ref{prop:main2} and Theorem 1 in~\cite{Stewart90}, we know that  the eigenvalues $\{\lambda_j\}_{1\leq j\leq np + m}$ of $\BZ^*\BZ$ fulfill  $\lambda_1 = 0$ and $\lambda_{j} \geq \frac{1}{2}$ for $j\geq 2$ since $\|\BZ^*\BZ - \BC\| \leq \frac{1}{2}$; and the eigenvalues of $\BC$ are 0, 1 and 2 with multiplicities 1, $np+m - 2$, $1$ respectively. 

The key is to obtain a bound for $\kappa(\A_{\bw,0}).$
From~\eqref{eq:A0-2},
\begin{eqnarray}
\lambda_{\max}(\A_{\bw,0}^*\A_{\bw,0}) & \leq & \|\BP\|^2\|\BQ\|^2 \|\BZ^*\BZ\| + \|\bw\|^2 \leq 3 x_{\max}^2 \lambda_{\max}(\BP^*\BP) + \|\bw\|^2 \nonumber \\
& \leq & 3 x_{\max}^2 \left( 1 + \frac{\|\bw\|^2}{m}\right) \lambda_{\max}(\BP^*\BP) \label{eq:model2-eigmax}
\end{eqnarray}
where $x^2_{\max} \lambda_{\max}(\BP^*\BP) \geq x^2_{\max} \frac{m}{x_{\min}^2} \geq m.$
On the other hand,~\eqref{eq:ZZ-C} gives
\begin{equation*}
\lambda_{\max}(\A_{\bw,0}^*\A_{\bw,0})\geq \max_{1\leq l\leq p}\{ \|\BA_l^*\BA_l\| \} \geq \frac{m}{2}
\end{equation*}
since $\left\|\frac{1}{m}\BA_l^*\BA_l - \I_n \right\| \leq \frac{1}{2}$. For $\lambda_{\min}(\A_{\bw,0}^*\A_{\bw,0} )$, 
\begin{equation*}
\A_{\bw,0}^*\A_{\bw,0} \succeq \lambda_{\min}(\BQ^*\BQ) \BP^*\BZ^*\BZ\BP + \bw\bw^*  \succeq  x_{\min}^2\BP^*\BZ^*\BZ\BP + \bw\bw^* =:  \BP^* \widetilde{\BC}\BP.
\end{equation*}
Denote $\bu_1 : = \frac{1}{\sqrt{2}} 
\begin{bmatrix}
\frac{1}{\sqrt{m}}\bone_m \\
\frac{1}{\sqrt{p}}\bv
\end{bmatrix}
$ such that $\BZ\bu_1 = \bzero$ and $\widetilde{\BC} = x_{\min}^2\BZ^*\BZ + \tbw\tbw^*$ where $\tbw = \BP^{-1}\bw$.
 By using the same procedure as~\eqref{eq:Aw1-low}, 
\begin{equation*}
\bu^*\widetilde{\BC}\bu \geq x_{\min}^2 \sum_{j=2}^{np+m} \lambda_j |\alpha_j|^2 + |\alpha_1|^2|\bu_1^*\BP^{-1} \bw|^2 
\end{equation*}
where $\bu : = \sum_{j=1}^{np+m} \alpha_j \bu_j$ with $\sum_j|\alpha_j|^2 = 1$ and $\lambda_j \geq \frac{1}{2}$ for $j\geq 2$ follows from Proposition~\ref{prop:main2}. Since $|\bu_1^*(\BP^{-1})^* \bw|^2 = \frac{1}{2mp} |\bw^*\bz_0|^2,$  
the smallest eigenvalue of $\widetilde{\BC}$ satisfies
\begin{align*}
\lambda_{\min}(\widetilde{\BC}) 
& \geq  \frac{x_{\min}^2}{2}\min\left\{ 1, \frac{|\bw^*\bz_0|^2}{mpx^2_{\min}} \right\} \nonumber \\
& \geq \frac{x^2_{\min}}{2} \min\left\{ 1, \frac{1}{m}\|\bw\|^2|\Corr(\bw, \bz_0)|^2 \right\}
\end{align*}
where $\frac{|\bw^*\bz_0|^2}{px^2_{\min}} \geq \frac{|\bw^*\bz_0|^2}{\|\bz_0\|^2} \geq \|\bw\|^2|\Corr(\bw,\bz_0)|^2$ follows from $\|\bz_0\|^2\geq px^2_{\min}$.

Therefore, the smallest eigenvalue of $\A^*_{\bw,0}\A_{\bw,0}$ satisfies
\begin{align}
\lambda_{\min}(\A_{\bw,0}^*\A_{\bw,0}) & \geq \lambda_{\min}(\widetilde{\BC})\lambda_{\min}(\BP^*\BP) \nonumber \\
& \geq \frac{x^2_{\min}}{2m} \min\left\{ m, \|\bw\|^2|\Corr(\bw, \bz_0)|^2 \right\} \lambda_{\min}(\BP^*\BP) \label{eq:model2-eigmin}
\end{align}
Combining~\eqref{eq:model2-eigmax} and~\eqref{eq:model2-eigmin} leads to
\begin{equation*}
\kappa(\A_{\bw,0}^*\A_{\bw,0}) \leq \frac{6 x_{\max}^2(m + \|\bw\|^2)}{x^2_{\min}\min\{ m, \|\bw\|^2|\Corr(\bw,\bz_0)|^2 \}} \kappa(\BP^*\BP).
\end{equation*}
Applying Proposition~\ref{prop:perturb} and $\eta = \frac{2\|\delta \A\|}{\sqrt{m}} \geq \frac{\|\delta\A\|}{\|\A_{\bw,0}\|}$, we have
\begin{equation*}
\frac{ \| \hat{\bz} - \alpha\bz_0 \| }{\|\alpha\bz_0\|} \leq \kappa(\A_{\bw, 0})\eta \left( 1 + \frac{2}{1 - \kappa(\A_{\bw, 0})\eta }\right), \quad \alpha = \frac{c}{\bw^*\bz_0}
\end{equation*}
if $\kappa(\A_{\bw, 0})\eta < 1$ where $\kappa(\A_{\bw,0})$ obeys
\begin{equation*}
\kappa(\A_{\bw,0}) \leq  \sqrt{\frac{6 x_{\max}^2(m + \|\bw\|^2)}{x^2_{\min}\min\{ m, \|\bw\|^2|\Corr(\bw,\bz_0)|^2 \}}  \frac{ \max\{ pd_{\max}^2, \frac{m}{x^2_{\min}} \} }{ \min\{ pd_{\min}^2, \frac{m}{x^2_{\max}} \}}}.
\end{equation*}
In particular, if $\|\bw\| = \sqrt{m}$, then $\kappa(\A_{\bw,0})$ has the following simpler upper bound:
\begin{equation*}
\kappa(\A_{\bw,0}) \leq  \frac{2\sqrt{3}x_{\max}}{|\text{Corr}(\bw,\bz_0)|x_{\min}}\sqrt{  \frac{\max\{ pd^2_{\max}, \frac{m}{x^2_{\min}} \} }{ \min\{pd^2_{\min}, \frac{m}{x^2_{\max}}\} }}
\end{equation*}
which finishes the proof of Theorem~\ref{thm:main2}.
\end{proof}

\subsubsection{Proof of Proposition~\ref{prop:main2}(a)}
\label{s:model2-A}
In this section, we will prove that Proposition~\ref{prop:main2}(a) if $\ba_{l,i}\sim \frac{1}{\sqrt{2}}\mathcal{N}(\bzero,\I_n) +\frac{\mi}{\sqrt{2}}\mathcal{N}(\bzero, \I_n)$ where $\ba_{l,i}\in\CC^n$ is the $i$-th column of $\BA_l^*$. Before moving to the proof, we compute a few quantities which will be used later.
Define $\bz_{l,i}$ as the $((l-1)m + i)$-th column of $\BZ^*,$
\begin{equation*}
\bz_{l,i} := 
\begin{bmatrix}
\frac{1}{\sqrt{p}}\overline{\lag \ba_{l,i}, \bv_l\rag} \be_i \\
-\frac{1}{\sqrt{m}} \tbe_l \otimes \ba_{l,i} 
\end{bmatrix}_{(np+m)\times 1}, \quad 1\leq l\leq p, \quad 1\leq i\leq m
\end{equation*}
where $\{\be_i\}_{i=1}^m\in\RR^m$ and $\{\tbe_l\}_{l=1}^p\in \RR^p$ are standard orthonormal basis in $\RR^m$ and $\RR^p$ respectively; ``$\otimes$" denotes Kronecker product. By definition, we have
$\BZ^*\BZ = \sum_{l=1}^p \sum_{i=1}^m \bz_{l,i}\bz_{l,i}^*$
and all $\bz_{l,i}$ are independent from one another.   

\begin{equation*}
\bz_{l,i}\bz_{l,i}^* = 
\begin{bmatrix}
\frac{1}{p}|\lag \ba_{l,i}, \bv_l\rag|^2 \be_i\be_i^* & -\frac{1}{\sqrt{mp}}\lag \bv_l, \ba_{l,i}\rag \be_i (\tbe_l^* \otimes\ba_{l,i}^*) \\ 
-\frac{1}{\sqrt{mp}}\lag \ba_{l,i}, \bv_l\rag (\tbe_l\otimes \ba_{l,i})\be_i^* & \frac{1}{m} \tbe_l\tbe_l^*\otimes \ba_{l,i}\ba_{l,i}^*
\end{bmatrix}
\end{equation*}
and its expectation is equal to
\begin{equation*}
\E(\bz_{l,i}\bz_{l,i}^*) = 
\begin{bmatrix}
\frac{1}{p}\be_i\be_i^* & -\frac{1}{\sqrt{mp}}\be_i (\tbe_l^* \otimes\bv_l^*) \\
-\frac{1}{\sqrt{mp}}(\tbe_l\otimes \bv_l)\be_i^* & \frac{1}{m}\tbe_l\tbe_l^*\otimes \I_n
\end{bmatrix}.
\end{equation*}
It is easy to verify that $\BC = \sum_{l=1}^p\sum_{i=1}^m \E(\bz_{l,i}\bz_{l,i}^*)$.

\begin{proof}{[\bf Proof of  Proposition~\ref{prop:main2}(a)]}
The key here is to use apply the matrix Bernstein inequality in Theorem~\ref{BernGaussian}.
Note that $\BZ^*\BZ = \sum_{l=1}^p\sum_{i=1}^m \bz_{l,i}\bz_{l,i}^*.$
Let $\CZ_{l,i} : =\bz_{l,i}\bz_{l,i}^* - \E(\bz_{l,i}\bz_{l,i}^*)$ and we have
\begin{equation*}
\|\CZ_{l,i}\| \leq \|\bz_{l,i}\|^2 + \|\E(\bz_{l,i}\bz_{l,i}^*)\| \leq \frac{1}{p}|\lag \ba_{l,i}, \bv_l\rag|^2 +\frac{1}{m} \|\ba_{l,i}\|^2 + 2\max\left\{\frac{1}{p}, \frac{1}{m}\right\}
\end{equation*}
since $\|\E(\bz_{l,i}\bz_{l,i}^*)\| \leq 2\max\{\frac{1}{p}, \frac{1}{m}\}$ follows from Lemma~\ref{lem:pos}.
Therefore, the exponential norm of $\|\CZ_{l,i}\|$ is bounded by
\begin{equation*}
\|\CZ_{l,i}\|_{\psi_1} \leq 2\left( \frac{1}{p}(|\lag \ba_{l,i}, \bv_l\rag|^2)_{\psi_1} + \frac{1}{m} (\|\ba_{l,i}\|^2)_{\psi_1} \right) + 2\max\left\{\frac{1}{p}, \frac{1}{m}\right\}  \leq  C\left( \frac{1}{p} + \frac{n}{m}\right)
\end{equation*}
which follows from Lemma~\ref{lem:psi1} and as a result $R : = \max_{l,i}\|\CZ_{l,i}\|_{\psi_1} \leq C\left( \frac{1}{p} + \frac{n}{m}\right).$

Now we proceed by estimating the variance $\sigma_0^2 : = \left\| \sum_{l=1}^p\sum_{i=1}^m \CZ_{l,i}\CZ_{l,i}^* \right\|$. We express $\CZ_{l,i}$ as follows: 
\begin{equation*}
\CZ_{l,i} = 
\begin{bmatrix}
\frac{1}{p}(|\lag \ba_{l,i}, \bv_l\rag|^2 - 1) \be_i\be_i^* & -\frac{1}{\sqrt{mp}} \be_i (\tbe_l^* \otimes (\bv_l^*(\ba_{l,i}\ba_{l,i}^* - \I_n))) \\ 
-\frac{1}{\sqrt{mp}}(\tbe_l \otimes ((\ba_{l,i}\ba_{l,i}^* - \I_n)\bv_l))\be_i^* & \frac{1}{m} \tbe_l\tbe_l^*\otimes (\ba_{l,i}\ba_{l,i}^* - \I_n)
\end{bmatrix}.
\end{equation*}
The $(1,1)$-th and the $(2,2)$-th block of $\CZ_{l,i}\CZ_{l,i}^*$ are 
\begin{equation*}
(\CZ_{l,i}\CZ_{l,i}^*)_{1,1} = 
\frac{1}{p^2} \left(|\lag \ba_{l,i}, \bv_l\rag|^2 - 1 \right)^2 \be_i\be_i^* + \frac{1}{mp}\bv_l^* (\ba_{l,i}\ba_{l,i}^* - \I_n)^2 \bv_l \be_i \be_i^*,
\end{equation*}
and
\begin{equation*}
(\CZ_{l,i}\CZ_{l,i}^*)_{2,2} =  \tbe_l\tbe_l^* \otimes \left[ \frac{1}{mp}(\ba_{l,i}\ba_{l,i}^* - \I_n) \bv_l\bv_l^*(\ba_{l,i}\ba_{l,i}^* - \I_n) + \frac{1}{m^2}(\ba_{l,i}\ba_{l,i}^* - \I_n)^2\right].  
\end{equation*}
Following from~\eqref{eq:gauss-3}, ~\eqref{eq:gauss-4} and~\eqref{eq:gauss-5}, we have 
\begin{equation*}
\E(\CZ_{l,i}\CZ_{l,i}^*)_{1,1} = \left( \frac{1}{p^2} + \frac{n}{mp}\right)\be_i\be_i^*, \quad \E(\CZ_{l,i}\CZ_{l,i}^*)_{2,2} = \tbe_l\tbe_l^*\otimes \left(\frac{1}{mp}\I_n + \frac{n}{m^2}\I_n\right).
\end{equation*}
Due to Lemma~\ref{lem:pos}, there holds,
\begin{eqnarray*}
\sigma^2_0 &: = & \left\| \sum_{l=1}^p\sum_{i=1}^m \CZ_{l,i}\CZ_{l,i}^* \right\| \leq 2\left\| \sum_{l=1}^p\sum_{i=1}^m
\begin{bmatrix}
\E(\CZ_{l,i}\CZ_{l,i}^*)_{1,1} & \bzero \\
\bzero & \E(\CZ_{l,i}\CZ_{l,i}^*)_{2,2}
 \end{bmatrix} 
\right\| \\
& = & 2\left\| 
\begin{bmatrix}
\left(\frac{1}{p} + \frac{n}{m}\right)\I_m & \bzero \\
\bzero & \left(\frac{1}{p} + \frac{n}{m}\right)\I_{np}
\end{bmatrix}
\right\| \leq 2\left(\frac{1}{p} + \frac{n}{m}\right).
\end{eqnarray*}
 Note that $\sigma_0^2\geq \| \sum_{l=1}^p\sum_{i=1}^m \E(\CZ_{l,i}\CZ_{l,i}^*)_{1,1} \| = \frac{1}{p} + \frac{n}{m}$ since $\E(\CZ_{l,i}\CZ_{l,i}^*)$ is a positive semi-definite matrix. By applying~\eqref{thm:bern2},
\begin{eqnarray*}
\|\BZ^*\BZ - \BC\| & = &
\left\|\sum_{l=1}^p\sum_{i=1}^m \CZ_{l,i}\right\| = \left\| \sum_{l=1}^p\sum_{i=1}^m \E(\bz_{l,i}\bz_{l,i}^*) - \BC\right\|\\
& \leq & C_0\max\Big\{ \sqrt{\frac{1}{p} + \frac{n}{m}} \sqrt{t + \log(np + m)}, \\ 
&& \left( \frac{1}{p} + \frac{n}{m}\right)(t + \log(np + m))\log(np + m) \Big\}  \leq \frac{1}{2}.
\end{eqnarray*}
With $t = \gamma\log(mp + n)$, there holds $
\left\|\BZ^*\BZ - \BC\right\| \leq \frac{1}{2}$
with probability at least $1 - (np +m)^{-\gamma}$ if 
$C_0\left( \frac{1}{p} + \frac{n}{m}\right)(\gamma + 1)\log^2(np + m)  \leq \frac{1}{4}.$
\end{proof}

\subsubsection{Proof of Proposition~\ref{prop:main2}(b)}
\label{s:model2-B}
We prove Proposition~\ref{prop:main2} based on assumption~\eqref{eq:model2-Al}. Denote $\ba_{l,i}$ and $\bh_{l,i}$ as the $i$-th column of $\BA_l^*$ and $\BH_l^*$ and obviously we have $\ba_{l,i} = \BM_l\bh_{l,i}.$
Denote $\BLam_l = \diag(\overline{\BA_l\bv_l})$ and let $\BZ_l$ be the $l$-th block of $\BZ^*$ in~\eqref{eq:A0-2}, i.e., 
\begin{equation*}
\BZ_l = 
\begin{bmatrix}
\frac{1}{\sqrt{p}}\BLam_l \\
-\frac{1}{\sqrt{m}}\tbe_l\otimes \BA_l^*\\
\end{bmatrix} \in\CC^{(np + m)\times m}.
\end{equation*}
With $\BA_l^*\BA_l = m\I_n$, we have
\begin{equation*}
\BZ_l\BZ_l^* = 
\begin{bmatrix}
\frac{1}{p}\BLam_l\BLam^*_l & -\frac{1}{\sqrt{mp}}\widetilde{\be}^*_l\otimes(\BLam_l\BA_l) \\
-\frac{1}{\sqrt{mp}}\widetilde{\be}_l\otimes(\BLam_l\BA_l)^* & (\tbe_l\tbe_l^*)\otimes \I_n
\end{bmatrix}
\in \CC^{(np + m)\times (np + m)} 
\end{equation*} 
 where the expectation of $i$-th row of $\BLam_l\BA_l$ yields $\E(\BLam_l\BA_l)_i = \E( \bv_l^* \ba_{l,i}\ba_{l,i}^* ) = \bv_l^*$. Hence $\E(\BLam_l \BA_l) = \bone_m\bv_l^*\in \CC^{m\times n}.$ Its expectation equals
\begin{equation*}
\E(\BZ_l\BZ_l^*) = 
\begin{bmatrix}
\frac{1}{p}\I_m & - \frac{1}{\sqrt{mp}}\tbe_l^*\otimes (\bone_m \bv_l^*) \\
-\frac{1}{\sqrt{mp}}\tbe_l \otimes (\bv_l\bone_m^*) & \tbe_l\tbe_l^*\otimes \I_n
\end{bmatrix}.
\end{equation*}

\begin{proof}{ [\bf Proof of Proposition~\ref{prop:main2}(b)] }
Note that each block $\BZ_l$ is independent and we want to apply the matrix Bernstein inequality to achieve the desired result. Let $\CZ_l : = \BZ_l\BZ_l^* - \E(\BZ_l\BZ_l^*)$ and by definition, we have 
\begin{eqnarray*}
\|\CZ_l\| 
& = & \left\|
\begin{bmatrix}
\frac{1}{p} (\BLam_l\BLam^*_l - \I_m)  & -\frac{1}{\sqrt{mp}}(\BLam_l\BA_l - \bone_m \bv_l^*) \\
-\frac{1}{\sqrt{mp}} (\BLam_l\BA_l - \bone_m \bv_l^*)^* & \bzero
\end{bmatrix}
\right\| \\
& \leq & \frac{1}{p} \|\BLam_l\BLam_l^* - \I_m\| + \frac{1}{\sqrt{mp}} \|\BLam_l\BA_l - \bone_m\bv_l^*\|.
\end{eqnarray*}
Note that 
\begin{equation*}
\|\BLam_l\BA_l - \bone_m\bv_l^*\| \leq \|\BLam_l\| \|\BA_l\| + \|\bone_m\bv_l^*\| \leq  \sqrt{m} \|\BLam_l\| + \sqrt{m}.
\end{equation*}
Since~\eqref{eq:BLam} implies that $
\Pr\left(\max_{1\leq l\leq p}\|\BLam_{l}\| \geq \sqrt{2\gamma\log (mp)} \right) \leq 2 (mp)^{-\gamma +1}$, we have
\begin{equation*}
R: = \max_{1\leq l\leq p}\|\CZ_l \| \leq \frac{\|\BLam_l\|^2 + 1}{p} + \frac{\|\BLam_l\| + 1}{\sqrt{p}}\leq \frac{2\gamma\log(mp) + 1}{p} + \frac{\sqrt{2\gamma \log(mp)} + 1}{\sqrt{p}}
\end{equation*}
with probability at least $1 - 2(mp)^{-\gamma+ 1}$.
We proceed with estimation of $\sigma_0^2 := \left\|\sum_{l=1}^p \E(\CZ_l\CZ_l^*)\right\|$ by looking at the $(1,1)$-th and $(2,2)$-th block of $\CZ_l\CZ_l^*$, i.e.,  
\begin{eqnarray*}
(\CZ_l\CZ_l^*)_{1,1} & = & 
\frac{1}{p^2} (\BLam_l\BLam^*_l - \I_m)^2 + \frac{1}{mp}(\BLam_l\BA_l - \bone_m \bv_l^*)(\BLam_l\BA_l - \bone_m \bv_l^*)^*, \\
(\CZ_l\CZ_l^*)_{2,2} & = &  \frac{1}{mp} (\tbe_l\tbe_l^*)\otimes((\BLam_l\BA_l - \bone_m \bv_l^*)^*(\BLam_l\BA_l - \bone_m \bv_l^*)).
\end{eqnarray*}
Note that $\E(\BLam_l\BLam^*_l - \I_m)^2 =  \E((\BLam_l \BLam^*_l)^2) - \I_m$.
The $i$-th diagonal entry of $(\BLam_l \BLam^*_l)^2$ is $|\lag \ba_{l,i}, \bv_l\rag|^4 = |\lag \bsm_l, \diag(\bar{\bh}_{l,i})\bv_l\rag|^4 $
where $\ba_{l,i} = \BM_l \bh_{l,i} = \diag(\bh_{l,i})\bsm_l$ and~\eqref{eq:binary-4} implies $\E(|\lag \ba_{l,i}, \bv_l\rag|^4 ) \leq 3$ since $\diag(\bar{\bh}_{l,i})\bv_l$ is still a unit vector (note that $\diag(\bar{\bh}_{l,i})$ is unitary since $\bh_{l,i}$ is a column of $\BH_l^*$). Therefore,
\begin{equation}\label{eq:div-sig-1}
\E(\BLam_l\BLam^*_l - \I_m)^2 =  \E((\BLam_l \BLam^*_l)^2) - \I_m \preceq 3\I_m - \I_m \preceq 2\I_m.
\end{equation}
By using Lemma~\ref{lem:LAAL}, we have
\begin{eqnarray}
\E(\BLam_l\BA_l - \bone_m \bv_l^*)(\BLam_l\BA_l - \bone_m \bv_l^*)^* & = & \E(\BLam_l\BA_l\BA_l^* \BLam_l^*) - \bone_m\bone_m^* \nonumber \\
& =& \frac{(n - 1)(m\I_m - \bone_m\bone_m^*)}{m-1} \preceq \frac{m(n-1)\I_m}{m-1}. \label{eq:div-sig-2}
\end{eqnarray}
With $\ba_{l,i} = \BM_l\bh_{l,i}$ and independence between $\{\bh_{l,i}\}_{i=1}^m$ and $\BM_l$, we have
\begin{align}
& \E( (\BLam_l\BA_l - \bone_m \bv_l^*)^*(\BLam_l\BA_l - \bone_m \bv_l^*)) \nonumber \\
& \qquad\qquad = \E(\BA_l^* \BLam_l^*\BLam_l \BA_l) - m\bv_l\bv_l^* \nonumber \\
& \qquad\qquad = \E\left( \sum_{i=1}^m |\lag \ba_{l,i}, \bv_l\rag|^2 \ba_{l,i}\ba_{l,i}^*\right) - m\bv_l\bv_l^* \nonumber \\
& \qquad\qquad= \E\left(\sum_{i=1}^m \diag(\bh_{l,i})\left( |\lag \bsm_l, \diag(\bar{\bh}_{l,i})\bv_l\rag|^2 \bsm_l\bsm_l^*\right) \diag(\bar{\bh}_{l,i})\right) \nonumber  - m\bv_l\bv_l^* \nonumber \\
& \qquad\qquad \preceq 3m\I_n - m\bv_l\bv_l^* \preceq 3m\I_n \label{eq:div-sig-3}
\end{align}
where $\E |\lag \bsm_l, \diag(\bar{\bh}_{l,i})\bv_l\rag|^2 \bsm_l\bsm_l^* \preceq 3\I_n$ follows from~\eqref{eq:binary-3} and $\diag(\bh_{l,i})\diag(\bar{\bh}_{l,i}) = \I_n.$
By using~\eqref{eq:div-sig-1},~\eqref{eq:div-sig-2},~\eqref{eq:div-sig-3} and Lemma~\ref{lem:pos}, $\sigma_0^2$ is bounded above by
\begin{eqnarray*}
\sigma_0^2 & : = & \left\| \E\left(\sum_{l=1}^p \CZ_l\CZ_l^*\right) \right\| \leq 2\left\| \sum_{l=1}^p
\begin{bmatrix}
\E(\CZ_l\CZ_l^*)_{1,1} & \bzero \\
\bzero & \E(\CZ_l\CZ_l^*)_{2,2} 
\end{bmatrix}
\right\|\\
& \leq & 2\left\| \sum_{l=1}^p
\begin{bmatrix}
\left(\frac{2}{p^2} + \frac{n - 1}{(m-1)p}\right)\I_m & \bzero \\
\bzero & \frac{3}{p}(\tbe_l\tbe_l^*)\otimes \I_n
\end{bmatrix}
\right\| \\
& \leq & 2\left\| 
\begin{bmatrix}
\left(\frac{2}{p} + \frac{n - 1}{m-1}\right)\I_m & \bzero \\
\bzero & \frac{3}{p}\I_{np}
\end{bmatrix}
\right\| \leq 6\left(\frac{1}{p} + \frac{n-1}{m-1}\right).
\end{eqnarray*}
Conditioned on the event $\left\{\max_{1\leq l\leq p}\|\BLam_{l}\| \geq \sqrt{2\gamma\log (mp)} \right\}$, applying~\eqref{thm:bern} with $t = \gamma\log(np + m)$ gives 
\begin{eqnarray*}
\left\| \sum_{l=1}^p \CZ_l\right\| & \leq & C_0 \max\Big\{ \sqrt{\frac{1}{p} + \frac{n-1}{m-1}} \sqrt{(\gamma + 1) \log(np + m)}, \\
&& (\gamma + 1) \sqrt{\frac{2\gamma \log(mp)}{p}} \log(mp)\log(np + m) \Big\} \leq \frac{1}{2}
\end{eqnarray*}
with probability at least $1 - (np + m)^{-\gamma} - 2 (mp)^{-\gamma +1}$ and it suffices to require the condition \\
$C_0\left(\frac{1}{p} + \frac{n-1}{m-1}\right)\gamma^3 \log^4(np + m) \leq \frac{1}{4}.$
\end{proof}

\subsection{Blind Calibration from multiple snapshots}
Recall that $\A_{\bw}$ in~\eqref{eq:Amodel3} and the only difference from~\eqref{eq:A0-2-ms} is that here all $\BA_l$ are equal to $\BA$. 
If $\beps_l = \bzero$,  $\A_{\bw}$ (excluding the last row) can be factorized into
\begin{equation}
\A_0 := 
\underbrace{
\begin{bmatrix}
\|\bx_1\|\I_m & \cdots & \bzero \\
\vdots & \ddots & \vdots \\
\bzero &  \cdots &\|\bx_p\|\I_m
\end{bmatrix}}_{\BQ}
\underbrace{
\begin{bmatrix}
\frac{\diag(\BA\bv_1)}{\sqrt{p}} & -\frac{\BA}{\sqrt{m}} & \cdots & \bzero \\
\vdots & \vdots  & \ddots &  \vdots \\
\frac{\diag(\BA\bv_p)}{\sqrt{p}}  & \bzero & \cdots &  -\frac{\BA}{\sqrt{m}}
\end{bmatrix}
}_{ \BZ\in\CC^{mp\times (np + m)}}
\underbrace{
\begin{bmatrix}
\sqrt{p}\BD & \bzero & \cdots & \bzero \\
\bzero & \frac{\sqrt{m}\I_n}{\|\bx_1\|} & \cdots & \bzero \\
\vdots & \vdots & \ddots & \vdots \\
\bzero & \bzero &  \cdots & \frac{\sqrt{m}\I_n}{\|\bx_p\|}
\end{bmatrix}}_{\BP}
\end{equation}
where $\bv_i = \frac{\bx_i}{\|\bx_i\|}\in\CC^n$ is the normalized $\bx_i$. 

Before we proceed to the main result in this section we need to introduce some notation.
Let $\ba_{i}$ be the $i$-th column of $\BA^*$, which is a complex Gaussian random matrix; define $\BZ_i\in \CC^{(np+m)\times p}$ to be a matrix whose columns consist of the $i$-th column of each block of $\BZ^*$, i.e., 
\begin{equation*}
\BZ_i = 
\begin{bmatrix}
\frac{1}{\sqrt{p}}\overline{\lag \ba_{i}, \bv_1\rag} \be_i & \cdots & \frac{1}{\sqrt{p}}\overline{\lag \ba_{i}, \bv_p\rag} \be_i \\
-\frac{1}{\sqrt{m}} \tbe_1 \otimes \ba_{i}  & \cdots & -\frac{1}{\sqrt{m}} \tbe_p \otimes \ba_{i} 
\end{bmatrix}_{(np+m)\times p}
\end{equation*}
where ``$\otimes$" denotes Kronecker product and both $\{\be_i\}_{i=1}^m\in\RR^m$ and $\{\tbe_l\}_{l=1}^p\in \RR^p$ are the standard orthonormal basis in $\RR^m$ and $\RR^p$, respectively. In this way, the $\BZ_i$ are independently from one another. By definition, 
\begin{eqnarray*}
\BZ_i\BZ_i^*
& = & 
\begin{bmatrix}
\frac{1}{p}\sum_{l=1}^p |\lag \ba_i, \bv_l\rag|^2 \be_i\be_i^* & -\frac{1}{\sqrt{mp}}\be_i\bv^* (\I_p\otimes \ba_i\ba_i^*) \\
-\frac{1}{\sqrt{mp}} (\I_p\otimes \ba_i\ba_i^*) \bv\be_i^*  & \frac{1}{m} \I_p \otimes (\ba_i\ba_i^*)
\end{bmatrix}
\end{eqnarray*}
where $\sum_{l=1}^p \tbe_l\otimes\bv_l = \bv\in \CC^{np\times 1}$ and $\bv = \begin{bmatrix} \bv_1\\ \vdots \\ \bv_p \end{bmatrix}$ with $\|\bv\| = \sqrt{p}.$ 
The expectation of $ \BZ_i\BZ_i^*$ is given by
\begin{equation*}
\E(\BZ_i\BZ_i^* ) = 
\begin{bmatrix}
\be_i\be_i^* & -\frac{1}{\sqrt{mp}}\be_i \bv^* \\
-\frac{1}{\sqrt{mp}}\bv\be_i^* & \frac{1}{m}\I_{np}
\end{bmatrix}, \quad
\BC: = \sum_{i=1}^m \E(\BZ_i\BZ_i^*)  = 
\begin{bmatrix}
\I_m & -\frac{1}{\sqrt{mp}}\bone_m\bv^* \\
-\frac{1}{\sqrt{mp}}\bv\bone_m^* & \I_{np}
\end{bmatrix}.
\end{equation*}

\vskip0.25cm
Our analysis depends on the~\emph{mutual coherence} of $\{\bv_l\}_{l=1}^p$. One cannot expect to recover all $\{\bv_l\}_{l=1}^p$ and $\BD$ if all $\{\bv_l\}_{l=1}^p$ are parallel to each other. 
Let $\BG$ be the Gram matrix of $\{\bv_l\}_{l=1}^p$ with $p\leq n$, i.e., $\BG\in \CC^{p\times p}$ and $G_{k,l} = \lag \bv_k, \bv_l\rag, 1\leq k\leq l\leq p$ and in particular, $G_{l,l} = 1, 1\leq l\leq p$.  Its eigenvalues are denoted by $\{\lambda_l\}_{l=1}^p$ with $\lambda_p\geq \cdots \geq \lambda_1\geq 0$. Basic linear algebra tells that
\begin{equation}\label{eq:Glambda}
\sum_{l=1}^p\lambda_l = p, \quad \BU\BG\BU^* =\BLam,
\end{equation}
where $\BU\in \CC^{p\times p}$ is unitary and $\BLam = \diag(\lambda_1, \cdots, \lambda_p).$ 
Let $\BV = \begin{bmatrix}\bv_1^* \\ \vdots \\ \bv^*_p \end{bmatrix} \in \CC^{p\times n}$, then there holds $\BLam = \BU \BG \BU^* = \BU\BV (\BU\BV)^*$ since $\BG = \BV \BV^*.$ Here $1\leq \|\BG\| \leq \sqrt{p}$ and $\sqrt{p} \leq \|\BG\|_F \leq p$.
In particular, if $\lag \bv_k, \bv_l\rag = \delta_{kl}$, then $\BG = \I_p$; if $|\lag \bv_k, \bv_l\rag| = 1$ for all $1\leq k,l\leq p$, then $\|\BG\| = \sqrt{p}$ and $\|\BG\|_F = p.$

We are now ready to state and prove the main result in this subsection.
\begin{proposition}\label{prop:main3}
There holds
\begin{equation*}
\|\BZ^*\BZ - \BC\| = \left\|\sum_{i=1}^m \BZ_i\BZ_i^* - \BC\right\| \leq \frac{1}{2}
\end{equation*}
with probability at least $1 - 2m(np + m)^{-\gamma}$ if
\begin{equation}\label{eq:cond-mult}
C_0\left(\max\left\{\frac{\|\BG\|}{p}, \frac{\|\BG\|_F^2}{p^2} \right\} + \frac{n}{m} \right)\log^2(np + m) \leq \frac{1}{16(\gamma + 1)}
\end{equation}
and each $\ba_l$ is i.i.d. complex Gaussian, i.e., $\ba_{i} \sim \frac{1}{\sqrt{2}}\mathcal{N}(\bzero, \I_n) + \frac{\mi}{\sqrt{2}}  \mathcal{N}(\bzero, \I_n)$.
In particular, if $\BG = \I_p$ and $\|\BG\|_F = \sqrt{p}$,~\eqref{eq:cond-mult} becomes
\begin{equation}
C_0\left( \frac{1}{p} + \frac{n}{m} \right)\log^2(np + m) \leq \frac{1}{16(\gamma + 1)}.
\end{equation}

\end{proposition}
\begin{remark}
The proof of Theorem~\ref{thm:main3}  follows exactly from that of Theorem~\ref{thm:main2} when Proposition~\ref{prop:main3} holds. Hence we just give a proof of Proposition~\ref{prop:main3}.
\end{remark}

\begin{proof}{[\bf Proof of Proposition~\ref{prop:main3}]}
Let $\BZ_i\BZ_i^* - \E(\BZ_i\BZ_i^*) =: \CZ_{i,1} + \CZ_{i,2}$, where $\CZ_{i,1}$ and $\CZ_{i,2}$ are defined as 
\begin{eqnarray*}
\CZ_{i,1} & := &
\begin{bmatrix}
\frac{1}{p}\sum_{l=1}^p (|\lag \ba_i, \bv_l\rag|^2 - 1) \be_i\be_i^* & \bzero \\
\bzero  & \bzero
\end{bmatrix}, \\
\CZ_{i,2} & := &  
\begin{bmatrix}
\bzero & -\frac{1}{\sqrt{mp}}\be_i\bv^* (\I_p\otimes (\ba_i\ba_i^* - \I_n)) \\
-\frac{1}{\sqrt{mp}} (\I_p\otimes (\ba_i\ba_i^* - \I_n)) \bv\be_i^*  & \frac{1}{m} \I_p \otimes (\ba_i\ba_i^* - \I_n)
\end{bmatrix}.
\end{eqnarray*}

\paragraph{Estimation of $\|\sum_{i=1}^m\CZ_{i,1}\|$}
Following from~\eqref{eq:Glambda}, we have 
\begin{equation*}
\sum_{l=1}^p |\lag \ba_i, \bv_l\rag|^2 = \| \BV\ba_i\|^2 = \| \BU\BV\ba_i\|^2=  \| \BLam^{1/2}\BLam^{-1/2}\BU\BV\ba_i\|^2 =\frac{1}{2} \sum_{l=1}^p  \lambda_l\xi_{i,l}^2
\end{equation*}
where $\BLam^{-1/2}\BU\BV$ is a $p\times n$ matrix with orthonormal rows and hence each $\xi_{i,l}$ is Rayleigh distributed. (We say $\xi$ is Rayleigh distributed if $\xi = \sqrt{X^2 + Y^2}$ where both $X$ and $Y$ are standard real Gaussian variables.)

Due to the simple form of $\CZ_{i,1}$, it is easy to see from Bernstein's inequality for scalar random variables (See Proposition 5.16 in~\cite{Ver12}) that
\begin{equation}\label{eq:czi1}
\left\| \sum_{i=1}^m \CZ_{i,1}\right\| 
\leq \max_{1\leq i\leq m}\left|\frac{1}{p}\sum_{l=1}^p |\lag \ba_i, \bv_l\rag|^2   - 1\right| \leq C_0\max\left\{ \frac{t\|\BG\|}{p}, \frac{\sqrt{t}\|\BG\|_F}{p} \right\}
\end{equation}
with probability $1 - me^{-t}.$ Here $\left|\frac{1}{p}\sum_{l=1}^p |\lag \ba_i, \bv_l\rag|^2   - 1\right| = \left|\frac{1}{p}\sum_{l=1}^p \lambda_l(\xi^2_{i,l} - 1) \right|$ where $\sum_{l=1}^p\lambda_l = p$. Therefore, 
\begin{eqnarray*}
\Var\left(\sum_{l=1}^p \lambda_l(\xi^2_{i,l} - 1) \right) & \leq & \Var(\xi^2_{i,1} - 1)\sum_{l=1}^p\lambda_l^2 = c_0 \| \BG\|_F^2,  \\
\max_{1\leq l\leq p} (\lambda_l |\xi^2_{i,l} - 1|)_{\psi_1 } & \leq & c_1 \max_{1\leq l\leq p}\lambda_l = c_1 \|\BG\|. 
\end{eqnarray*}
Now we only need to find out $\|\CZ_{i,2}\|_{\psi_1}$ and $\sigma_0^2 = \|\sum_{i=1}^m \E(\CZ_{i,2}\CZ_{i,2}^*)\|$ in order to bound $\|\sum_{i=1}^m \CZ_{i,2}\|$.

\paragraph{Estimation of $\|\CZ_{i,2} \|_{\psi_1}$}
Denote 
$\bz_i :=  (\I_p\otimes (\ba_i\ba_i^* - \I)) \bv$ and there holds, 
\begin{eqnarray}
\|\CZ_{i,2} \| 
& \leq & \frac{1}{m} \max\{\|\ba_i\|^2, 1\} + \frac{1}{\sqrt{mp}} \|\bz_{i}\be_i^*\| \nonumber = \frac{1}{m} \max\{\|\ba_i\|^2, 1\} + \frac{1}{\sqrt{mp}} \|\bz_{i}\| \label{eq:CZi2}.
\end{eqnarray}
For $\|\ba_{i}\|^2$, its $\psi_1$-norm is bounded by $C_0 n$ due to Lemma~\ref{lem:psi1} and we only need to know $\|\bz_i\|_{\psi_1}:$
\begin{eqnarray*}
\| \bz_i \| 
& = & \|(\I_p\otimes (\ba_i\ba_i^* - \I_n)) \bv\|  \leq  \|\I_p\otimes (\ba_i\ba_i^*) \bv \| + \|\bv\|  = \sqrt{ \sum_{l=1}^p |\lag \ba_i, \bv_l\rag|^2}  \| \ba_i\| +  \sqrt{p}.
\end{eqnarray*}
Let $u = \sum_{l=1}^p |\lag \ba_i, \bv_l\rag|^2 = \sum_{l=1}^p\lambda_l \xi_{i,l}^2 $ and $v = \|\ba_i\|^2 $.  The $\psi_1$-norm of $\|\I_p \otimes (\ba_i\ba_i^*)\bv\|$ satifies
\begin{eqnarray*}
(\sqrt{uv})_{\psi_1}
& = & \sup_{q\geq 1} q^{-1}(\E(\sqrt{uv})^q)^{1/q} \leq  \sup_{q\geq 1} q^{-1}(\E u^q)^{1/2q}  (\E v^q)^{1/2q} \\
& \leq & \sqrt{ \sup_{q \geq 1} q^{-1} (\E u^q)^{1/q} \sup_{q \geq 1} q^{-1} (\E v^q)^{1/q} }  \leq  C_1\sqrt{np}
\end{eqnarray*}
where the second inequality follows from the Cauchy-Schwarz inequality, $\|u\|_{\psi_1} \leq C_1 p$, and $\|v\|_{\psi_1} \leq C_1n$.
Therefore, $\|\bz_i\|_{\psi_1} \leq C_2\sqrt{np}$ and there holds
\begin{equation}\label{eq:Rpsi-1}
\| \CZ_{i,2} \|_{\psi_1} \leq C_0\left(\frac{n}{m} + \frac{\sqrt{np}}{\sqrt{mp}}  \right) \leq C_0\left(\frac{n}{m} + \sqrt{\frac{n}{m}} \right).
\end{equation}

\paragraph{Estimation of $\sigma_0^2$} Note that
$\sigma_0^2 
:=  \left\| \sum_{i=1}^m  \E(\CZ_{i,2}\CZ_{i,2}^*) \right\|$.
Let $\BA_{i,1,1}$ and $\BA_{i,2,2}$ be the $(1,1)$-th and $(2,2)$-th block of $\CZ_{i,2}\CZ_{i,2}^*$ respectively, i.e., 
\begin{eqnarray}
\BA_{i,1,1} & = & 
\frac{1}{mp} \| (\I_p \otimes (\ba_i\ba_i^* - \I_n))\bv \|^2 \be_i\be_i^*, \nonumber \\
\BA_{i,2,2} & = &  
\frac{1}{mp} (\I_p\otimes (\ba_i\ba_i^* - \I_n)) \bv\bv^* (\I_p\otimes (\ba_i\ba_i^* - \I_n)) + \frac{1}{m^2} \I_p\otimes (\ba_i\ba_i^* - \I_n)^2 \nonumber \\
& = & \frac{1}{mp} [(\ba_i\ba_i^* - \I_n)\bv_k\bv_l^*(\ba_i\ba_i^* - \I_n) ]_{1\leq k\leq l\leq p} +  \frac{1}{m^2} \I_p\otimes (\ba_i\ba_i^* - \I_n)^2.  \label{eq:Ai22}
\end{eqnarray}
Since $\CZ_{i,2}\CZ_{i,2}^*$ is a positive semi-definite matrix, Lemma~\ref{lem:pos} implies 
\begin{equation}\label{eq:sigma0A}
\sigma_0^2 \leq 2\left\| \sum_{i=1}^m \E
\begin{bmatrix}
\BA_{i,1,1} & \bzero \\
\bzero & \BA_{i,2,2}
\end{bmatrix}
\right\|.
\end{equation}
So we only need to compute $\E(\BA_{i,1,1})$ and $\E(\BA_{i,2,2})$.
\begin{equation*}
\E\| (\I_p \otimes (\ba_i\ba_i^* - \I_n))\bv \|^2 = \sum_{l=1}^p \E\| (\ba_i\ba_i^* - \I_n)\bv_l \|^2 = n\sum_{l=1}^p \|\bv_l\|^2 = np
\end{equation*}
where $\E (\ba_i\ba_i^* - \I_n)^2 = n\I_n.$ Now we have $\E \BA_{i,1,1} = \frac{n}{m}  \be_i\be_i^*.$
For $\E(\BA_{i,2,2}),$ note that
\begin{equation*}
\E((\ba_i\ba_i^* - \I_n)\bv_k\bv_l^*(\ba_i\ba_i^* - \I_n)) = \Tr(\bv_k\bv_l^*) \I_n = \lag \bv_l, \bv_k\rag \I_n = \overline{G}_{k,l}\I_n
\end{equation*}
which follows from~\eqref{eq:gauss-quad}.
By~\eqref{eq:Ai22},~\eqref{eq:sigma0A} and Lemma~\ref{lem:pos}, there holds,
\begin{equation*}
\sigma_0^2 : \leq 2\left\| \sum_{i=1}^m\begin{bmatrix}
\frac{n}{m}\be_i\be_i^* & \bzero \\
\bzero & \frac{1}{mp} \overline{\BG}\otimes \I_n + \frac{n}{m^2}\I_{np}
\end{bmatrix} \right\| 
\leq 2\left\|
\begin{bmatrix}
\frac{n}{m} & \bzero \\
\bzero &  \frac{1}{p}\overline{\BG}\otimes \I_n + \frac{n}{m}\I_{np}
\end{bmatrix}
\right\| \leq 2\left( \frac{\|\BG\|}{p} + \frac{n}{m}\right).
\end{equation*}

One lower bound of $\sigma^2_0$ is
$\sigma_0^2 \geq \left\|\sum_{i=1}^m \E(\BA_{i,1,1})\right\| = \frac{n}{m}$ because each $\E(\CZ_{i,2}\CZ_{i,2}^*)$ is positive semi-definite. Therefore, we have
$\log\left( \frac{\sqrt{m}R}{\sigma_0}\right) \leq \log\left( 2C_0 \frac{\sqrt{n}}{\sqrt{n/m}} \right) \leq \log(2C_0 \sqrt{m})$
where $R : = \max_{1\leq i\leq m}\|\CZ_{i,2}\|_{\psi_1} \leq C_0(\frac{n}{m} + \sqrt{\frac{n}{m}})$ and $m \geq n$.

Applying~\eqref{thm:bern2} to $\sum_{i=1}^m \CZ_{i,2}$ with~\eqref{eq:sigma0A} and~\eqref{eq:Rpsi-1}, we have
\begin{eqnarray}
\left\|\sum_{i=1}^m \CZ_{i,2} \right\| & \leq & C_0\max\Big\{ \left(\frac{n}{m} + \sqrt{\frac{n}{m}}\right)(t + \log(np + m) )\log(np + m), \nonumber \\
&& \sqrt{\left(\frac{\|\BG\|}{p} + \frac{n }{m} \right)(\log(np + m) + t)}  \Big\} \label{eq:czi2}
\end{eqnarray}
with probability $1 - e^{-t}$.
By combining~\eqref{eq:czi2} with~\eqref{eq:czi1} and letting $t = \gamma \log(np + m)$, we have
$\left\|\sum_{i=1}^m \BZ_i\BZ_i^* - \BC\right\| \leq \frac{1}{2}$
with probability $1 - 2m(np + m)^{-\gamma}$ if 
\begin{equation*}
C_0\frac{\sqrt{\gamma\log(np +m)}\|\BG\|_F}{p} \leq \frac{1}{4}, \quad C_0(\gamma + 1)\left(\frac{\|\BG\|}{p} + \frac{n}{m} \right)\log^2(np + m) \leq \frac{1}{16}
\end{equation*}
or equivalently, 
$C_0\left(\max\left\{\frac{\|\BG\|}{p}, \frac{\|\BG\|_F^2}{p^2} \right\} + \frac{n}{m} \right)\log^2(np + m) \leq \frac{1}{16(\gamma + 1)}.$
\end{proof}

\subsection{Proof of the spectral method}\label{ss:svd-bilinear}
We provide the proof of the spectral method proposed in Section~\ref{ss:svd}.  The proof follows two steps: i) we provide an error bound for the recovery under noise by using singular value/vector perturbation. The error bound involves the second smallest singular value of $\MS_0$ and the noise strength $\|\delta \MS\|;$ ii) we give a lower bound for $\sigma_2(\MS_0)$, which is achieved with the help of Proposition~\ref{prop:main1},~\ref{prop:main2} and~\ref{prop:main3} respectively for three different models.


The first result is a variant of perturbation theory for the singular vector corresponding to the smallest singular value. A more general version can be found in~\cite{Wedin72,Stewart90}.
\begin{lemma}\label{lem:svd2}
Suppose $\MS =\MS_0 + \delta\MS$ where $\MS_0$ is the noiseless part of $\MS$ and $\delta\MS$ is the noisy part.
Then
\begin{equation*}
\min_{\alpha_0\in\CC}\frac{\|\alpha_0\hat{\bz} - \bz_0\|}{\|\bz_0\|}  = \left\| \frac{(\I - \hat{\bz}\hat{\bz}^*)\bz_0}{\|\bz_0\|}\right\|\leq \frac{\|\delta \MS\|}{[ \sigma_2(\MS_0) - \|\delta\MS\| ]_+}.
\end{equation*}
where $\sigma_2(\MS_0)$ is the second smallest singular value of $\MS_0$, $\bz_0$ satisfies $\MS_0\bz_0 = \bzero$, and $\hat{\bz}$ is the right singular vector with respect to the smallest singular value of $\MS$, i.e., the solution to~\eqref{prog:svdmin}.
\end{lemma}

\begin{proof}
By definition, there holds $\MS_0\bz_0 = \bzero$
where $\bz_0 = \begin{bmatrix} \bs_0 \\ \bx_0 \end{bmatrix}$ is the ground truth and also the right singular vector corresponding to the smallest singular value. Without loss of generality, we assume $\|\bz_0\| = \|\hat{\bz}\|=1$.
For $\MS$, we denote its singular value decomposition as 
\begin{equation*}
\MS := \underbrace{\MS - {\sigma}_1(\MS) \hat{\bu}\hat{\bz}^*  }_{\BB_1} + \underbrace{{\sigma}_1(\MS) \hat{\bu}\hat{\bz}^*}_{\BB_0}
\end{equation*}
where $\sigma_1(\MS), \hat{\bu}$ and $\hat{\bz}$ are the smallest singular value/vectors of $\MS$. By definition, the vector $\hat{\bz}$ is also the solution to~\eqref{prog:svdmin}.

First note that $\I - \hat{\bz}\hat{\bz}^* =   \BB_1^* (\BB_1^*)^{\dagger}$ where $(\BB_1^*)^{\dagger}$ is the pseudo-inverse of $\BB_1$.
Therefore, we have
\begin{eqnarray*}
\left\|\frac{(\I - \hat{\bz}\hat{\bz}^*)\bz_0}{\|\bz_0\|}\right\|
& = & \|   \bz_0\bz_0^* \BB_1^* (\BB_1^*)^{\dagger} \| =  \| \bz_0\bz_0^* (\MS^*_0+ (\delta \MS)^* - (\BB_0)^*)(\BB_1^*)^{\dagger} \| \\
& = & \|  \bz_0\bz_0^* ( \delta \MS )^* (\BB_1^*)^{\dagger} \| \leq \|\delta \MS\| \| (\BB_1^*)^{\dagger}\| \leq \frac{\|\delta \A\|}{\sigma_2(\MS)} 
\leq \frac{\|\delta \MS\|}{[\sigma_2(\MS_0) - \|\delta\MS\|]_+}
\end{eqnarray*}
where $\MS_0\bz_0 = 0$ and $(\BB_0)^*(\BB_1^*)^{\dagger} = 0$. And the last inequality follows from $|\sigma_2(\MS_0) - \sigma_2(\MS)| \leq \|\delta\MS\|$ and $\sigma_2(\MS) \geq [\sigma_2(\MS_0) - \|\delta \MS\|]_+.$
\end{proof}

The second smallest singular value $\sigma_2(\MS_0)$ is estimated by using the proposition~\ref{prop:main1},~\ref{prop:main2} and~\ref{prop:main3} and the following fact:
\begin{lemma}\label{lem:sm2-svd}
Suppose $\BP$ is an invertible matrix, and $\BA$ is a positive semi-definite matrix with the dimension of its null space equal to 1. Then the second smallest singular value of $\BP\BA\BP^*$ is nonzero and satisfies
\begin{equation*}
\sigma_2 (\BP\BA\BP^*) \geq \sigma_2(\BA) \sigma_1^2(\BP)
\end{equation*}
where $\sigma_1(\cdot)$ and $\sigma_2(\cdot)$ denotes the smallest and second smallest singular values respectively.
\end{lemma}
\begin{proof}
The proof is very straightforward.
Note that
$\| (\BP\BA\BP^*)^{\dagger} \| = \frac{1}{\sigma_2(\BP\BA\BP^*)}$
since $\BA$ is rank-1 deficient.
Also from the property of pseudo inverse, there holds
\begin{equation*}
\| (\BP\BA\BP^*)^{\dagger} \|  = \| (\BP^{-1})^*\BA^{\dagger}\BP^{-1}\| \leq \|\BA^{\dagger}\| \|\BP^{-1}\|^2 = \frac{1}{\sigma_2(\BA)} \frac{1}{\sigma_1^2(\BP)}.
\end{equation*}
Hence, $\sigma_2 (\BP\BA\BP^*) \geq \sigma_2(\BA) \sigma_1^2(\BP).$
\end{proof}


\begin{proof}{ [Proof of Theorem~\ref{thm:svd}] }
Combined with Lemma~\ref{lem:svd2}, it suffices to estimate the lower bound of the second smallest singular value of $\MS_0$ for the proposed three models.
We start with (a). From~\eqref{def:model1-Zl},
we know that
$\MS_0^*\MS_0 = \sum_{l=1}^p \BP\BZ_l\BZ_l^*\BP^* $
where
\begin{equation*}
\BZ_l := 
\begin{bmatrix}
\BLam_l \\
 -\frac{1}{\sqrt{m}}\BA_l^* \\
\end{bmatrix} \in \CC^{ (m + n)\times m},
\quad 
\BP := 
\begin{bmatrix}
\BD^* \|\bx\|& \bzero \\
\bzero & \sqrt{m}\I_n
\end{bmatrix}
\in \CC^{(m + n)\times (m + n)}.
\end{equation*}
From Proposition~\ref{prop:main1}, we know that the second smallest eigenvalue of $\sum_{l=1}^p\BZ_l\BZ_l^*$ is at least $\frac{p}{2}$ and it is also rank-1 deficient. Applying Lemma~\ref{lem:sm2-svd} gives
\begin{equation*}
\sigma_2(\MS_0^*\MS_0) \geq \sigma_2\left(\sum_{l=1}^p\BZ_l\BZ_l^*\right) \sigma_1(\BP\BP^*)  \geq \frac{p}{2} (\min\{ \sqrt{m}, d_{\min}\|\bx\| \})^2
\end{equation*}
and hence
$\sigma_2(\MS_0) \geq \sqrt{\frac{p}{2}} \min\{ \sqrt{m}, d_{\min}\|\bx\| \}$.

\vskip0.5cm
Since (b) and (c) are exactly the same, it suffices to show (b).
From~\eqref{eq:A0-2}, we know that $\MS_0$ can be factorized into 
$\mathcal{A}_0 := \BQ\BZ\BP$
and Proposition~\ref{prop:main2} implies
\begin{equation*}
\left\| \BZ^*\BZ - \BC\right\|   \leq \frac{1}{2}, \quad \BC: = \E(\BZ^*\BZ)= \begin{bmatrix}
\I_m & -\frac{1}{\sqrt{mp}}\bone_m\bv^* \\
-\frac{1}{\sqrt{mp}}\bv\bone_m^* & \I_{np}
\end{bmatrix}.
\end{equation*}
Therefore, $\MS_0^*\MS_0 = \BQ\BZ^*\BP\BP^*\BZ\BQ^*$ and $\sigma_2(\BZ^*\BZ) \geq \frac{1}{2}$.
Applying Lemma~\ref{lem:sm2-svd} leads to
\begin{align*}
\sigma_2(\MS^*_0\MS_0) & \geq \sigma_2(\BZ^*\BP\BP^*\BZ) \sigma_1^2(\BQ) \geq \sigma_2(\BZ^*\BZ) \sigma_1^2(\BP)\sigma_1^2(\BQ) \\
& \geq \frac{1}{2}x_{\min}^2 \min\left\{\sqrt{p}d_{\min}, \frac{\sqrt{m}}{x_{\max}}\right\}^2
\end{align*}
where $x_{\min} = \min\{\|\bx_l\|\}$ and $x_{\max} = \max\{\|\bx_l\|\}$.
\end{proof}

\section*{Appendix}

\begin{lemma}\label{lem:pos}
For any Hermitian positive semi-definite matrix $\BS$ with $
\BS := 
\begin{bmatrix}
\BS_{11} & \BS_{12} \\
\BS_{12}^* & \BS_{22}
\end{bmatrix}
$, there holds,
\begin{equation*}
\begin{bmatrix}
\BS_{11} & \BS_{12} \\
\BS_{12}^* & \BS_{22}
\end{bmatrix}
\preceq 
2
\begin{bmatrix}
\BS_{11} & \bzero \\
\bzero & \BS_{22}
\end{bmatrix}.
\end{equation*}
In other words, $\|\BS\| \leq 2\max\{\|\BS_{11}\|, \|\BS_{22}\|\}$.
\end{lemma}

\begin{lemma}{Corollary 7.21 in~\cite{FR12}}\label{lem:rade}.
Let $\ba\in \CC^M$ and $\beps = (\eps_1, \cdots, \eps_M)$ be a Rademacher sequence, then for $u > 0$, 
\begin{equation*}
\Pr\left( \left| \sum_{j=1}^M \eps_ja_j \right| \geq \|\ba\| u \right) \leq 2\exp(-u^2/2).
\end{equation*}
\end{lemma}

For Gaussian and random Hadamard cases, the concentration inequalities are slightly different. The following theorem is mostly due to Theorem 6.1 in \cite{tropp12user} as well as due to~\cite{Ver12}.

\begin{theorem}\label{thm:bern1}
Consider a finite sequence of $\CZ_l$ of independent centered random matrices with dimension $M_1\times M_2$. We assume that $\|\CZ_l \| \leq R$ and introduce the random matrix $\mathcal{S}: = \sum_{l=1}^L \CZ_l.$
Compute the variance parameter
\begin{equation}\label{sigmasq}
\sigma_0^2 = \max\Big\{ \| \sum_{l=1}^L \E(\CZ_l\CZ_l^*)\|, \| \sum_{l=1}^L \E(\CZ_l^* \CZ_l)\| \Big\},
\end{equation}
then for all $t \geq 0$
\begin{equation}\label{thm:bern}
\|\mathcal{S}\| \leq C_0 \max\{ \sigma_0 \sqrt{t + \log(M_1 + M_2)}, R(t + \log(M_1 + M_2)) \}
\end{equation}
with probability at least $1 - e^{-t}$ where $C_0$ is an absolute constant. 

\end{theorem}
The concentration inequality is slightly different  from Theorem~\ref{thm:bern1} if $\|\CZ_l\|$ is a sub-exponential random variable. Here we are using the form in  \cite{KolVal11}. Before presenting the inequality, we introduce the sub-exponential norm $\|\cdot\|_{\psi_1}$ of a matrix, defined as
\begin{equation}\label{def:psi1}
\|\BZ\|_{\psi_1} := \inf_{u \geq 0} \{ u: \E[ \exp(\|\BZ\|/u)] \leq 2 \}.
\end{equation}
One can find more details about this norm and norm on Orlicz spaces in~\cite{Ver12} and~\cite{VW96}.

\begin{theorem}\label{BernGaussian}
For a finite sequence of independent $M_1\times M_2$ random matrices $\CZ_l$ with $R : = \max_{1\leq l\leq L}
\|\CZ_l\|_{\psi_1}$ and $\sigma_0^2$ as defined in~\eqref{sigmasq}, we have the tail bound on
the operator norm of $\mathcal{S}$, 
\begin{equation}\label{thm:bern2}
\|\mathcal{S}\| \leq C_0 \max\{ \sigma_0 \sqrt{t + \log(M_1 + M_2)}, R\log\left( \frac{\sqrt{L}R}{\sigma_0}\right)(t + \log(M_1 + M_2)) \}
\end{equation}
with probability at least $1 - e^{-t}$ where $C_0$ is an absolute constant. 
\end{theorem}

The estimation of the $\psi_1$-norm of a sub-exponential random variable easily follows from the following lemma.
\begin{lemma}[Lemma 2.2.1 in~\cite{VW96}]\label{lem:psi1}
Let $z$ be a random variable which obeys $\Pr\{ |z| > u \} \leq a e^{-b u }$, then
$\|z\|_{\psi_1} \leq (1 + a)/b.$ 
\end{lemma}
\begin{remark}\label{rem:psi1}
A direct implication of Lemma~\ref{lem:psi1} gives the $\psi_1$-norm of $\|\ba\|^2$ where $\ba\sim \mathcal{N}(0, \frac{1}{2}\I_n) + \mi\mathcal{N}(0, \frac{1}{2}\I_n)$ is a complex Gaussian random vector, i.e., $(\|\ba\|^2)_{\psi_1} \leq C_0n$ for the absolute constant $C_0.$
\end{remark}
\begin{lemma}
For $\ba\sim\mathcal{N}(0, \frac{1}{2}\I_n) + \mi\mathcal{N}(0, \frac{1}{2}\I_n),$ there holds
\begin{eqnarray}
\E( \|\ba\|^2 \ba\ba^*) & = & (n + 1) \I_n, \nonumber \\
\E( (\ba^*\BX\ba)\ba\ba^* ) & = & \BX + \Tr(\BX)\I_n, \label{eq:gauss-quad} 
\end{eqnarray}
for any fixed $\bx\in\CC^n$ and $\BX\in\CC^{n\times n}$. In particular, we have
\begin{eqnarray} 
\E |\lag \ba, \bx\rag|^2 \ba\ba^* & = & \|\bx\|^2\I_n + \bx\bx^*, \label{eq:gauss-3}\\
\E |\lag \ba, \bx\rag|^4 & = & 2\|\bx\|^4, \label{eq:gauss-4}\\
\E(\ba\ba^* - \I_n)^2 & = & n\I_n. \label{eq:gauss-5}
\end{eqnarray}
\end{lemma}


\begin{lemma}
Suppose that $\ba\in \RR^n$ is a Rademacher sequence and for any fixed $\bx\in\CC^n$ and $\BX\in\CC^{n\times n}$, there holds
\begin{eqnarray}
\E( \|\ba\|^2 \ba\ba^*) & = & n\I_n, \label{eq:binary-1} \\
\E( (\ba^*\BX\ba)\ba\ba^* ) & = & \Tr(\BX) \I_n + \BX + \BX^T - 2\sum_{k=1}^n X_{kk}\BE_{kk} \label{eq:binary-2}
\end{eqnarray}
where $\BE_{kk}$ is an $n\times n$ matrix with only one nonzero entry equal to 1 and at position $(k,k).$
In particular, setting $\BX = \bx\bx^*$ gives 
\begin{eqnarray}
\E |\lag \ba, \bx\rag|^2 \ba\ba^* & = & \|\bx\|^2\I_n + \bx\bx^* + \overline{\bx\bx^*} - 2\diag(\bx)\diag(\bar{\bx}) \preceq 3\|\bx\|^2 \I_n, \label{eq:binary-3}\\
\E |\lag \ba, \bx\rag|^4  & = & 2\|\bx\|^4 + \left|\sum_{k=1}^n x_i^2 \right|^2 - 2\sum_{k=1}^n |x_k|^4 \leq 3\|\bx\|^4. \label{eq:binary-4}
\end{eqnarray}
\end{lemma}

\begin{proof}
Since $\ba$ is a Rademacher sequence, i.e, each $a_i$ takes $\pm 1$ independently with equal probability, this implies $a_i^2 = 1$ and $\|\ba\|^2 = n.$ Therefore, $\E(\|\ba\|^2 \ba\ba^*) = n\E(\ba\ba^*) = n\I_n.$ 
The $(k,l)$-th entry of $(\ba^*\BX\ba)\ba\ba^*$ is 
$\sum_{i=1}^n\sum_{j=1}^n X_{ij}a_i a_j a_k a_l$.
\begin{enumerate}
\item If $k = l$, 
\begin{equation*}
\E\left(\sum_{i=1}^n\sum_{j=1}^n X_{ij}a_i a_j |a_k|^2\right) = \sum_{i=1}^n \E(X_{ii}|a_i|^2 |a_k|^2) = \Tr(\BX)
\end{equation*}
where $\E(a_ia_j |a_k|^2) = 0$ if $i\neq j.$
\item If $k\neq l$, 
\begin{equation*}
\E\left(\sum_{i=1}^n\sum_{j=1}^n X_{ij}a_i a_j a_ka_l\right) = X_{kl}\E(|a_k|^2 |a_l|^2) + X_{lk} \E(|a_k|^2|a_l|^2)= X_{kl} + X_{lk}.  
\end{equation*}
\end{enumerate}
Hence, we have 
$\E((\ba^*\BX\ba) \ba\ba^*) = \Tr(\BX) \I_n + \BX + \BX^T - 2\sum_{k=1}^n X_{kk}\BE_{kk}.$
\end{proof}

\begin{lemma}\label{lem:LAAL}
There holds
\begin{equation*}
\E(\BLam\BA(\BLam\BA)^*) = 
\frac{(n - 1)(m\I_m - \bone_m\bone_m^*)}{m-1} + \bone_m\bone_m^*
\end{equation*}
where $\BA = \BH\BM$, $\BLam = \diag(\overline{\BA\bv})$ and $\bv = (v_1, \cdots, v_n)^T\in\CC^n$ is a deterministic unit vector. $\BH\in\CC^{m\times n}$ is a random partial Fourier/Hadamard matrix with $\BH^*\BH = m\I_n$ and $m\geq n$, i.e., the columns of $\BH$ are uniformly sampled without replacement from an $m\times m$ DFT/Hadamard matrix; $\BM$ is a diagonal matrix with entries random sampled from $\pm 1$ with equal probability; moreover, we assume $\BM$ and $\BH$ are independent from each other. In particular, if $m=n=1$, $\E(\BLam\BA(\BLam\BA)^*) = 1.$
\end{lemma}
\begin{proof}
We only prove the case when $\BA$ is a random Fourier matrix since the Hadamard case is essentially the same modulo very minor differences. 
By definition, 
\begin{equation*}
\BLam\BA\BA^* \BLam^* = \BLam\BH\BH^* \BLam^* =  \diag(\overline{\BH\BM\bv}) \BH\BH^*  \diag(\BH\BM\bv).
\end{equation*}

Let $\bh_i$ be the $i$-th column of $\BH^*$ and the $(i,j)$-th entry of $\diag(\overline{\BH\BM\bv}) \BH\BH^* \diag(\BH\BM\bv)$ is
$\overline{\lag \bh_{i}, \BM\bv\rag} \lag \bh_{j}, \BM\bv\rag \lag \bh_{i}, \bh_{j}\rag$
where $\lag \bu, \bv\rag = \bu^*\bv $. The randomness of $\BLam\BA\BA^* \BLam^*$ comes from both $\BH$ and $\BM$ and we first take the expectation with respect to $\BM$.
\begin{eqnarray*}
\E(\overline{\lag \bh_{i}, \BM\bv\rag} \lag \bh_{j}, \BM\bv\rag |\BH) 
& = & \overline{\bh_{i}^*\E( \BM\bv \bv^*\BM|\BH) \bh_{j}} =  \overline{\bh_{i}^* \diag(\bv)\diag(\bar{\bv}) \bh_{j}} \\
& = & (\overline{\BH \diag(\bv)\diag(\bar{\bv})\BH^* })_{i,j}
\end{eqnarray*}
where $\E(\BM\bv\bv^*\BM) = \diag(\bv)\diag(\bar{\bv})$ follows from each entry in $\BM$ being a Bernoulli random variable. Hence, $\E((\BLam\BA\BA^* \BLam^*)_{i,j}|\BH) = (\overline{\BH \diag(\bv)\diag(\bar{\bv})\BH^* })_{i,j} \cdot (\BH\BH^*)_{i,j}.$

Let $\bu_k$ be the $k$-th column of $\BH$ and $1\leq k\leq n$ and ``$\odot$" denotes the Hadamard (pointwise) product. So we have $\BH\BH^* = \sum_{k=1}^n \bu_k\bu_k^*.$ There holds,
\begin{eqnarray}
\E(\BLam\BA\BA^* \BLam^*|\BH)  & = &  (\overline{\BH \diag(\bv)\diag(\bar{\bv})\BH^* })\odot \BH\BH^* = \left( \sum_{k=1}^n |v_k|^2 \bar{\bu}_k\bar{\bu}_k^* \right)\odot \BH\BH^* \nonumber \\
& = & \sum_{k=1}^n |v_k|^2\diag(\bar{\bu}_k)\BH\BH^*\diag(\bu_k) = \sum_{k=1}^n \sum_{l=1}^n |v_k|^2 \diag(\bar{\bu}_k)\bu_l\bu_l^*\diag(\bu_k) \nonumber \\
& = &  \sum_{1\leq k\neq l \leq n} |v_k|^2 \diag(\bar{\bu}_k)\bu_l\bu_l^*\diag(\bu_k) + \bone_m\bone_m^* \label{eq:LAAL-H}
 \end{eqnarray}
where the third equation follows from linearity of the Hadamard product and from
\begin{equation*}
\bar{\bu}_k\bar{\bu}_k^*\odot \BH\BH^* = \diag(\bar{\bu}_k)\BH\BH^*\diag(\bu_k).
\end{equation*} 
The last one uses the fact that $\diag(\bar{\bu}_k)\bu_k = \bone_m$ if $\bu_k$ is a vector from the DFT matrix or Hadamard matrix. 
By the property of conditional expectation, we know that $\E(\BLam \BA\BA^*\BLam^*)  =\E(\E(\BLam \BA\BA^*\BLam^*)|\BH)$ and due to the linearity of expectation, it suffices to find out for $k\neq l$,
$\E(\diag(\bar{\bu}_k)\bu_l\bu_l^*\diag(\bu_k))$
where $\bu_k$ and $\bu_l$, by definition, are the $k$-th and $l$-th columns of $\BH$ which are sampled uniformly without replacement from an $m\times m$ DFT matrix $\BF$. Note that $(\bu_k, \bu_l)$ is actually an \emph{ordered} pair of random vectors sampled without replacement from columns of $\BF$. Hence there are in total $m(m-1)$ different choices of $\diag(\bar{\bu}_k)\bu_l$ and
\begin{equation*}
\Pr(\bu_k = \bphi_i, \bu_l = \bphi_j) = \frac{1}{m(m - 1)}, \quad \forall 1\leq i\neq j\leq m, \forall 1\leq k\neq l\leq n
\end{equation*}
where $\bphi_i$ is defined as the $i$-th column of an $m\times m$ DFT matrix $\BF$. Now we have, for any $k\neq l$, 
\begin{eqnarray}
\E(\diag(\bar{\bu}_k)\bu_l\bu_l^*\diag(\bu_k))
& = & \frac{1}{m(m - 1)}\sum_{i \neq j} \diag(\bar{\bphi}_i)\bphi_j\bphi_j^*\diag(\bphi_i) \nonumber \\
& = &  \frac{1}{m(m - 1)}\left( \sum_{1\leq i, j\leq m}\diag(\bar{\bphi}_i)\bphi_j\bphi_j^*\diag(\bphi_i) - m\bone_m\bone_m^* \right) \nonumber \\
& = & \frac{1}{m-1}\left(\sum_{i=1}^m \diag(\bar{\bphi}_i) \diag(\bphi_i) - \bone_m\bone_m^* \right) \nonumber \\
& = & \frac{m\I_m - \bone_m\bone_m^*}{m - 1} \label{eq:ukul}.
\end{eqnarray}
where $\diag(\bar{\bphi}_i)\bphi_i = \bone_m$ and $\sum_{i=1}^m \bphi_i\bphi_i^* = m\I_m.$ Now we return to $\E(\BLam\BA\BA^*\BLam^*)$. By substituting~\eqref{eq:ukul} into~\eqref{eq:LAAL-H}, we get the desired formula:
\begin{eqnarray*}
\E(\BLam\BA\BA^* \BLam^*) & = & \E(\E(\BLam\BA\BA^* \BLam^*|\BH)  ) =  \E\left(\sum_{1\leq k\neq l\leq n} |v_k|^2 \diag(\bu_k)\bu_l\bu_l^*\diag(\bu_k) \right) + \bone_m\bone_m^* \\
& = & \frac{m\I_m - \bone_m\bone_m^*}{m-1} \sum_{1\leq k\neq l\leq n}|v_k|^2 + \bone_m\bone_m^* =\frac{(n - 1)(m\I_m - \bone_m\bone_m^*)}{m-1} + \bone_m\bone_m^*,
\end{eqnarray*}
where $\sum_{1\leq k\neq l\leq n} |v_k|^2  = n - \sum_{k=1}^n |v_k|^2 = n - 1$ follows from $\sum_{k=1}^m |v_k|^2 = 1.$
\end{proof}


\end{document}